\documentclass{article}

\usepackage{algorithm}
\usepackage[noend]{algpseudocode}
\usepackage{amsmath,amssymb,amsthm}
\usepackage{bm}
\usepackage{bbm}
\usepackage{bbold}
\usepackage{booktabs}
\usepackage{caption}
\usepackage{enumerate}
\usepackage{float}
\usepackage[colorlinks,citecolor=blue,linkcolor=BrickRed]{hyperref}
\usepackage{cleveref}
\usepackage[utf8]{inputenc}
\usepackage{parskip}
\usepackage{subcaption}
\usepackage[normalem]{ulem}
\usepackage[usenames,dvipsnames]{xcolor}
\usepackage{xspace}
\usepackage[numbers,sort&compress]{natbib}

\newtheorem{theorem}{Theorem}

\newtheorem{lemma}[theorem]{Lemma}
\newtheorem{corollary}[theorem]{Corollary}
\newtheorem{definition}[theorem]{Definition}
\newenvironment{proofof}[1]{\begin{proof}[Proof of #1]}{\end{proof}}

\newcounter{note}[section]

\DeclareMathOperator*{\argmin}{arg\,min}

\newcommand{\E}{\mathop{{}\mathbb{E}}}

\newcommand{\OPT}{\mathrm{OPT}}
\newcommand{\ALG}{\mathrm{ALG}}
\newcommand{\LP}{\mathrm{LP}}
\newcommand{\poly}{\text{poly}}
\newcommand{\st}{\text{\xspace s.t. \xspace}}

\newcommand{\calC}{\mathcal{C}}
\newcommand{\calP}{\mathcal{P}}
\newcommand{\mcD}{\mathcal{D}}

\newcommand{\mcG}{\mathcal{G}}
\newcommand{\mcH}{\mathcal{H}}

\newcommand{\mcP}{\mathcal{P}}

\newcommand{\calS}{\mathcal{S}}

\newcommand{\calT}{\mathcal{T}}

\allowdisplaybreaks

\usepackage{fullpage}
\usepackage{thm-restate}
\usepackage{appendix}
\usepackage{MnSymbol}
\usepackage{graphicx}

\DeclareMathOperator{\maxflow}{maxflow}

\title{Deterministic Tree Embeddings with Copies for\\ Algorithms Against Adaptive Adversaries\footnote{Supported in part by NSF grants CCF-1527110, CCF-1618280, CCF-1814603, CCF-1910588, NSF CAREER award CCF-1750808, a Sloan Research Fellowship, funding from the European Research Council (ERC) under the European Union's Horizon 2020 research and innovation program (ERC grant agreement 949272), Swiss National Foundation (project grant 200021-184735) and the Air Force Office of Scientific Research under award number FA9550-20-1-0080.}}

\author{\begin{tabular}[t]{c@{\extracolsep{1.5em}}ccc} 
        Bernhard Haeupler\footnote{\texttt{haeupler@cs.cmu.edu}}  & D Ellis Hershkowitz\footnote{\texttt{dhershko@cs.cmu.edu}} & Goran Zuzic\footnote{\texttt{goran.zuzic@inf.ethz.ch}} \\
        \small Carnegie Mellon University   & \small Carnegie Mellon University & \small ETH Z\"urich \\
        \small  \& ETH Z\"urich &   & \\
\end{tabular}}
\date{}

\begin{document}
        
    \clearpage\maketitle
    \thispagestyle{empty}
        \begin{abstract}
                Embeddings of graphs into distributions of trees that preserve distances in expectation are a cornerstone of many optimization algorithms. Unfortunately, online or dynamic algorithms which use these embeddings seem inherently randomized and ill-suited against adaptive adversaries.
								
                
\smallskip

              In this paper we provide a new tree embedding which addresses these issues by \emph{deterministically} embedding a graph into a single tree containing $O(\log n)$ copies of each vertex while preserving the connectivity structure of every subgraph and $O(\log^2 n)$-approximating the cost of every subgraph.
														

\smallskip

Using this embedding we obtain several new algorithmic results:  We reduce an open question of \citet{alon2006general}---the existence of a deterministic poly-log-competitive algorithm for online group Steiner tree on a general graph---to its tree case. We give a poly-log-competitive deterministic algorithm for a closely related problem---online partial group Steiner tree---which, roughly, is a bicriteria version of online group Steiner tree. Lastly, we give the first poly-log approximations for demand-robust Steiner forest, group Steiner tree and group Steiner forest.
\end{abstract}

 \newpage
{
    \hypersetup{linkcolor=Blue}
    \tableofcontents
}\thispagestyle{empty}

\newpage \setcounter{page}{1} 

\section{Introduction} 
  
Probabilistic embedding of general metrics into distributions over trees are one of the most versatile tools in combinatorial and network optimization. The beauty and utility of these tree embeddings comes from the fact that their application is often simple, yet extremely powerful. Indeed, when modeling a network with length, costs, or capacities as a weighted graph, these embeddings often allow one to pretend that the graph is a tree. A common template for countless network design algorithms is to (1) embed the input weighted graph $G$ into a randomly sampled tree $T$ that approximately preserves the weight structure of $G$; (2) solve the input problem on $T$
and; (3) project the solution on $T$ back into $G$.

        
        A long and celebrated line of work \cite{karp19892k,alon1995graph,bartal1996probabilistic,fakcharoenphol2004tight} culminated in the embedding of Fakcharoenphol, Rao and Talwar \cite{fakcharoenphol2004tight}---henceforth the ``FRT embedding''---which showed that any weighted graph on $n$ nodes can be embedded into a distribution over weighted trees in a way that $O(\log n)$-approximately preserves distances in expectation. Together with the above template this reduces many graph problems to much easier problems on trees at the cost of an $O(\log n)$ approximation factor. This has lead to a myriad of approximation, online, and dynamic algorithms with poly-logarithmic approximations and competitive ratios for NP-hard problems such as for $k$-server \cite{bansal2011polylogarithmic}, metrical task systems \cite{bartal1997polylog},  group Steiner tree and group Steiner forest \cite{alon2006general,naor2011online, garg2000polylogarithmic}, buy-at-bulk network design \cite{awerbuch1997buy} and (oblivious) routing \cite{racke2002minimizing}. For many of these problems tree embeddings are the only known way of obtaining such algorithms on general graphs.
 
        
However, probabilistic tree embeddings have one drawback: Algorithms based on them naturally require randomization and their approximation guarantees only hold in expectation. For approximation algorithms---i.e., in the offline setting---there 
 are derandomization tools, such as the FRT derandomizations given in \cite{charikar1998approximating,fakcharoenphol2004tight}, to overcome these issues. These derandomization results are so general that essentially any offline algorithm based on tree embeddings can be transformed into a deterministic algorithm with matching approximation guarantees (with only a moderate increase in running time). Unfortunately, these strategies are not applicable to online or dynamic settings where an adversary progressively reveals the input. Indeed, to our knowledge, all online and dynamic algorithms that use FRT are randomized (e.g.\  \cite{guo2020facility,gupta2019permutation,alon2006general,fiat2003better,bartal1997polylog,naor2011online,englert2017reordering,englert2007reordering}).

This overwhelming evidence in the literature is driven by a well-known and fundamental barrier to the use of probabilistic tree embeddings in deterministic online and dynamic algorithms. More specifically and even worse, this is a barrier which prevents these algorithms from working against all but the weakest type of adversary. In particular, designing an online or dynamic algorithm which is robust to an oblivious adversary (which fixes all requests in advance, independently of the algorithm's randomness) is often much easier than designing an algorithm which is robust to an adaptive adversary (which chooses the next request based on the algorithm's current solution). As the actions of a deterministic algorithm can be fully predicted this distinction only holds for randomized algorithms---any deterministic algorithm has to always work against an adaptive adversary. For these reasons, many online and dynamic algorithms have exponentially worse competitive ratios in the deterministic or adaptive adversary setting than in the oblivious adversary setting. This is independent of computational complexity considerations. 

The above barrier results from a repeatedly recognized and seemingly unavoidable phenomenon which prevents online algorithms built on FRT from working against adaptive adversaries. Specifically, there are graphs where every tree embedding must have many node pairs with polynomially-stretched distances \cite{bartal1996probabilistic}. There is nothing that prevents an adversary then from learning through the online algorithm's responses which tree was sampled and then tailoring the remainder of the online instance to pairs of nodes that have highly stretched distances. The exact same phenomenon occurs in the dynamic setting; see, for example, \citet{guo2020facility} and \citet{gupta2019permutation} for dynamic algorithms with expected cost guarantees that only hold against oblivious adversaries because they are based on FRT. In summary, online and dynamic algorithms that use probabilistic tree embeddings seem inherently randomized and seem to necessarily only work against adversaries oblivious to this randomness.   

Similar, albeit not identical,\footnote{We remark that, unlike the online and dynamic setting, the barrier to obtaining demand-robust algorithms which work against the ``adaptive adversary'' implicit in the setting is merely computational and thus seems potentially less inherent.} issues also arise in other settings, most notably demand-robust optimization. The demand-robust model is a well-studied model of optimization under uncertainty \cite{dhamdhere2005pay,hershkowitz2018prepare,feige2007robust,gupta2015robust,gupta2010thresholded,golovin2006pay} in which an algorithm first buys a partial solution given a large collection of potential problem instances. An ``adaptive adversary'' then chooses which of the potential instances must be solved and the algorithm must extend its partial solution to solve the selected instance at inflated costs. The adversary is adaptive in the sense that it chooses the final instance with full knowledge of the algorithm's partial solution. To thwart an algorithm which reduces a demand-robust problem to its tree version via a sampled FRT tree, the adversary can present a collection of potential instances which for every tree $T$ in the FRT distribution contains an instance for which $T$ is an arbitrarily bad approximation and then always choose the worst-case problem instance. 
The fact that there do not exist any demand-robust algorithms which use FRT despite this setting having received considerable attention seems at least partially due to the issues pointed out here.

Overall it seems fair to say that prior to this work tree embeddings seemed fundamentally incapable of enabling adaptive-adversary-robust and deterministic algorithms in several well-studied settings.

\subsection{Our Contributions}

We provide a conceptually new type of metric embedding---the copy tree embedding--- which is deterministic and therefore also adaptive-adversary-robust. 
Specifically, we show that any weighted graph $G$ can be deterministically embedded into a single weighted tree with a small number of copies for each vertex. Any subgraph of $G$ will project onto this tree in a connectivity and approximate-cost preserving way. 

To precisely define our embeddings we define a copy mapping $\phi$ which maps a vertex $v$ to its copies.


\begin{definition}[Copy Mapping]
    Given vertex sets $V$ and $V'$ we say $\phi : V \to 2^{V'}$ is a copy mapping if every node has at least one copy (i.e.\ $|\phi(v)| \geq 1$ for all $v \in V$), copies are disjoint (i.e.\ $\phi(v) \cap \phi(u) = \emptyset$ for $u \neq v$) and every node in $V'$ is a copy of some node (i.e. for every $v' \in V'$ there is some $v \in V$ where $v' \in \phi(v)$). For $v' \in V'$, we use the shorthand $\phi^{-1}(v')$ to stand for the unique $v \in V$ such that $v' \in \phi(v)$.
\end{definition}

A copy tree embedding for a weighted graph $G$ now simply consists of a tree $T$ on copies of vertices of $G$ with one distinguished root and two mappings $\pi_{G \to T}$ and $\pi_{T \to G}$ which map subsets of edges from $G$ to $T$ and from $T$ to $G$ in a way that preserves connectivity and approximately preserves costs. We say that \emph{two vertex subsets $U, W$ are connected} in a graph if there is a $u \in U$ and $w \in W$ such that $u$ and $w$ are connected. We also say that a mapping $\pi : 2^E \to 2^{E'}$ is \emph{monotone} if for every $A \subseteq B$ we have that $\pi(A) \subseteq \pi(B)$. A rooted tree $T = (V, E, w)$ is \emph{well-separated} if for all edges $e$ if $e'$ is a child edge of $e$ in $T$ then $w(e') \leq \frac{1}{2}w(e)$.

\begin{definition}[$\alpha$-Approximate Copy Tree Embedding with Copy Number $\chi$]\label{dfn:repTree}
    Let $G = (V, E, w)$ be a weighted graph with some distinguished root $r \in V$. An $\alpha$-approximate copy tree embedding with copy number $\chi$ consists of a weighted rooted tree $T = (V', E',w')$, a copy mapping $\phi : V \to 2^{V'}$ and edge mapping functions $\pi_{G \to T} : 2^E \to 2^{E'}$ and $\pi_{T \to G} : 2^{E'} \to 2^{E}$ where $\pi_{T \to G} : 2^{E'} \to 2^{E}$ is monotone and:
    \begin{enumerate}
        \item \textbf{Connectivity Preservation:} For all $F \subseteq E$ and $u,v \in V$ if $u, v$ are connected by $F$, then $\phi(u), \phi(v) \subseteq V'$ are connected by $\pi_{G \to T}(F)$.  Symmetrically, for all $F' \subseteq E'$ and $u', v' \in V'$ if $u'$ and $v'$ are connected by $F'$ then $\phi^{-1}(u')$ and $\phi^{-1}(v')$ are connected by $\pi_{T \to G}(F')$.
        \item \textbf{$\alpha$-Cost Preservation}: For any $F \subseteq E$ we have $w(F) \leq \alpha \cdot w'(\pi_{G \to T}(F))$ and for any $F' \subseteq E'$ we have $w'(F') \leq w(\pi_{T \to G}(F'))$.
        \item \textbf{Copy Number:} $|\phi(v)| \leq \chi$ for all $v \in V$ and $\phi(r) = \{r'\}$ where $r'$ is the root of $T$.
    \end{enumerate} 
    A copy tree embedding is efficient if $T$, $\phi$, and $\pi_{T \to G}$ are deterministically poly-time computable and well-separated if $T$ is well-separated.
\end{definition}
We emphasize that, whereas standard tree embeddings guarantee costs are preserved in expectation, our copy tree embeddings preserve costs deterministically. Also notice that for efficient copy tree embeddings we do not require that $\pi_{G \to T}$ is efficiently computable; this is because $\pi_{G \to T}$ will be used in our analyses but not in any of our algorithms.

We first give two copy tree embedding constructions which trade off between the number of copies and cost preservation. Both constructions are based on the idea of merging appropriately chosen tree embeddings as pictured in \Cref{fig:constrPart} and \Cref{fig:constrFRT} where we color nodes according to the node whose copy they are.

\begin{figure}
    \centering
    \begin{subfigure}[b]{0.32\textwidth}
        \centering
        \includegraphics[width=\textwidth,trim=0mm 100mm 0mm 60mm, clip]{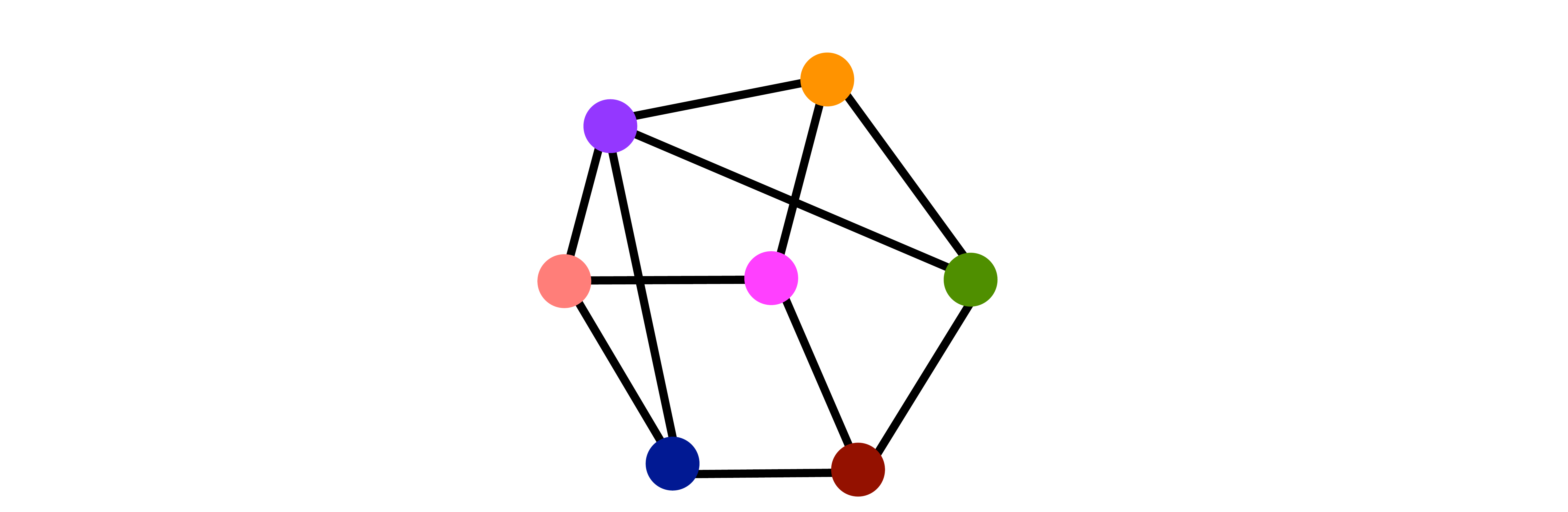}
        \caption{Graph $G$.}
    \end{subfigure}
    \hfill
    \begin{subfigure}[b]{0.32\textwidth}
        \centering
        \includegraphics[width=\textwidth,trim=0mm 150mm 0mm 60mm, clip]{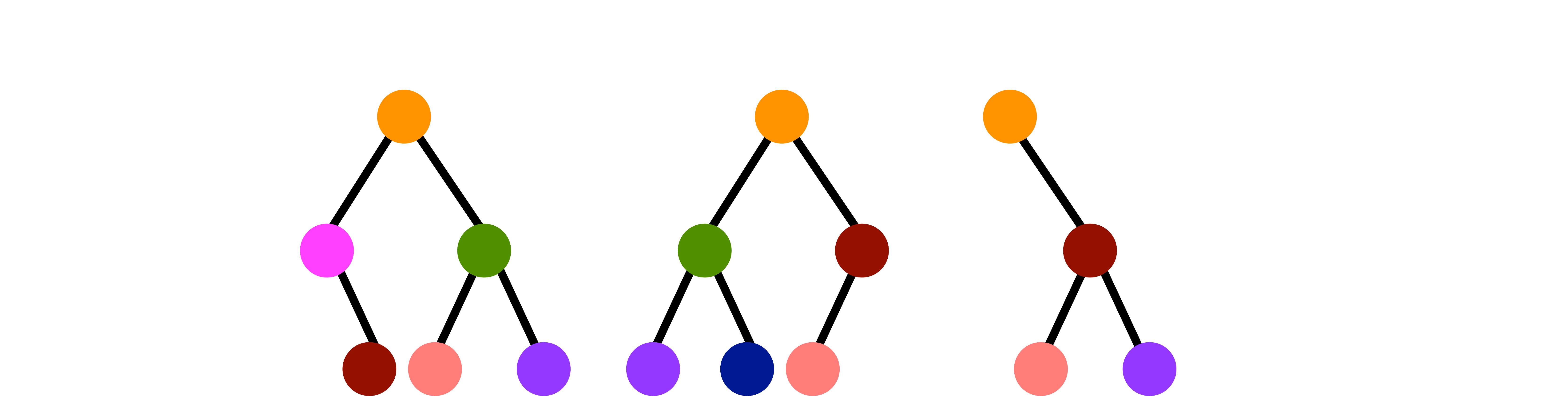}
        \caption{Compute partial tree embeddings.}
    \end{subfigure}
    \hfill
    \begin{subfigure}[b]{0.32\textwidth}
        \centering
        \includegraphics[width=\textwidth,trim=0mm 150mm 0mm 60mm, clip]{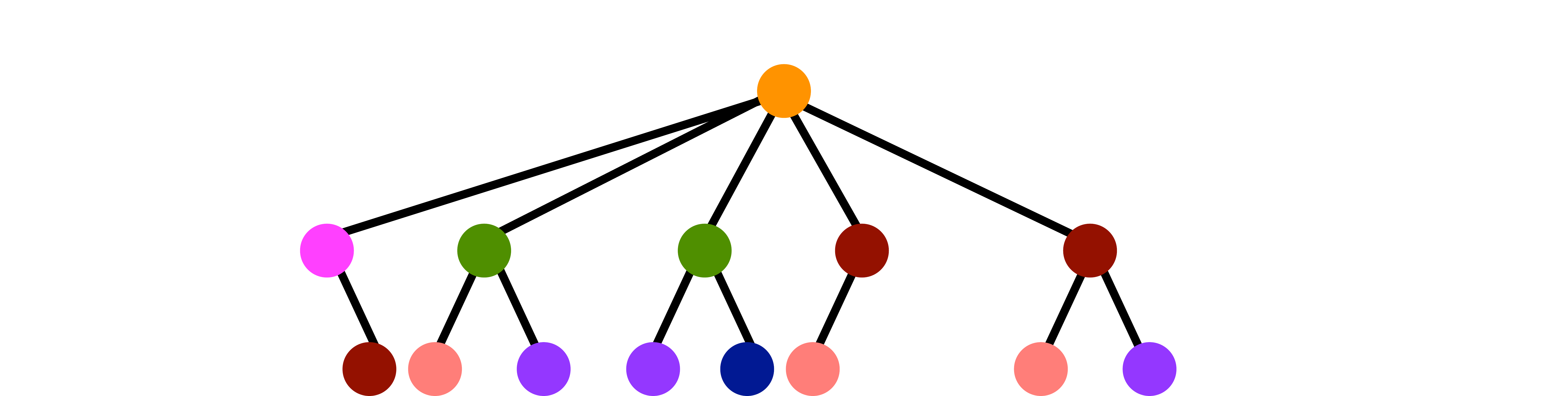}
        \caption{Merge trees.}
    \end{subfigure}
    \hfill
    \caption{Illustration of our first construction where we merge $O(\log n)$ partial tree embeddings.}\label{fig:constrPart}
\end{figure}

\begin{figure}
    \centering
    \begin{subfigure}[b]{0.32\textwidth}
        \centering
        \includegraphics[width=\textwidth,trim=0mm 100mm 0mm 60mm, clip]{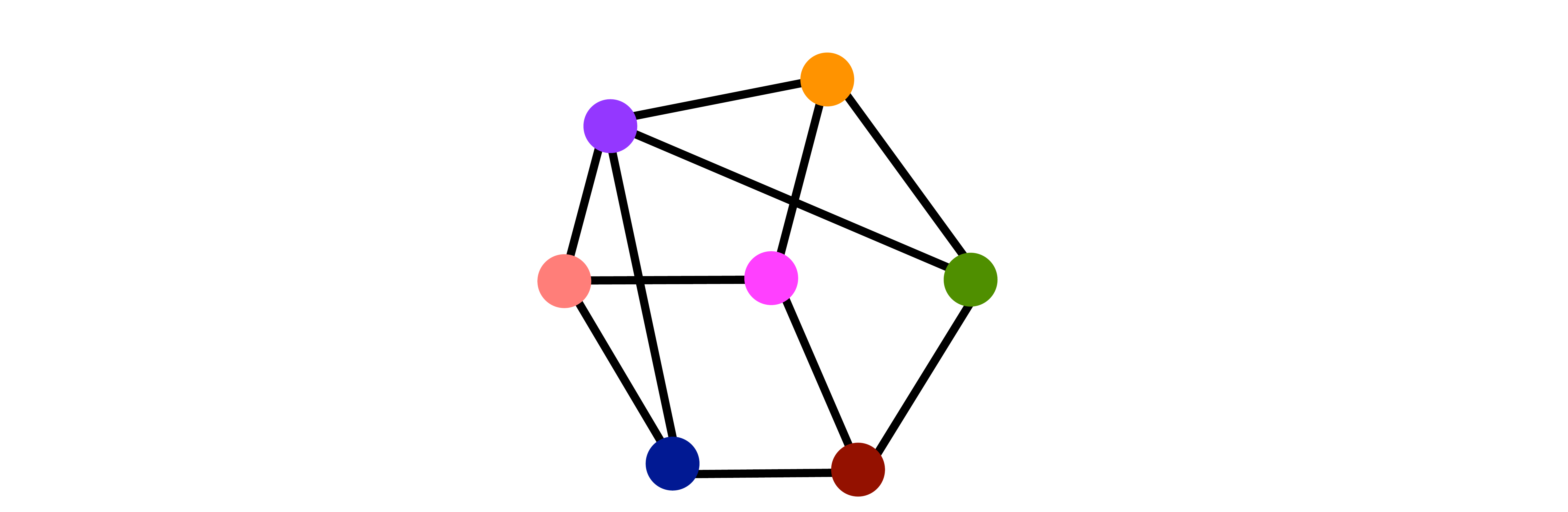}
        \caption{Graph $G$.}
    \end{subfigure}
    \hfill
    \begin{subfigure}[b]{0.32\textwidth}
        \centering
        \includegraphics[width=\textwidth,trim=0mm 150mm 0mm 60mm, clip]{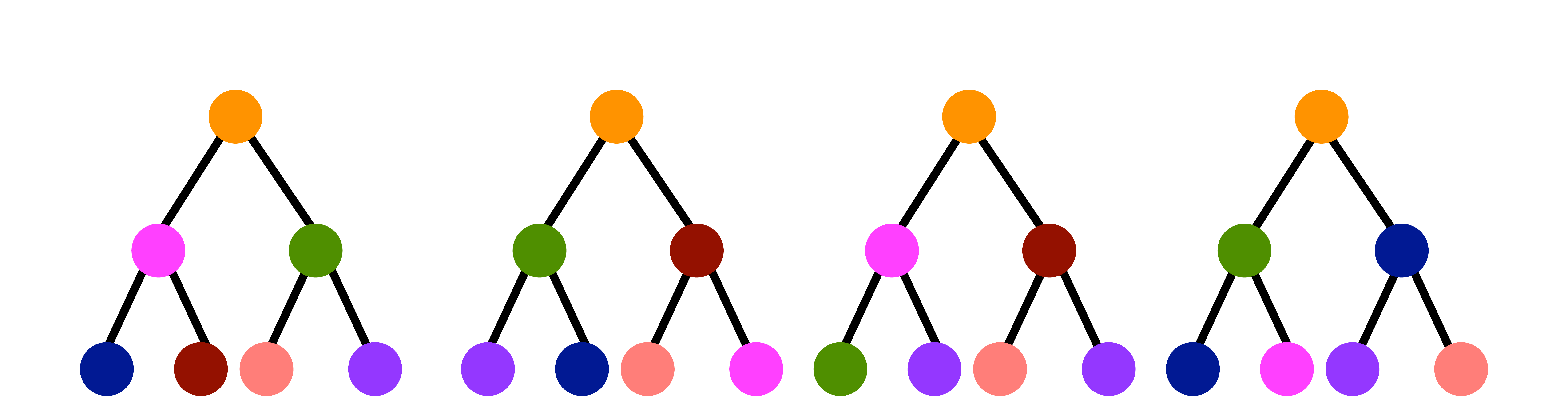}
        \caption{Enumerate FRT support.}
    \end{subfigure}
    \hfill
    \begin{subfigure}[b]{0.32\textwidth}
        \centering
        \includegraphics[width=\textwidth,trim=0mm 150mm 0mm 60mm, clip]{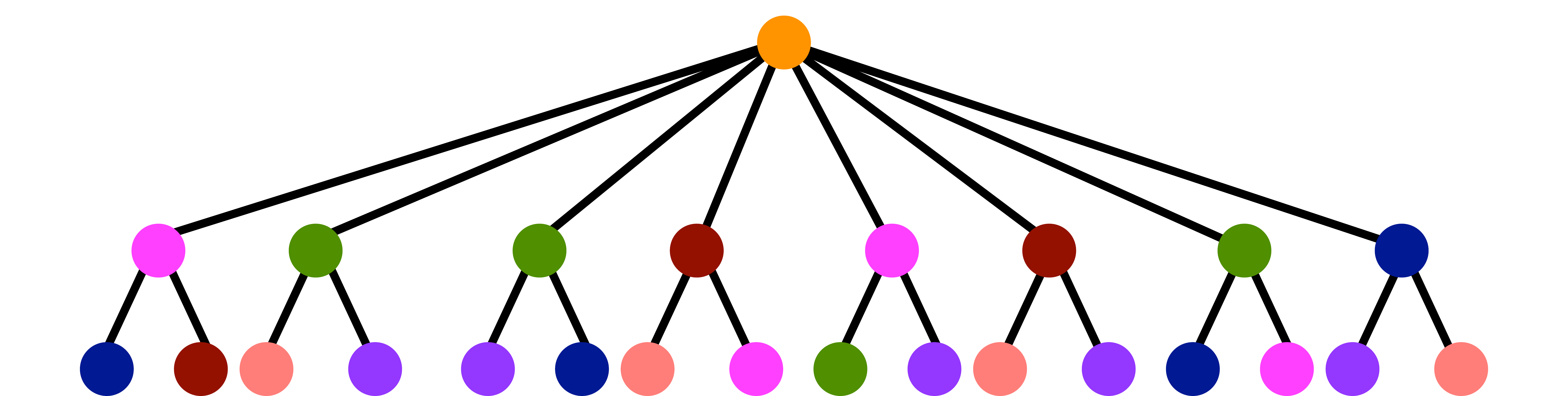}
        \caption{Merge trees.}
    \end{subfigure}
    \hfill
    \caption{Illustration of our second construction where we merge the $O(n \log n)$ trees in the FRT support.}\label{fig:constrFRT}
\end{figure}

\textbf{Construction 1: Merging Partial Tree Embeddings (\Cref{sec:partTree})}. 
The cornerstone of our first construction is the idea of merging embeddings which give good \emph{deterministic} distance preservation. If our goal is to embed the entire input metric into a tree this is impossible. However, it is possible to embed a random constant fraction of nodes in an input metric into a tree in a way that deterministically preserves distances of the embedded nodes; an embedding which we call a ``partial tree embedding'' (see also \citet{gupta2006oblivious,haeupler2020tree}). We then use the method of conditional expectation to derandomize a node-weighted version of this random process and apply this derandomization $O(\log n)$ times, down-weighting nodes as they are embedded. The result of this process is $O(\log n)$ partial tree embeddings where a multiplicative-weights-type argument shows that each node appears in a constant fraction of these embeddings. Merging these $O(\log n)$ embeddings gives our copy tree while an Euler-tour-type proof shows that subgraphs of the input graph can be mapped to our copy tree in a cost and connectivity-preserving fashion. The following theorem summarizes our first construction.
\begin{restatable}{theorem}{repTree}\label{thm:repTreeConst}
    There is a poly-time deterministic algorithm which given any weighted graph $G = (V, E, w)$ and root $r \in V$ computes an efficient and well-separated $O(\log^2n)$-approximate copy tree embedding with copy number $O(\log n)$.
\end{restatable}

\textbf{Construction 2: Merging FRT Support (\Cref{sec:FRTSup}).}
Our second construction follows from a known fact that the size of the support of the FRT distribution can be made $O(n \log n)$ and this support can be computed deterministically in poly-time \cite{charikar1998approximating}. Merging each tree in this support at the root and some simple probabilistic method arguments give a copy tree embedding that is $O(\log n)$-cost preserving but with an $O(n \log n)$ copy number. The next theorem summarizes this construction.
\begin{restatable}{theorem}{frtSupp}\label{thm:frtSupp}
    There is a poly-time deterministic algorithm which given any weighted graph $G = (V, E, w)$ and root $r \in V$ computes an efficient and well-separated $O(\log n)$-approximate copy tree embedding with copy number $O(n \log n)$.
\end{restatable}

While our second construction achieves a slightly better cost bound than our first construction, it has the significant downside of a linear copy number. Notably, this linear copy number makes our second construction unsuitable for some applications, including, for example, our second application as described below. Moreover, our first construction also has several desirable properties which our second does not which we expect might be useful for future applications. These include: (1) $\pi_{G \to T}$ is monotone (in addition to $\pi_{T \to G}$ being monotone as stipulated by \Cref{dfn:repTree}); (2) if $u$ and $v$ are connected by $F \subseteq E$ then $\Omega(\log n)$ vertices of $\phi(u)$ are connected to $\Omega(\log n)$ vertices of $\phi(v)$ in $\pi_{G \to T}(F)$ (as opposed to just one vertex of $\phi(u)$ and one vertex of $\phi(v)$ as in \Cref{dfn:repTree}) and; (3) if $u$ is connected to $r$ by $F \subseteq E$ then every vertex in $\phi(u)$ is connected to $\phi(r)$ in $\pi_{G \to T}(F)$ (as opposed to just one vertex of $\phi(u)$ as in \Cref{dfn:repTree}).

We next apply our constructions to obtain new results for several online and demand-robust connectivity problems whose history we briefly summarize now. Group Steiner tree and group Steiner forest are two well-studied generalizations of set cover and Steiner tree. In the group Steiner tree problem, we are given a weighted graph $G=(V,E,w)$ and groups $g_1, \ldots, g_k \subseteq V$ and must return a subgraph of $G$ of minimum weight which contains at least one vertex from each group. The group Steiner forest problem  generalizes group Steiner tree. Here, we are given $A_i, B_i \subseteq V$ pairs and for each $i$ we must connect some vertex from $A_i$ to some vertex in $B_i$. \citet{alon2006general} and \citet{naor2011online} each gave a poly-log approximation for online group Steiner tree and forest respectively but both of these approximation guarantees are randomized and only hold against oblivious adversaries because they rely on FRT. Indeed, \citet{alon2006general} posed the existence of a deterministic poly-log approximation for online group Steiner tree as an open question which has since been restated several times \cite{buchbinder2009design,bienkowski2020nearly}. Similarly, while demand-robust minimum spanning tree and special cases of demand-robust Steiner tree have received considerable attention \cite{dhamdhere2005pay,khandekar2008two,kasperski2011approximability}, there are no known poly-log approximations for demand-robust Steiner tree, group Steiner tree or group Steiner forest.


\textbf{Application 1: Reducing Deterministic Online Group Problems to Tree Case (\Cref{sec:detOGST}).} 
In our first application we demonstrate that our copy tree embeddings reduce solving online group Steiner tree and forest deterministically on a general graph to the case of solving it on a tree. In particular, we show that a deterministic poly-log approximation for online group Steiner tree and forest on a tree graph gives a deterministic poly-log approximation on general graphs, thereby reducing the aforementioned open question of \citet{alon2006general} to its tree case.
\begin{theorem}\label{thm:GSTAndFor}
  If there exists an $\alpha$-competitive poly-time deterministic algorithm for group Steiner tree (resp. group Steiner forest) on well-separated trees then there exists an $O(\log n \cdot \alpha)$-competitive poly-time deterministic algorithm for group Steiner tree (resp. group Steiner forest) on general graphs.
\end{theorem}

Group Steiner tree has the notable property that mapping it onto a copy tree embedding simply results in another instance of the group Steiner tree problem, this time on a tree (our application 2 shows that this is not always the case). Therefore, this result is nearly immediate from either of the above constructions. In particular, if we have an instance of group Steiner tree on a general graph with groups $\{g_i\}_i$ then we can solve group Steiner tree on our embedding with groups $\{g_i'\}_i$ where $g_i' := \bigcup_{v \in g_i} \phi(v)$ and our root is the one copy of $r$, say $r'$. The connectivity properties of our mappings guarantee that a feasible solution for one of these problems is a feasible solution for the other when projected: if $g_i$ is connected to $r$ by $F$ then $g_i'$ is connected to $r'$ by $\pi_{G\to T}(F)$ and if $g_i'$ is connected to $r'$ by $F'$ then $g_i$ is connected to $r$ by $\pi_{T \to G}(F')$. Moreover, the cost preservation of $\pi_{G \to T}$ applied to the optimal solution on the input graph shows that our problem on the embedding has a cheap solution while the cost preservation of $\pi_{T \to G}$ allows us to map our solution on the embedding back to the input graph without increasing its cost. Lastly, the monotonicity of $\pi_{T \to G}$ guarantees that the resulting online algorithm only adds and never attempts to remove edges from its solution in $G$.

\textbf{Application 2: Deterministic Online Partial Group Steiner Tree (\Cref{sec:onPGST}).} We next introduce a new group connectivity problem---the online partial group Steiner tree problem. Partial group Steiner tree is group Steiner tree but where we must connect at least half of the vertices in each group to the root. As we discuss in \Cref{sec:onPGST}, partial group Steiner tree generalizes group Steiner tree. However, unlike group Steiner tree it admits a natural bicriteria relaxation: instead of connecting $\frac{1}{2}$ of the nodes in each group we could require that our algorithm only connects, say, $\frac{(1-\epsilon)}{2}$ of all nodes in each group for some $\epsilon > 0$.  Thus, this result can be seen as showing that there is indeed a deterministic poly-log competitive algorithm for online group Steiner tree---as posed in the above open question of \citet{alon2006general}---\emph{provided the algorithm can be bicriteria} in the relevant sense. More formally, we obtain a deterministic poly-log bicriteria approximation for this problem which connects at least $\frac{1-\epsilon}{2}$ of the nodes in each group (notated ``$(1-\epsilon)$-connection competitive'' below) by using our copy tree embeddings and a ``water-filling'' algorithm to solve the tree case.
\begin{restatable}{theorem}{partGST}
  There is a deterministic poly-time algorithm for online partial group Steiner tree which given any $\epsilon >0$ is $O\left(\frac{\log ^ 3 n}{\epsilon} \right)$-cost-competitive and $(1-\epsilon)$-connection competitive.
\end{restatable}

As we later observe, providing a deterministic poly-log-competitive algorithm for online partial group Steiner tree with any constant bicriteria relaxation is strictly harder than providing a deterministic poly-log-competitive algorithm for online (non-group) Steiner tree. Thus, this result also generalizes the fact that a deterministic poly-log approximation is known for online (non-group) Steiner tree \cite{imase1991dynamic}. Additionally, as a corollary we obtain the first non-trivial deterministic approximation algorithm for online group Steiner tree---albeit one with a linear dependence on the maximum group size.\footnote{We explicitly note here that this bicriteria guarantee does not yield a solution to the open problem of \cite{alon2006general} of finding a poly-log deterministic approximation to the online group Steiner tree problem.} As mentioned above, our approach for this problem requires that we use a copy tree with a poly-log copy number, thereby requiring that we use our first rather than our second construction.

We next adapt and apply our embeddings in the demand-robust setting. 


\textbf{Application 3: Demand-Robust Steiner Problems (\Cref{sec:DRGSTF}).} We begin by generalizing copy tree embeddings to demand-robust copy tree embeddings. Roughly, these are copy tree embeddings which simultaneously work well for every possible demand-robust scenario. We then adapt our analysis from our previous constructions to show that these copy tree embeddings exist. Lastly, we apply demand-robust copy tree embeddings to give poly-log approximations for the demand-robust versions of several Steiner problems---Steiner forest, group Steiner tree and group Steiner forest---for which, prior to this work, nearly nothing was known. In particular, the only non-trivial algorithms known for demand-robust Steiner problems prior to this work are an algorithm for Steiner tree \cite{dhamdhere2005pay} and an algorithm for demand-robust Steiner forest \emph{on trees} with exponential scenarios \cite{feige2007robust} (which is, in general, incomparable to the usual demand-robust setting). To show these results, we apply our demand-robust copy tree embeddings to reduce these problems to their tree case. Thus, we also give our results on trees which are themselves non-trivial.

\begin{restatable}{theorem}{DRSTT}\label{thm:demand-robust-steiner-tree-algo}   There is a randomized poly-time $O(\log^2 n)$-approximation algorithm for the demand-robust group Steiner tree problem on weighted trees.
\end{restatable}  \vspace{-\baselineskip}
\begin{restatable}{theorem}{DRSFT}\label{thm:demand-robust-steiner-forest-algo} 
    There is a randomized poly-time $O(D \cdot \log^3 n)$-approximation algorithm for the demand-robust group Steiner forest problem on weighted trees of depth $D$.
\end{restatable}

\begin{theorem}\label{thm:demand-robust-steiner-tree-algo-on-general-graph}
  There is a randomized poly-time $O(\log^4 n)$-approximation algorithm for the demand-robust group Steiner tree problem on weighted graphs.
\end{theorem}

\begin{theorem}\label{thm:demand-robust-steiner-forest-algo-on-general-graph}
  There is a randomized poly-time $O(\log^6 n)$-approximation for the demand-robust group Steiner forest problem on weighted graphs with polynomially-bounded aspect ratio.
\end{theorem}

Demand-robust group Steiner forest generalizes demand-robust Steiner forest and prior to this work no poly-log approximations were known for demand-robust Steiner forest; thus the above result gives the first poly-log approximation for demand-robust Steiner forest. We solve the tree case of the above problems by observing a connection between demand-robust and online algorithms. In particular, we exploit the fact that for certain online rounding schemes a demand-robust problem can be seen as an online problem with two time steps provided certain natural properties are met. Notably, these properties will be met for these problems \emph{on trees}. Thus, we emphasize that going through the copy tree embedding is crucial for our application---a more direct approach of using online rounding schemes on the general problem does not seem to yield useful results.

\textbf{Further Applications.} Lastly, we note that copy tree embeddings were integral to a follow-up work of the same set of authors~\cite{haeupler2020tree}, in which we gave the first poly-log approximations for the hop-constrained version of many classic network design problems, including hop-constrained Steiner forest \cite{agrawal1995trees}, group Steiner tree and buy-at-bulk network design~\cite{awerbuch1997buy}.

\section{Additional Related Work}
We survey some additional work before moving on to our results.

\subsection{Group Steiner Tree and Group Steiner Forest}

The group Steiner tree problem was introduced by \citet{reich1989beyond} as an important problem in VLSI design. \citet{garg2000polylogarithmic} gave the first randomized poly-log approximation for offline group Steiner tree using linear program rounding. \citet{charikar1998rounding} derandomized this result and \citet{chekuri2006greedy} showed that a greedy algorithm achieves similar results. \citet{demaine2009node} gave improved algorithms for group Steiner tree on planar graphs. 

As earlier mentioned, \citet{alon2006general} gave the first randomized poly-logarithmic algorithm for \emph{online} group Steiner tree which works against oblivious adveraries and posed the existence of a deterministic poly-log approximation as an open question. Very recently \citet{bienkowski2020nearly} made exciting progress towards this open question by giving a poly-log deterministic approximation for online non-metric facility location---which is equivalent to the online group Steiner tree on trees with depth $2$. We complement this result by narrowing the remaining gap on this question ``from the other end'' by showing that the tree case is all that needs to be considered. The authors also note that they believe that their methods could be used to give a deterministic poly-log-competitive algorithm for group Steiner tree on trees which, when combined with our own results, would settle this open question.


\citet{alon2006general} introduced the group Steiner \emph{forest} problem to study online network formation. \citet{chekuri2011set} gave the first poly-log approximation algorithm for offline group Steiner forest and posed the existence of a poly-log-competitive online algorithm as an open question. \citet{naor2011online} answered this question in the affirmative by showing that a randomized algorithm which works against oblivious adversaries exists but presently no adaptive-adversary-robust or deterministic poly-log-competitive online algorithm is known.

We note some nuances regarding necessary assumptions on the power of online algorithms for group Steiner tree and forest with an adaptive adversary.  \citet{alon2003online} observed that online set cover has no sub-polynomial-competitive algorithm against an adaptive adversary if the set system is not known beforehand. On the other hand, the same work showed how to give a poly-log-competitive algorithm for online set cover if the algorithm knows all possible elements the adaptive adversary might reveal (where the poly-log is poly-logarithmic in the total number of possible revealed elements). Set cover can easily be reduced to group Steiner tree on a tree where edges correspond to sets and elements correspond to leaves of the tree. Consequently, formulating any poly-log-competitive and adaptive-adversary-robust or deterministic algorithm for group Steiner tree requires that the algorithm knows all possible groups the adversary might reveal and that the number of possible groups is polynomially-bounded. As group Steiner tree is a special case of group Steiner forest, an analogous fact holds for group Steiner forest; namely all possible $(A_i, B_i)$ pairs that the adaptive adversary might reveal must be known beforehand to the algorithm for a poly-log competitive ratio and the number of such pairs must be polynomially-bounded.

\subsection{Tree Embedding Variants}


Our embeddings are similar in spirit to Ramsey trees and Ramsey tree covers \cite{mendel2006ramsey,naor2012scale,blelloch2016efficient,abraham2018ramsey,bartal2019covering}. Specifically, it is known that for every metric $(V,d)$ and $k$ there is some subset $S \subseteq V$ of size at least $n^{1-1/k}$ which embeds into a tree---a so-called Ramsey tree---with distortion $O(k)$ \cite{mendel2006ramsey}. Recursively applying (a slight strengthening of) this fact shows that there exist collections of Ramsey trees---so-called Ramsey tree covers---where each vertex $v$ has some ``home tree'' in which the distances to $v$ are preserved. A concurrent work of \citet{filtser2021clan} employed this machinery to devise ``clan embeddings'' where the trees of a Ramsey tree cover are merged and---like in our work---each vertex is mapped to its copies. This line of work has led to many applications in metric-type problems such as compact routing schemes. However, the guarantees of Ramsey tree covers and the embeddings built on them are insufficient for the connectivity problems in which we are interested in a slightly subtle way. We are interested in preserving the costs of entire subgraphs which, roughly speaking, requires that pairwise distances be preserved in \emph{every} tree that we merge. For this reason our copy tree embedding construction will use much of the machinery of the ``well-padded tree covers'' of \citet{gupta2006oblivious} which (implicitly) give exactly this guarantee rather than Ramsey-tree-type machinery.

Another recent work of \citet{barta2020online} was also concerned with tree embeddings for (not necessarily deterministic) online algorithms. This work designed tree embeddings to give competitive algorithms for network design problems competitive ratios are poly-logarithmic in the number of relevant terminals as opposed to the total number of nodes, $n$.

Lastly, we note that there has been considerable work on extending the power of tree embeddings to a variety of other settings including tree embeddings for planar graphs \cite{konjevod2001approximating}, dynamic tree embeddings \cite{forster2020dynamic,chechik2020dynamic}, distributed tree embeddings \cite{khan2012efficient} and tree embeddings where the resulting tree is a subgraph of the input graph \cite{alon1995graph,elkin2008lower,abraham2008nearly,koutis2011nearly,abraham2012using}.

\section{Graph Notation And Assumptions}
Throughout this paper we will work with weighted graphs of the form $G = (V, E, w)$ where $V$ and $E$ are the vertex and edge sets of $G$ and $w : E \to \mathbb{R}_{\ge 1}$ gives the weight of edges. We typically assume that $n := |V|$ is the number of nodes and write $[n] = \{ 1, 2, \ldots, n \}$. We will also use $V(G)$, $E(G)$ and $w_G$ to stand for the vertex set, edge set and weight function of $G$. Similarly, we will use $w_e$ to stand for $w(e)$ where convenient. For a subset of edges $F \subseteq E$, we use the notation $w(F) := \sum_{e \in F} w_G(e)$. We use $d_G : V \times V \to \mathbb{R}_{\geq 0}$ to give the shortest path metric according to $w$. We will talk about the diameter of a metric $(V,d)$ which is $\max_{u,v \in V} d(u,v)$; we notate the diameter with $D$. We use $B(v, x) := \{u \in V : d(v, u) \leq x\}$ to stand for the closed ball of $v$ of radius $x$ in metric $(V,d)$ and and $B_G(v, x)$ if $(V,d)$ is the shortest path metric of $G$ and we need to disambiguate which graph we are taking balls with respect to. We will sometimes identify a graph with the metric which it induces.

Notice that we have assumed that edge weights are non-zero and at least $1$. This will be without loss generality as for our purposes any $0$ weight edges may be contracted and scaling of edge weights ensures that the minimum edge weight is at least $1$.

\section{Copy Tree Embedding Constructions}\label{sec:partTree}
        
        In this section we give our two constructions of copy tree embeddings. We begin by giving our first copy tree embedding construction based on merging partial tree embeddings.
        
        \repTree*
        
        If it were possible to give a single tree embedding which simultaneously preserved all distances between all nodes then we could simply take such a tree embedding as our copy tree embedding. However, such a tree embedding is, in general, impossible. The key insight we use to overcome this issue is that one can approximately preserve distances in a  \emph{deterministic} way if one only embeds a constant fraction of all nodes in the input metric; we call such an embedding a partial tree embedding. Combining $O(\log n)$ such partial tree embeddings will give our construction.
        
        In more detail, in \Cref{sec:padHDtoRHST} we show that an appropriate $O(\log n)$ ``padded hierarchical decompositions'' gives $O(\log n)$ partial tree embeddings where every node is embedded a constant number of times. Next, we show that such a collection of partial tree embeddings indeed gives us a copy tree embedding as in \Cref{thm:repTreeConst}; the main observation that this reduction relies on is the constant congestion induced by Euler tours which will allow us to project from our input graph to our partial tree embeddings in a cost and connectivity-preserving fashion. Thus, our goal after this point is to compute an appropriate collection of padded hierarchical decompositions.
        
        In \Cref{sec:detpHD} we proceed to show how to compute the required collection of padded hierarchical decompositions. Our construction of hierarchical decompositions will make use of the FRT cutting scheme and paddedness properties of it previously observed by \citet{gupta2006oblivious}. To this end, we provide a novel derandomization of a node-weighted version of the FRT cutting scheme by  combing the powerful multiplicative weights methodology~\cite{arora2012multiplicative} together with the classic method of conditional expectation and pessimistic estimators.

        \subsection{From Padded Hierarchical Decompositions to Copy Tree Embeddings}\label{sec:padHDtoRHST}
        
              \citet{gupta2006oblivious} introduced the idea of padded hierarchical decompositions which we illustrate in \Cref{fig:HDAndPadded}.
        \begin{definition}
            A hierarchical decomposition $\mcH$ of a metric $(V,d)$ of diameter $D$ is a sequence of partitions $\mcP_0, \ldots, \mcP_h$ of $V$ where $h = \Theta(\log D)$ and:
            \begin{enumerate}
                \item The partition $\mcP_h$ is one part containing all of $V$;
                \item Each part in $\mcP_i$ has diameter at most $2^i$;
                \item $\mcP_{i}$ is a refinement of $\mcP_{i+1}$; that is, every part in $\mcP_{i}$ is contained in some part of $\mcP_{i+1}$.
            \end{enumerate}
        \end{definition}
        Notice that each part of $\mcP_0$ is a singleton node by our assumption that edge weights are at least $1$ (we assume that the constant in the theta notation of $h = \Theta(\log D)$ is sufficiently large).

        \begin{definition}[$\alpha$-Padded Node]
            For some $\alpha \le 1$, a node $v$ is $\alpha$-padded in hierarchical decomposition $\mcP_0, \ldots, \mcP_h$ if for all $i \in [0, h]$ the ball $B(v, \alpha \cdot 2^i)$ is contained in some part of $\mcP_i$.
        \end{definition}
    
    \begin{figure}
        \centering
        \begin{subfigure}[b]{0.49\textwidth}
            \centering
            \includegraphics[width=\textwidth,trim=0mm 20mm 0mm 20mm, clip]{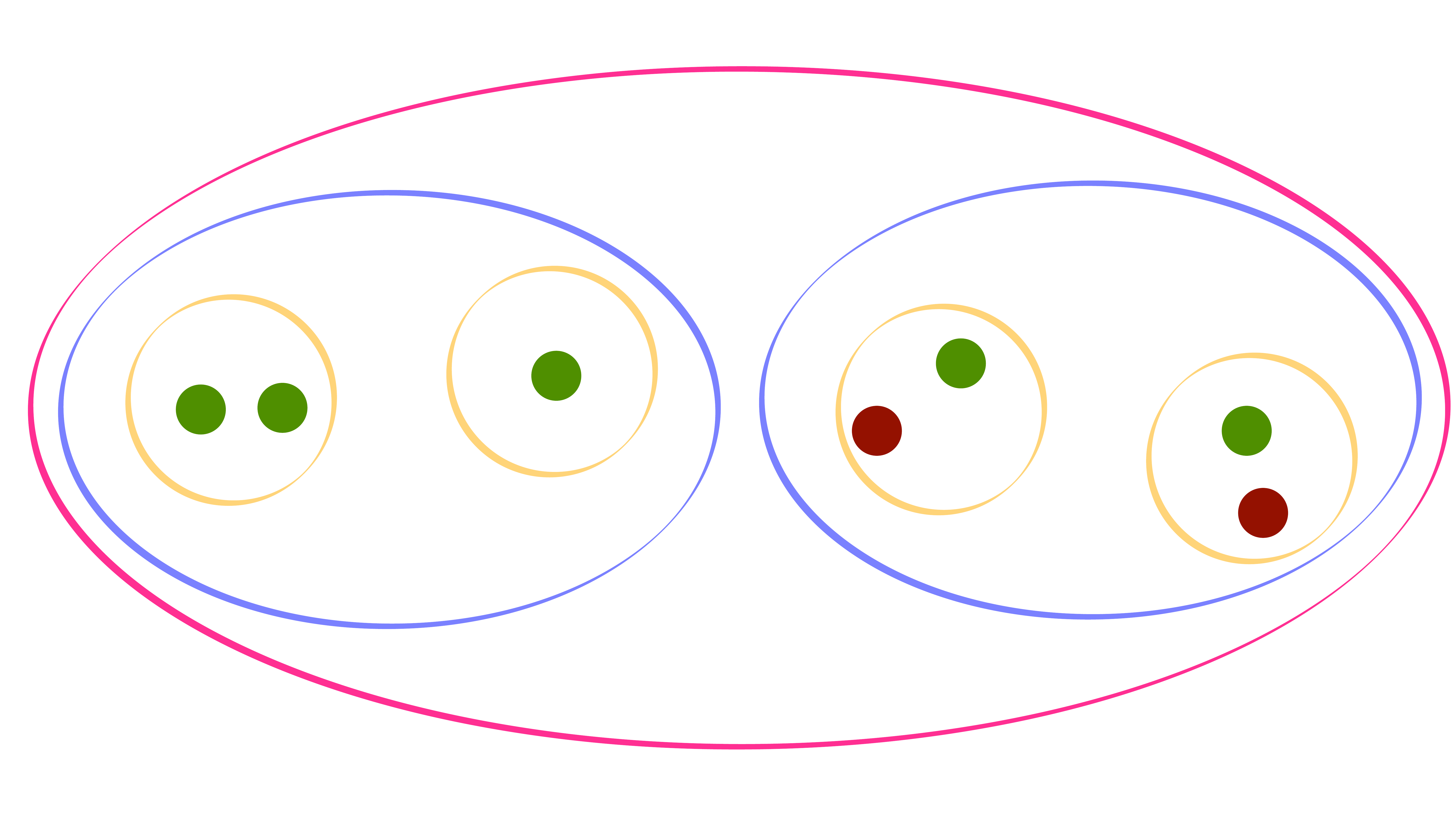}
            \caption{Hierarchical decomposition.}\label{fig:HDEG}
        \end{subfigure}
        \hfill
        \begin{subfigure}[b]{0.49\textwidth}
            \centering
            \includegraphics[width=\textwidth,trim=0mm 20mm 0mm 20mm, clip]{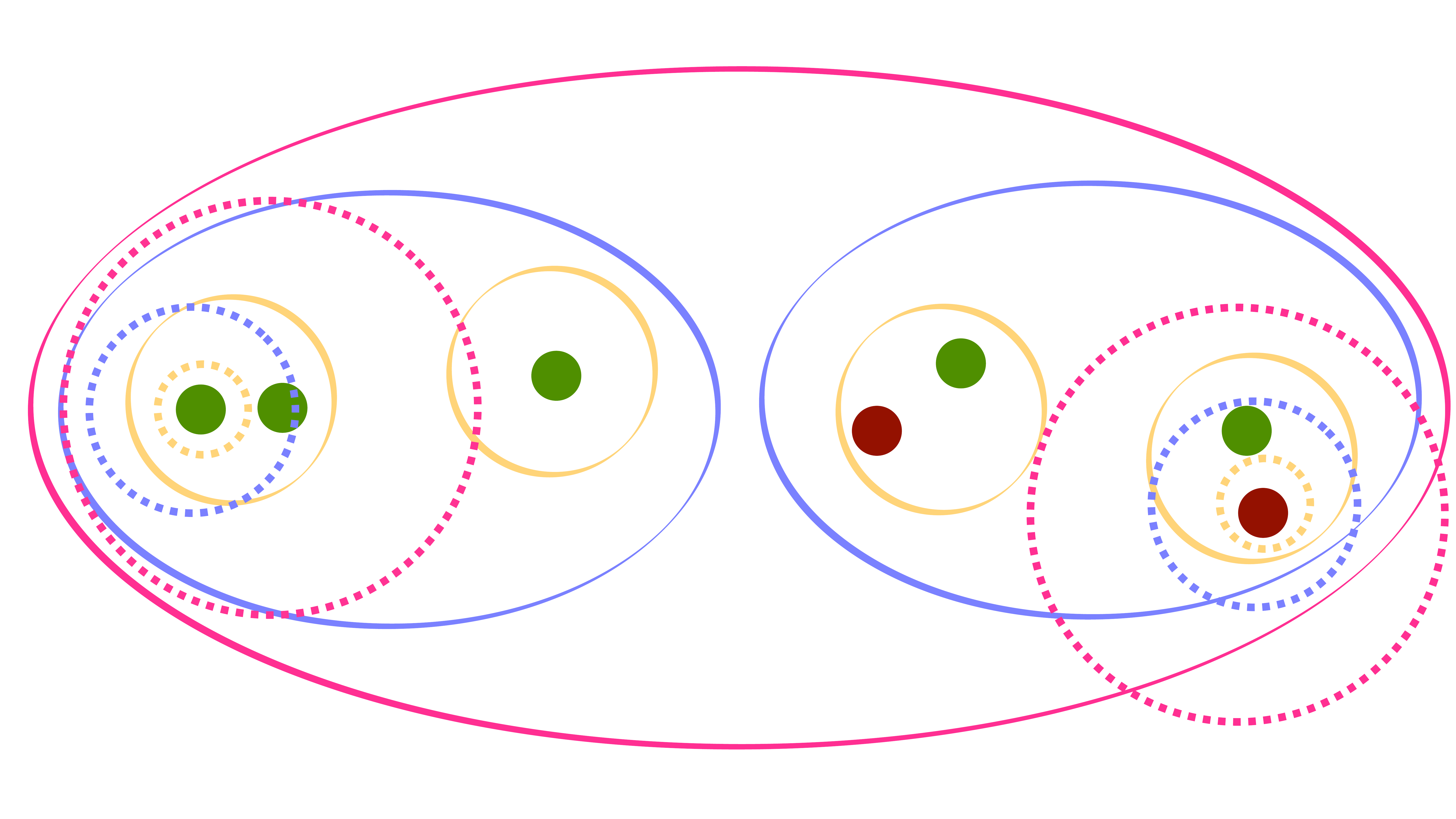}
            \caption{Why node on left is padded and node on right is not.}
        \end{subfigure}
        \hfill
        \caption{Illustration of a hierarchical decomposition $\mcH$ with $h=4$ with $n=7$. Each part in each $\mcP_i \in \mcH$ is colored according to $i$; singleton parts not pictured. We give $\alpha$-padded nodes in green and all other nodes in red where we illustrate why the node on the far left is $\alpha$-padded and the node on the far right is not by drawing $B(v, \alpha \cdot 2^i)$ for $i \geq 1$ in colors according to $i$ for these two nodes.}\label{fig:HDAndPadded}
    \end{figure}
        
        The main result we show in this section is how to use a collection of padded hierarchical decompositions to construct a copy tree embedding.
        \begin{restatable}{lemma}{RTEFromHDs}\label{lem:RTEFromHDs}
            Let $\{\mcH_i\}_{i=1}^{k}$ be a collection of hierarchical decompositions of weighted graph $G = (V, E , w)$ such that every $v$ is $\alpha$-padded in at least $.9k$ decompositions. Then, there is a poly-time deterministic algorithm which, given $\{\mcH_i\}_{i=1}^{k}$ and a root $r \in V$, returns an efficient and well-separated $O(\frac{k}{\alpha})$-approximate copy tree embedding with copy number $k$.
        \end{restatable}
  
        \subsubsection{From Padded Hierarchical Decompositions to Partial Tree Embeddings}
        
        We now formalize the notion of a partial tree embedding.
        
        \begin{definition}[Partial Tree Embedding]
            A $\gamma$-partial tree embedding of metric $(V, d)$ is a well-separated weighted tree $T = (V', E', w)$ where:
            \begin{enumerate}
                \item \textbf{Partial Embedding:} $V' \subseteq V$;
                \item \textbf{Worst-Case Distance Preservation} For any $u,v \in V'$ we have $d(u,v) \leq d_T(u,v) \leq \gamma \cdot d(u,v)$.
            \end{enumerate}
        \end{definition}
    In the remainder of this section we show how good padded hierarchical decompositions deterministically give good partial tree embeddings.
        

        The reason padded decompositions will be useful for us is that---as we prove in the following lemma---all distances between padded nodes are well-preserved.\footnote{This fact seems to be implicit in \citet{gupta2006oblivious} but is never explicitly proven.} Given a hierarchical decomposition $\mcH$ we let $T_\mcH$ be the natural well-separated tree corresponding to $\mcH$. In particular, a hierarchical decomposition $\mcH$ naturally corresponds to a well-separated tree which has a node for each part and an edge of weight $2^i$ between a part in $\mcP_i$ and a part in $\mcP_{i+1}$ if the latter contains the former. In \Cref{fig:HDtoTreeEG} we illustrate the well-separated tree corresponding to the hierarchical decomposition in \Cref{fig:HDEG}. We will slightly abuse notation and identify each singleton set in such a tree with its one constituent vertex. 
        \begin{lemma}\label{lem:padGivesDist}
                If nodes $u, v$ are $\alpha$-padded in a hierarchical decomposition $\mcH$ then $d(u,v) \leq d_{T_\mcH}(u,v) \leq O\left(\frac{1}{\alpha}\cdot d(u,v)\right)$.
        \end{lemma}
        \begin{proof}
                Let $T_\mcH$ be the well-separated tree corresponding to $\mcH$. Let $w$ be the least common ancestor of $u$ and $v$ in $T_\mcH$ and let $l$ be the height of $w$ in $T_\mcH$. By the definition of $T_\mcH$, the distance between $u$ and $v$ in $T_\mcH$ is $d_{T_{\mcH}}(u,v) = 2 \cdot \sum_{i=0}^{l} 2^i$ and so we have 
                \begin{align}\label{eq:lcadist}
                2^{l+1} \leq d_{T_\mcH}(u,v) \leq 2^{l+2}.
                \end{align}
                
                We next prove that $d_{T_\mcH}(u, v) \leq O(\frac{1}{\alpha} \cdot d(u,v))$. Notice that for $j = \lceil \log (d(u,v)/\alpha) \rceil$ we know that $B(v, \alpha \cdot 2^j)$ contains $u$ since for this $j$ it holds that $\alpha \cdot 2^j \geq d(u, v)$. Since $\mcH$ is $\alpha$-padded it follows that $B(v, \alpha \cdot 2^j)$ is contained in some part of $\mcP_j$; but it then follows that the least common ancestor of $u$ and $v$ is at height at most $j$ and so $l \leq \lceil \log (d(u,v)/\alpha) \rceil$. Combining this with the upper bound in \Cref{eq:lcadist} we have
                \begin{align*}
                d_{T_\mcH}(u,v) &\leq 2^{l+2}\\ 
                &\leq 2^{\lceil \log (d(u,v)/\alpha) \rceil + 2}\\
                & \leq O\left(\frac{1}{\alpha} \cdot d(u,v) \right)
                \end{align*}
                
                We now prove that $d(u,v) \leq d_{T_\mcH}(u,v)$. Since the diameter of each part in $\mcP_i$ is at most $2^i$ we know that the least common ancestor of $u$ and $v$ in $T$ corresponds to a part with diameter at most $2^l$. However, since the least common ancestor of $u$ and $v$ corresponds to a part which contains both $u$ and $v$, we must have $d(u,v) \leq 2^l \leq 2^{l+1}$. Combining this with the lower bound in \Cref{eq:lcadist} we have $d(u,v) \leq d_{T_\mcH}(u,v)$ as desired.
        \end{proof}

            \begin{figure}
            \centering
            \begin{subfigure}[b]{0.32\textwidth}
                \centering
                \includegraphics[width=\textwidth,trim=150mm 200mm 200mm 45mm, clip]{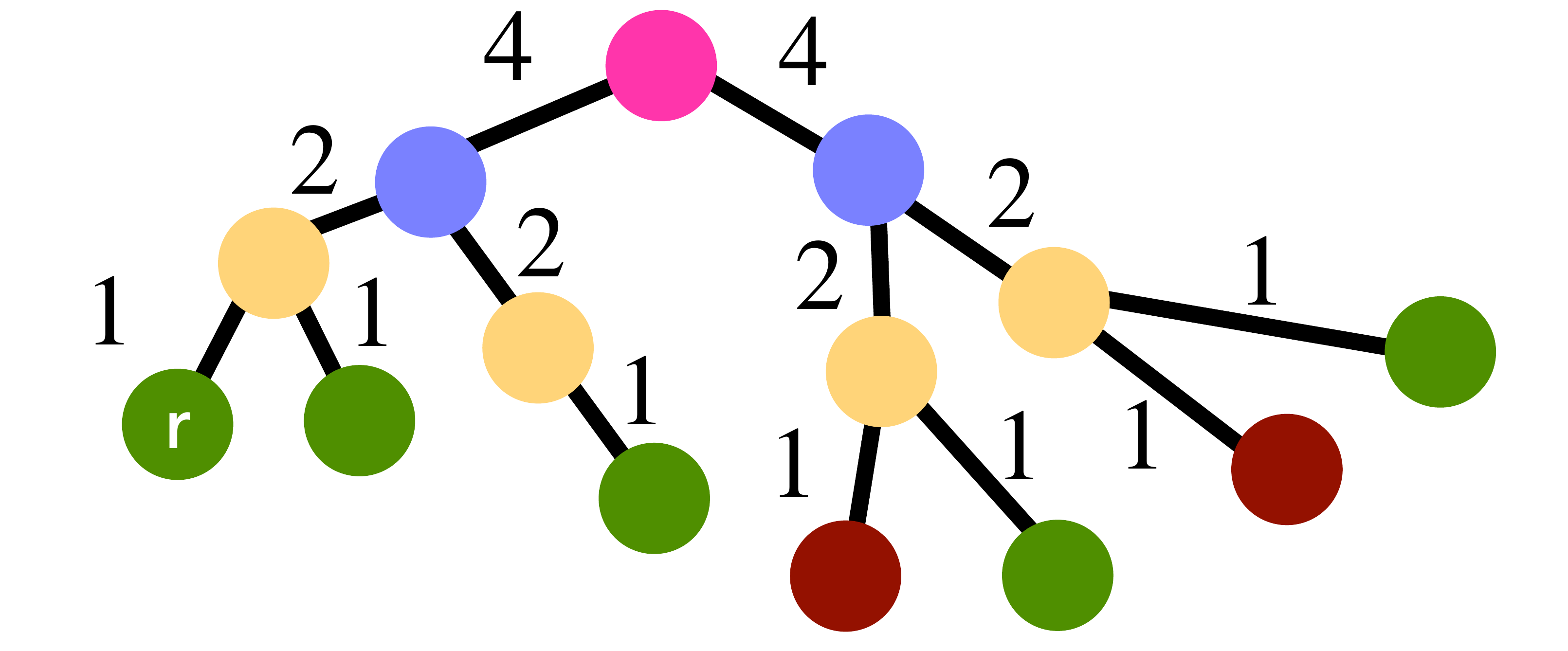}
                \caption{Tree corresponding to \Cref{fig:HDEG} hierarchical decomposition.}\label{fig:HDtoTreeEG}
            \end{subfigure}
            \hfill
            \begin{subfigure}[b]{0.32\textwidth}
                \centering
                \includegraphics[width=\textwidth,trim=150mm 200mm 200mm 45mm, clip]{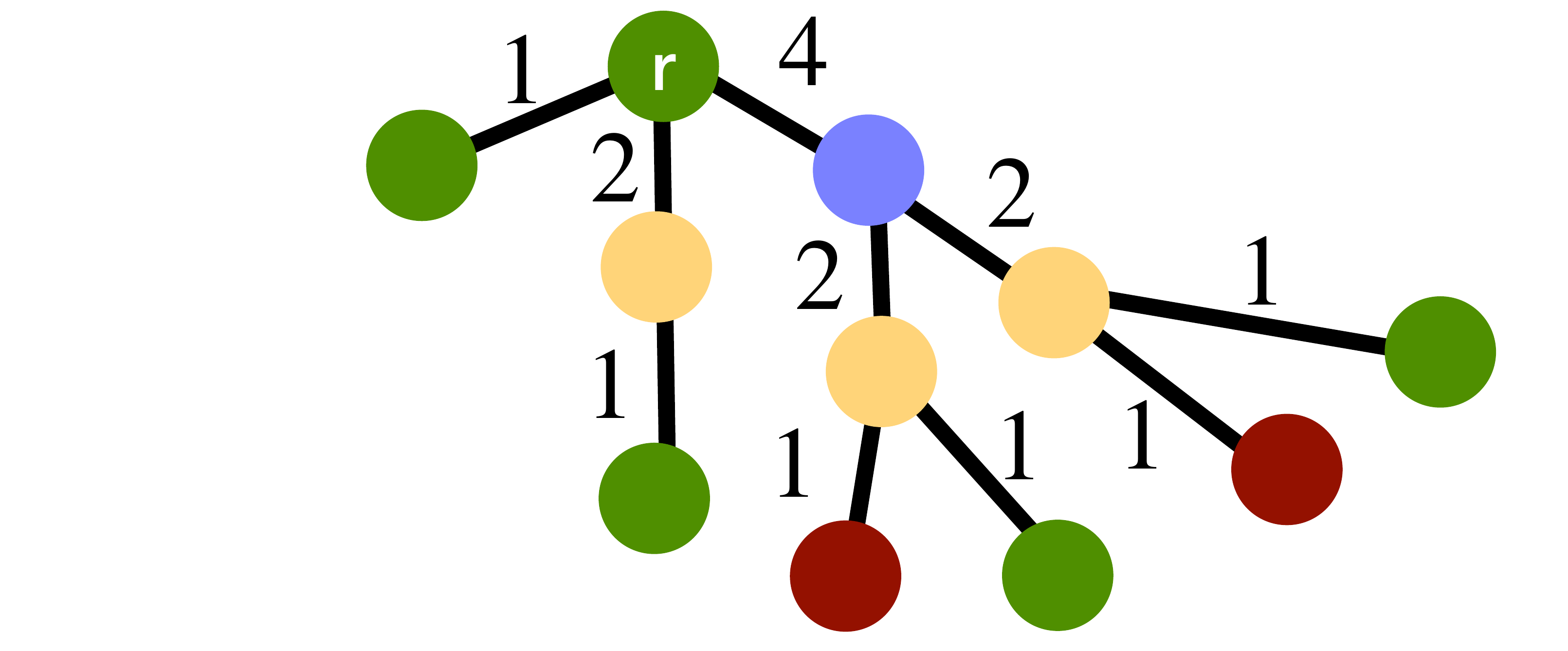}
                \caption{Contract to ensure $r$ is root of resulting tree.}
            \end{subfigure}
            \hfill
            \begin{subfigure}[b]{0.32\textwidth}
                \centering
                \includegraphics[width=\textwidth,trim=150mm 200mm 200mm 45mm, clip]{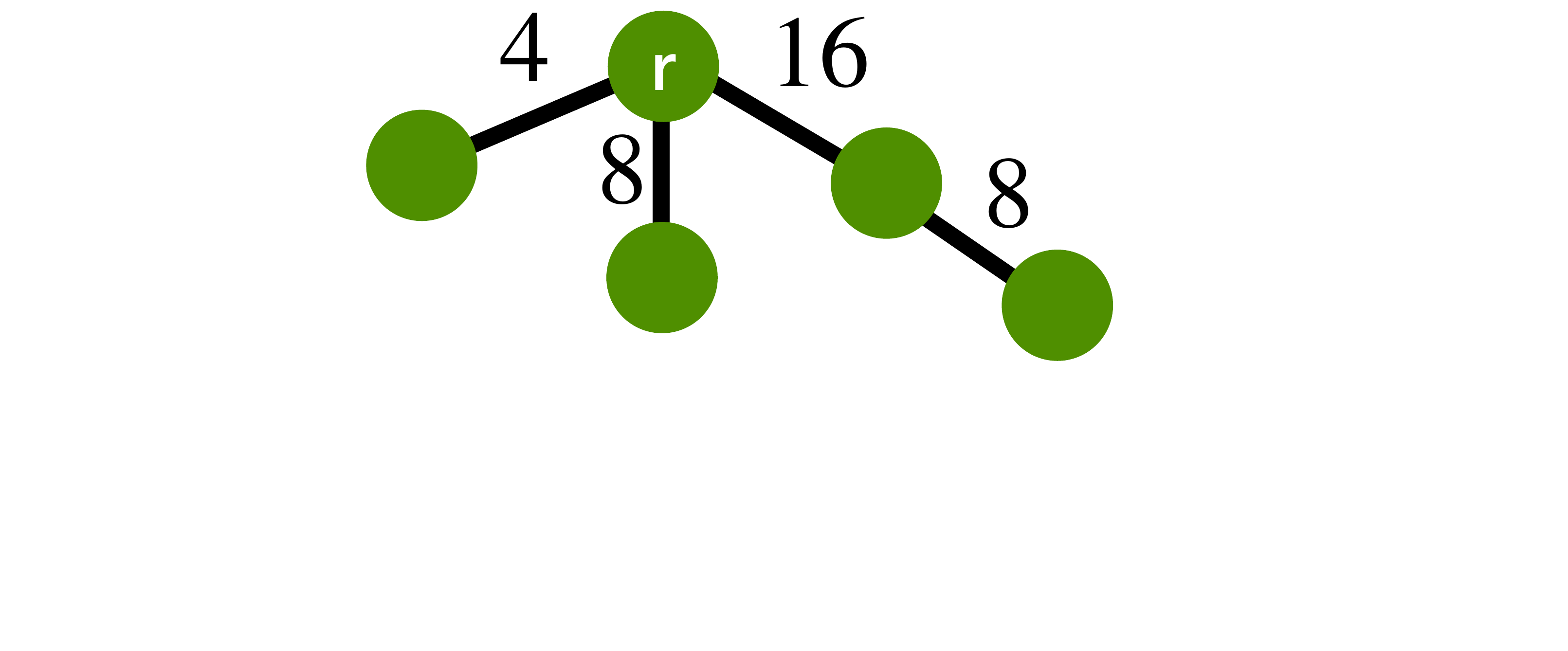}
                \caption{Multiply weights by $4$ and contract non-$\alpha$-padded vertices.}
            \end{subfigure}
            \hfill
            \caption{How to turn a hierarchical decomposition into a partial tree embedding. We color nodes from the input metric in green if they are padded and red otherwise. Remaining nodes colored according to their corresponding hierarchical decomposition part. $r$ is the node on the far left of the tree.}\label{fig:HDContract}
        \end{figure}
    
        We show how to turn a hierarchical decomposition into a partial tree embedding in the next lemma which we illustrate in \Cref{fig:HDContract}.
        \begin{lemma}\label{lem:decompToPartialHST}
                Given a hierarchical decomposition $\mcH$ on metric $(V,d)$ and root $r \in V$ which is $\alpha$-padded in $\mcH$, one can compute in deterministic poly-time a $O(\frac{1}{\alpha})$-partial tree embedding $T=(V', E')$ with root $r$ where $V' := \{v \in V : \text{$v$ is $\alpha $ padded}\}$,
        \end{lemma}
        \begin{proof}
                Let $T_{\mcH}$ be the well-separated tree which corresponds to $\mcH$ as described above.
                
                We construct $T$ from $T_{\mcH}$ using \Cref{lem:padGivesDist} and a trick of \citet{konjevod2001approximating}. Let $V'$ be all leaves of $T_{\mcH}$ whose corresponding nodes are $\Omega(\frac{1}{\log n})$-padded in $\mcH$. Next, contract the path from $r$ to the root of $T_{\mcH}$ and identify the resulting node with $r$. Then, delete from $T_{\mcH}$ all sub-trees which do not contain a node in $V'$; in the resulting tree every node is either in $V'$ or the ancestor of a node in $V'$. Next, while there exists a node $v$ such that its parent $u$ is not in $V'$ we contract $\{v,u\}$ into one node and identify the resulting node with $v$. Lastly, we multiply the weight of every edge by $4$ and return the result as $T = (V', E', w)$ where, again, $w$ is the weight function of $T_{\mcH}$ times $4$.
                
                Clearly, the vertex set of $T$ will be $V'$. Moreover, $T$ is well-separated since $T_{\mcH}$ was well-separated and $r$ will be the root of $T$ by construction.
                
                We now use an analysis of \citet{konjevod2001approximating} to show that for any pair of vertices $u,v \in V'$ we have 
                \begin{align}\label{eq:raviFRTTrick}
                        d_{T_{\mcH}}(u,v) \leq d_T(u,v) \leq 4 \cdot d_{T_{\mcH}}(u,v)
                \end{align}
                The upper bound is immediate from the fact that we only contract edges and then multiply all edge weights by $4$. To see the lower bound---$d_{T_{\mcH}}(u,v) \leq d_T(u,v)$---notice that if $u$ and $v$ have a least common ancestor $a$ at height $l$ in $T_{\mcH}$, then $d_{T_{\mcH}}(u, v) = 2^{l+2} - 4$. However, the closest $u$ and $v$ can be in $T$ is if (without loss of generality) $u$ is identified with $a$ and (without loss of generality) $v$ is a child of $u$ in $T$; the length of this edge is the length of a child edge of $a$ in $T_{\mcH}$ times four which is $2^{l+2}$. Thus $d_{T_{\mcH}}(u, v) = 2^{l+2} - 4 \leq 2^{l+2} = d_T(u,v)$.
                
                Finally, we conclude by applying \Cref{lem:padGivesDist}. In particular, it remains to show $d(u,v) \leq d_T(u,v) \leq O(\frac{1}{\alpha} \cdot d(u,v))$ but this is immediate by combining \Cref{lem:padGivesDist} and \Cref{eq:raviFRTTrick}.
        \end{proof}
    
        \subsubsection{From Partial Tree Embeddings to Copy Tree Embeddings}
        
        We now describe how partial tree embeddings satisfy useful connectivity properties and then use these properties to construct a copy tree embedding from a collection of good partial tree embeddings.
        
        The following two lemmas demonstrate how to map to and from partial tree embeddings in a way that preserves cost and connectivity.
        
        \begin{lemma}[Graph $\to$ Partial Tree Projection]\label{lem:conPresGrapTree}
            Let $G = (V, E, w_G)$ be a weighted graph and let $T = (V', E', w_T)$ be a $\gamma$-partial tree embedding of (the metric induced by) $G$.
            There exists a deterministic, poly-time computable function $\pi : 2^{E} \to 2^{E'}$ such that for all sets of edges $F \subseteq E$ the following holds:
            \begin{enumerate}
                \item \textbf{Connectivity Preservation:} If $u, v \in V'$ are connected by $F$ in $G$, then they are connected in $\pi(F)$ in $T$;
                \item \textbf{Cost Preservation:} $w_T(\pi(F)) \leq O(\gamma) \cdot w_G(F)$.
            \end{enumerate}
        \end{lemma}
        \begin{proof}
            We first simplify $F$ by noticing it is sufficient to prove the claim on every connected component in isolation. Furthermore, we can assume without loss of generality that $F$ is a tree since taking a spanning tree of $F$ can only decrease $w_G(F)$ and appropriately maintains connectivity. Finally, we delete every leaf that is not in $V'$, which decreases $w_G(F)$ and maintains connectivities in $V'$.
            
            We define $\pi(F)$ to be the unique minimal subtree of $T$ which contains all nodes of $V'$ that are incident to an edge in $F$. By transitivity of connectedness, we know that if $u,v \in V'$ are connected in $F$ then they must also be connected in $\pi(F)$. Also, note that $\pi$ is trivially deterministic poly-time computable.
            
            It remains to argue the $\gamma$-cost preservation property. Double the edges of $F$; we call this multigraph $2F$. Since the degree of every vertex in $2F$ is even, we know that $2F$ has an Euler tour. Using this tour we can partition $2F$ into a set $\mcP$ of paths where each path connects two nodes in $V'$ and the paths in $\mcP$ are multiedge-disjoint. Therefore, we have that $2 w_G(F) = \sum_{P \in \mcP} w_G(P)$.
            
            For each path $P \in \mcP$ in the tour between nodes $u, v \in V'$, we say that $P$ \textbf{covers} all edges in $T$ between $u$ and $v$ and let $P'$ be the path in $T$ between $u$ and $v$. We note that every edge in $\pi(F)$ is covered by at least one path, hence $w_T(\pi(F)) \le \sum_{P \in \mcP} w_T(P')$.
            
            For every path in $G$ connecting two nodes $u, v \in V'$ the distance-preservation properties of $\gamma$-partial tree embeddings implies that $w_T(P') \le O(\gamma) \cdot w_G(P)$. Hence we have that $w_T(\pi(F)) \le \sum_{P \in \mcP} w_T(P') \le O(\gamma) \cdot \sum_{P \in \mcP} w_G(P) \le O(\gamma) \cdot w_G(F)$ as required.
        \end{proof}
        
        We now show how to project in the reverse direction.
        
        \begin{lemma}[Partial Tree $\to$ Graph Projection]\label{lem:conPresTreeGraph}
            Let $G = (V, E, w_G)$ be a weighted graph and let $T = (V', E', w_T)$ be a $\gamma$-partial tree embedding of (the metric induced by) $G$.
            There exists a deterministic, poly-time computable function $\imath : 2^{E'} \to 2^{E}$ such that for all sets of edges $F' \subseteq E'$ the following holds:
            \begin{enumerate}
                \item \textbf{Connectivity Preservation:} If $u, v \in V'$ are connected by $F'$ in $T$, then they are connected by $\imath(F)$ in $G$;
                \item \textbf{Cost Preservation:} $w_G(\imath(F')) \le w_T(F')$.
            \end{enumerate}
        \end{lemma}
        \begin{proof}
            For an edge $e' \in E'$, connecting $u, v \in V'$, we define $\imath(\{e'\})$ as some shortest path between $u$ and $v$ in $G$. Note that this implies that $w_G(\imath(\{e'\})) \le w_T(e')$ by the properties of a partial tree embedding. We extend $\imath$ to $F' \subseteq E'$ by defining $\imath(F') := \bigcup_{e' \in F'} \imath(\{e'\})$. Notice that $\imath$ is indeed deterministic, poly-time computable and is connectivity preserving by the transitivity of connectivity. 
            
            We now verify the cost preservation of $\imath$: we have that $w_G(\imath(F')) = w_G(\bigcup_{e' \in F'} \imath(\{e'\})) \le \sum_{e' \in F'} w_G(\imath(\{e'\})) \le \sum_{e' \in F'} w_T(e') = w_T(F')$.
        \end{proof}
    
    Using these two properties we can conclude our proof of \Cref{lem:RTEFromHDs}, which we restate here.
    \RTEFromHDs*
    \begin{proof}
        Our embedding is gotten by combining the above lemmas in the natural way.
        
        Specifically, we first apply \Cref{lem:decompToPartialHST} to all decompositions in $\{\mcH_i\}_{i=1}^{k}$ in which $r$ is $\alpha$-padded to get back $O(\frac{1}{\alpha})$-partial tree embeddings $\{T_i\}_{i}$ where $V(T_i) = \{ v : \text{$v$ is $\alpha$-padded in $\mcH_i$\}}$. Next we apply \Cref{lem:conPresGrapTree} and \Cref{lem:conPresTreeGraph} to each $T_i$ to get back mapping functions $\pi_i$ and $\imath_i$ respectively.
        
        We now describe our $O(\frac{k}{\alpha})$-approximate copy tree embedding $(T, \phi, \pi_{G \to T}, \pi_{T \to G})$. We let $T$ be the tree resulting from taking all trees in $\{T_i\}_i$ and then identifying all copies of $r$ as the same vertex. Similarly, we let $\phi(v)$ be the set of all copies of $v$ in $T$ in the natural way. Next we let $\pi_{G \to T}(F)$ be $\bigcup_i \pi_i(F)$ where $\pi_i$ is projected onto $T$ in the natural way. We let $\pi_{T \to G}(F') := \bigcup_i \imath_i(F')$ be defined analogously.
        
        Since each vertex appears in at least a $.9$ fraction of all $T_i$, by the pigeonhole principle we know that any pair connected by $F$ in $G$ must occur in some $\mcH_i$ together with $r$ and so must be connected in $\pi_i(F)$ for some $i$ where $T_i \in \{T_i\}_i$ and so some pair of corresponding copies are connected by $\pi_{G \to T}$; an analogous result holds for $\pi_{T \to G}$. The remaining properties of our embedding are immediate from the above cited lemmas.
    \end{proof}

     \subsection{Deterministically Constructing Padded Hierarchical Decompositions}\label{sec:detpHD}
     In the previous section we reduced computing good copy tree embeddings to computing good hierarchical decompositions. The existence of good hierarchical decompositions is immediate from prior work of \citet{gupta2006oblivious} and FRT.
        
        \begin{lemma}[\citet{gupta2006oblivious}]\label{lem:FRTIsPadded}
                Let $\mcH$ be the hierarchical decompositions resulting from $a$ tree drawn from the \citet{fakcharoenphol2004tight} cutting scheme. Then, every vertex is $\Omega(\frac{1}{\log n})$-padded with constant probability in $\mcH$.
        \end{lemma}
        A simple Chernoff and union bound proof then gives that $O(\log n)$ draws gives a collection of hierarchical decompositions in which every vertex is $\Omega(\frac{1}{\log n})$-padded in a constant fraction of the decompositions \emph{with high probability}, i.e.\ at least $1-\frac{1}{\poly (n)}$.
    
        However, we are ultimately interested in a deterministic algorithm which is robust to adaptive adversaries and so we must derandomize the above with high probability result. We proceed to do so in this section.
        
        
%

        
        To our knowledge, prior derandomizations of this cutting scheme---see, e.g. \citet{chekuri2006approximation} or \citet{fakcharoenphol2004tight}---do not provide sufficiently strong guarantees for our purposes.
        We also note that the authors of \citet{gupta2006oblivious} claim to give a deterministic algorithm for computing hierarchical decompositions in a forthcoming journal version of their paper but said journal version never seems to have been published.

        \subsubsection{Derandomization Intuition}
        
        The intuition behind our derandomization is as follows. A single draw from the FRT cutting scheme guarantees that each node is $\Omega(1 / \log n)$-padded with constant probability. If we could derandomize this result then we could produce one hierarchical decomposition such that at least a $.99$ fraction of all nodes are $\Omega(1 / \log n)$-padded. Indeed, as we will see, standard derandomization techniques---the method of pessimistic estimators and conditional expectation---will allow us to do exactly this. However, since we must produce a collection of hierarchical decompositions in which every node is in a large percentage in all decompositions it is not clear how, then, to handle the remaining $.01$ fraction of nodes. One might simply rerun the aforementioned derandomization result on the remaining $.01$ nodes, then on the remaining $.001$ nodes and so on logarithmically-many times; however, it is easy to see that in the resulting collection of decompositions, while every node is padded in some decomposition, no node is necessarily padded in a large fraction of all the decompositions.
        
        Rather, we would like to repeatedely run our derandomization on all nodes but in a way that takes into account which nodes are already padded in a large fraction of the decompositions we have already produced. In particular, if a node was already padded in most of the decompositions we have so far produced, we need not worry about producing decompositions in which this node is padded. Thus, we would like to derandomize in a way that would make such a node less likely to be padded in the remaining decompositions we produce while making nodes which have not so far been padded in many decompositions we produced more likely to be padded.
        
        To accomplish this, we will formulate and then derandomize a \emph{node-weighted} version of \Cref{lem:FRTIsPadded}; this, in turn, will allow us to down-weight nodes which are padded in a large fraction of the decompositions we have so far produced when we run our derandomization; a multiplicative-weights-type analysis will then allow us to conclude our deterministic construction.
        
        \subsubsection{The FRT Cutting Scheme}
        
        In order to give our deterministic construction we must unpack the black box of the FRT cutting scheme. 
        
        The \citet{fakcharoenphol2004tight} cutting scheme given metric $(V, d)$ where $d(u,v) \geq 1$ for all $u,v \in V$ produces a hierarchical decomposition $\mcH = \{\mcP_0, \ldots, \mcP_h\}$ and is as follows. We first pick a uniformly random permutation $\pi$ on $V$ and a uniformly random value $\beta \in [\frac{1}{2}, 1)$. We let the radius for level $i$ be $r_i := 2^{i-1} \cdot \beta$.
        
        We let $\mcP_h$ be the trivial partition containing all vertices of $V$ with $h = O(\log \max_{u,v} d(u,v))$. Next, we construct $\mcP_{i}$ by refining $\mcP_{i+1}$; in particular we divide each part $P_{i+1} \in \mcP_{i+1}$ into additional parts as follows. Each $v \in P_{i+1}$ is assigned to the first vertex $u$ in $\pi$ for which $v \in B(u, r_i)$. Notice that $u$ need not be in $P_{i+1}$. Let $C_u$ be all vertices in $P_{i+1}$ which are assigned to $u$ and add to $\mcP_i$ all $C_u$ which are non-empty. Notice that here $C_u$ really depends on $i$; we suppress this dependence in our notation for cleanliness of presentation.
        
        One can easily verify that the resulting partitions indeed form a hierarchical decomposition.
        
        \subsubsection{Derandomizing via Multiplicative Weights and Pessimistic Estimators}
        
        As discussed above, our goal is to derandomize \Cref{lem:FRTIsPadded} while taking node weights into account. Suppose we have a distribution $\{p_v\}_v$ over vertices in $v$; intuitively this distribution how important each vertex is in regards to being $\alpha$-padded. Then by \Cref{lem:FRTIsPadded} and linearity of expectation we have
        \begin{align*}
        \E_{\pi, \beta} \left[\sum_v p_v \cdot \mathbb{I}\left(\text{$v$ is $\Omega\left(\frac{1}{\log n}\right)$-padded in $\mcH$}\right) \right] &= \sum_v p_v \cdot \Pr_{\pi, \beta} \left(\text{$v$ is $\Omega\left(\frac{1}{\log n}\right)$-padded in $\mcH$}\right)\\ &\geq .99.
        \end{align*}
        where $\mathbb{I}$ is the indicator function.
        
        Thus, our goal will be to gradually fix the randomness of $\pi$ and $\beta$ until we have found a way to deterministically set $\beta$ and $\pi$ so that at least a $.95$ fraction of nodes (weighted by $p_v$s) are $\Omega(\frac{1}{\log n})$-padded. That is, we aim to use the method of conditional expectation. We will treat a permutation $\pi$ as an ordering of the elements of $[V]$. E.g. $(v_2, v_1, v_3)$ is a permutation of $V = \{v_1, v_2, v_3\}$. Now, suppose we have fixed a prefix $\pi_P$ of $\pi$ which orders nodes $P \subseteq V$ and among the remaining $\bar{P} := V \setminus P$ we uniformly at randomly choose the remaining suffix $\pi_{\bar{P}}$. That is, $\pi = \pi_P \odot \pi_{\bar{P}}$ where $\pi_P$ is fixed and $\pi_{\bar{P}}$ is a uniformly random permutation over $\bar{P}$ and $\odot$ is concatenation. Notice that it follows that every vertex of $P$ will precede every vertex of $\bar{P}$ in $\pi$.
        
         Let $\mcH(\pi_P, \beta)$ be the hierarchical decomposition returned when we run the FRT cutting scheme as above with the input value of $\beta$ and with $\pi$ chosen as $\pi = \pi_P \odot \pi_{\bar{P}}$. Notice that provided $P \neq V$ we have that $\mcH$ is a randomly generated. Let $f(\pi_P, \beta) := \sum_v p_v \cdot \Pr_{\pi_{\bar{P}}} \left(\text{$v$ is $\Omega\left(\frac{1}{\log n}\right)$-padded in $\mcH(\pi_{P}, \beta)$}\right)$ be the fraction of $\Omega(\frac{1}{\log n})$-padded nodes by weight in expectation in $\mcH(\pi_P, \beta)$. We now show that there is a so called ``pessimistic estimator'' $\hat{f}$ of $f$.
        \begin{lemma}\label{lem:pessEst}
                There is a function $\hat{f}$ such that 
                \begin{enumerate}
                        \item \textbf{Good start:} There is some deterministically poly-time computable set $R \subseteq \mathbb{R}$ such that for some $\beta \in R$ we have $\hat{f}(\pi_\emptyset, \beta) \geq .95$.
                \end{enumerate}
                and for any $P \subseteq V$, $\pi_P$ and $\beta$
                \begin{enumerate} \setcounter{enumi}{1}
                        \item \textbf{Computable:} $\hat{f}(\pi_P, \beta)$ is computable in deterministic poly-time;
                        \item \textbf{Monotone:} $\hat{f}(\pi_P, \beta) \leq \hat{f}(\pi_{P \cup \{v\}}, \beta)$ for some $v \in \bar{P}$;
                        \item \textbf{Pessimistic:} $\hat{f}(\pi_P, \beta) \leq f(\pi_P, \beta)$ for all $\pi_P$ and $\beta$.

                \end{enumerate}
        \end{lemma}
        \begin{proof}
                We will use an analysis similar to \citet{gupta2006oblivious} but which accounts for the fixed prefix $\pi_P$ of our permutation, demonstrates the above properties of our pessimistic estimator and which guarantees that $R$ is computable in deterministic, poly-time.
                
                 We begin by defining $\hat{f}$. Fix a $\pi_P$ and $\beta$ and let $\alpha = \Omega(\frac{1}{\log n})$.

                For node $v$, let $B_{i, v} := B(v, \alpha 2^i)$. Say that node $u$ \emph{protects} $B_{i,v}$ if its ball at level $i$ contains $B_{i,v}$, i.e.\ if $r_i \geq d(u,v) + 2^i \alpha$. Say that $u$ \emph{threatens} $B_{i,v}$ if its ball at level $i$ intersects $B_{i,v}$ but does not contain it, i.e.\ $d(u,v) - \alpha 2^i < r_i < d(u,v) + 2^i \alpha$. Finally, say that $u$ \emph{cuts} $B_{i,v}$ if it threatens $B_{i,v}$ and is the first node in $\pi$ to threaten or protect $B_{i,v}$. Notice that if $B_{i,v}$ is not cut by any node for all $i$ then $v$ will be $\alpha$-padded. 
                
                In order for $B_{i,v}$ to be cut by $u$ it must be the case that $u$ threatens $B_{i,v}$ and no node before $u$ in $\pi$ threatens or protects $B_{i,v}$. By how we choose $r_i$, $u$ threatens $B_{i,v}$ if 
                \begin{align}\label{eq:threaten}
                d(u,v) - 2^i \alpha < \beta \cdot 2^{i-1} < d(u,v) + 2^i \alpha
                \end{align}
                
                In order for $u$ to be the first node to threaten or protect $B_{i,v}$, it certainly must be the case that every node which is closer to $v$ than $u$ appears after $u$ in $\pi$ (since every such node either threatens or protects $B_{i,v}$). Thus, we let $N_{v}(u) := \{w : d(w,v) \leq d(u,v)\}$ be all nodes which are nearer to $v$ than $u$.
                
                Lastly, a node which is too far or too close to $v$ cannot cut $B_{i,v}$. In particular, a node $u$ can only cut $B_{i,v}$ if 
                \begin{align}\label{eq:cutSet}
                2^{i-2} - 2^i \alpha \leq d(u,v) \leq 2^{i-1} +2^i \alpha
                \end{align}
                
                
                We let $C_{i,v} := \{u : 2^{i-2} - 2^i \alpha \leq d(u,v) \leq 2^{i-1} + 2^i\alpha \}$ be all such nodes which might cut $B_{i,v}$.
                
                It follows that we have that $B_{i,v}$ is cut only if there exists some $u$ in $C_{i,v}$ which both threatens $v$ and precedes all $w \in N_v(u) \setminus \{u\}$ in $\pi$. Thus, we define $\hat{f}$ as follows
                \begin{align*}
                \hat{f}(\pi_P, \beta) := 1 - \sum_{v,i} p_v \sum_{u \in C_{i,v}}\Pr_{\pi_{\bar{P}}}(\text{$u$ precedes all $w \in N_v(u) \setminus \{u\}$ in $\pi$}) \cdot \mathbb{I}(\text{$u$ threatens $B_{i,v}$}).
                \end{align*}
                where, again, $\mathbb{I}$ is the indicator function. We now verify properties (2)-(4). 
                \begin{enumerate}\setcounter{enumi}{1}
                        \item Computable: Clearly $C_{i,v}$ is deterministically computable in poly-time since we need only check if \Cref{eq:cutSet} holds for each vertex. Similarly $\mathbb{I}(\text{$u$ threatens $B_{i,v}$})$ for each $u \in C_{i,v}$ can be computed by checking if \Cref{eq:threaten} holds.  We can deterministically compute $\Pr_{\pi_{\bar{P}}}(\text{$u$ precedes all $w \in N_v(u) \setminus \{u\}$ in $\pi$})$ for each $u \in C_{i,v}$ as follows: if $u$ precedes all $w \in N_v(u) \cap \pi_P$ then this probability is $1$; if $u$ is preceded in $\pi_P$ by some $w \in N_v(u)$ then this probability is $0$; otherwise $\pi_P \cap N_v(u) = \emptyset$, meaning all nodes in $N_v(u)$'s order in $\pi$ are set by $\pi_{\bar{P}}$; in this case $u$ precedes all nodes in $N_v(u)\setminus \{u\}$ with probability exactly $\frac{1}{|N_v(u)|}$. 
                        \item Monotonicity is immediate by an averaging argument: in particular, $\hat{f}(\pi_P, \beta)$ is just an expectation taken over the randomness of $\pi_{\bar{P}}$ and so there must be some way to fix an element of $P$ to achieve the expectation.
                        \item Pessimism is immediate from the above discussion; in particular, as discussed above a ball $B_{i,v}$ is cut only if there is some $u \in C_{i,v}$ which threatens $B_{i,v}$ and which precedes all $w$ in $N_v(u) \setminus \{u\}$ in $\pi$; it follows by a union bound that $v$ fails to be $\alpha$-padded with probability at most 
                        \begin{align*}
                        \sum_{i} \sum_{u \in C_{i,v}}\Pr_{\pi_{\bar{P}}}(\text{$u$ precedes all $w \in N_v(u) \setminus \{u\}$ in $\pi$}) \cdot \mathbb{I}(\text{$u$ threatens $B_{i,v}$}).
                        \end{align*}

                \end{enumerate}
                
                Finally, we conclude property (1): that there is some $\beta \in R$ where $R$ is computable in deterministic poly-time and $\hat{f}(\pi_\emptyset, \beta) \geq .95$. Consider drawing a $\beta \in [\frac{1}{2}, 1]$ as in the FRT cutting scheme; we will argue that $\E_\beta \left[\hat{f}(\pi_\emptyset, \beta) \right] \geq .95$ and so there must be some $\beta$ for which $\hat{f}(\pi_\emptyset, \beta) \geq .95$.
                
                Letting $\pi$ be a uniformly random permutation, we have
                \begin{align*}
                \E_{\beta}\left[\hat{f}(\pi_\emptyset, \beta) \right] = 1- \sum_{v,i} p_v \sum_{u \in C_{i,v}}\Pr_{\pi}(\text{$u$ precedes all $w \in N_v(u) \setminus \{u\}$ in $\pi$}) \cdot \Pr_\beta(\text{$u$ threatens $B_{i,v}$}).
                \end{align*}
                
                If $u$ is the $s$th closest node to $v$ then we have that $\Pr_{\pi}(\text{$u$ precedes all $w \in N_v(u) \setminus \{u\}$ in $\pi$}) = \frac{1}{s}$. Moreover, $u$ threatens $B_{i,v}$ only if \Cref{eq:threaten} holds and since $\beta \cdot 2^{i-1}$ is distributed uniformly in $[2^{i-2}, 2^{i-1})$, this happens with probability $2^{i+1}\alpha/2^{i-2} = 8\alpha$. Next, we claim that for a fixed $v$, each $u$ occurs in at most $3$ of the $C_{i,v}$. In particular, notice that if $u$ is in $C_{i,v}$ and $C_{i', v}$ then we know that $2^{i-2}-2^i \alpha \leq d(u,v) \leq 2^{i'-1} + 2^{i'}\alpha$ which for $\alpha \leq \frac{1}{8}$ (which we may assume since $\alpha = \Omega(\frac{1}{\log n})$) implies $i < i' + 3$. Combining these facts with the fact that $H_n := \sum_{i=1}^n \frac{1}{i} \leq O(\log n)$ we get
                \begin{align*}
                 \E_{\beta}\left[\hat{f}(\pi_\emptyset, \beta) \right] \geq 1 - O(\alpha \log n).
                 \end{align*}
                 and since $\alpha = \Omega(\frac{1}{\log n})$, by fixing the constant in the $\Omega(\frac{1}{\log n})$ to be sufficiently small we have $\E_{\beta}\left[\hat{f}(\pi_\emptyset, \beta) \right] \geq .95$ as desired
                
                Lastly, we define $R$ and argue that there must be some $\beta\in R$ such that $\hat{f}(\pi_\emptyset, \beta) \geq .95$. In particular, notice that since $\E_{\beta}\left[\hat{f}(\pi_\emptyset, \beta) \right] \geq .95$,  it suffices to argue that there are polynomially-many efficiently computable intervals which partition $[\frac{1}{2}, 1)$ such that any $\beta_1$ and $\beta_2$ in the same interval satisfy $\hat{f}(\pi_\emptyset, \beta_1) = \hat{f}(\pi_\emptyset, \beta_2)$; letting $R$ take an arbitrary element from each such interval will give the desired result.
                
                Notice that $\hat{f}(\pi_\emptyset, \beta_1) \neq \hat{f}(\pi_\emptyset, \beta_2)$ only if there is some $i,v$ and $u$ such that $u$ threatens $B_{i,v}$ with $\beta$ set to $\beta_1$ but does not threaten $B_{i,v}$ with $\beta$ set to $\beta_2$. By definition of what it means to threaten, we have
                \begin{align*}
                d(u,v) - 2^i \alpha < \beta_1 \cdot 2^{i-1} < d(u,v) + 2^i \alpha
                \end{align*}
                but either $d(u,v) - 2^i \alpha \geq \beta_2 \cdot 2^{i-1}$ or $\beta_2 \cdot 2^{i-1} \geq d(u,v) + 2^i \alpha$. We then have either 
                \begin{align}\label{eq:changeOne}
                \beta_2 \leq d(u,v)\cdot2^{1-i} - 2\alpha < \beta_1 
                \end{align}
                or
                \begin{align}\label{eq:changeTwo}
                \beta_1  < d(u,v) \cdot 2^{1-i} + 2 \alpha \leq \beta_2.
                \end{align}
                
                With Equations \ref{eq:changeOne} and \ref{eq:changeTwo} in mind, we define $R_l := \{ d(u,v) \cdot 2^{1-i} + 2 \alpha: u, v \in V,  i \in [h] \}$ to be all the lower thresholds of when a change in $\beta$ affects $\hat{f}$ and define $R_u := \{ d(u,v)\cdot2^{1-i} - 2\alpha: u, v \in V,  i \in [h] \}$  to be all such upper thresholds. Let $t^{(l)}$ be the $l$th largest element of $(R_l \cup R_r) \cap [\frac{1}{2}, 1)$ and let $R$ consist of one arbitrary element from the interval between $t^{(l)}$ and $t^{(l+1)}$ for $l \geq 0$ where the interval includes $t^{(l)}$ only if $t^{(l)} \in R_l$ and $t^{(l+1)}$ only if $t^{(l+1)} \in R_u$; $t^{(0)} = \frac{1}{2}$ is always included and $t^{(|R|)} = 1$ is never included. By the above discussion every $\beta_1$ and $\beta_2$ which are in the same interval satisfy $\hat{f}(\pi_\emptyset, \beta_1) = \hat{f}(\pi_\emptyset, \beta_1)$; moreover, these intervals partition $[\frac{1}{2}, 1]$ by construction.
                
                We know $|R| = \poly(n)$ since $h \leq O(\log n)$ by our assumption that $\max_{u,v}d(u,v)$ is $\poly(n)$ and there are $n^2$ pairs $u,v$. Clearly $R$ is computable in deterministic poly-time. Thus, by the above discussion $R$ must contain some $\beta$ such that $\hat{f}(\pi_\emptyset, \beta) \geq .95$.
        \end{proof}
        
        We now formalize our node-weighted derandomization.
        
        \begin{lemma}\label{lem:derandPadding}
                There is a deterministic algorithm which given metric $(V, d)$ and a distribution $\{p_v\}_v$ over nodes returns a hierarchical decomposition $\mcH$ in which at least a $.95$ fraction of nodes are $\Omega(\frac{1}{\log n})$-padded by weight; i.e. 
                \begin{align*}
                \sum_v p_v \cdot \mathbb{I}\left(\text{$v$ is $\Omega\left(\frac{1}{\log n}\right)$-padded in $\mcH$}\right)\geq .95.
                \end{align*}
        \end{lemma}
        \begin{proof}
                
                Our derandomization algorithm is as follows. First, choose the $\beta \in R$ which maximizes $\hat{f}(\pi_\emptyset, \beta)$. Call this $\beta^*$. Next, initially let $P = \emptyset$ and repeat the following until $P = V$: for $v \in \bar{P}$ we compute $f(\pi_{P \cup \{v\}}, \beta^*)$; we add to $P$ whichever $v$ maximizes $f(\pi_{P \cup \{v\}}, \beta^*)$. Lastly, we return $\mcH(\pi_V, \beta^*)$.
                
                By \Cref{lem:pessEst} we know that $\beta^*$ will satisfy $\hat{f}(\pi_\emptyset, \beta^*) \geq .95$. Moreover, since $\hat{f}$ is monotone by \Cref{lem:pessEst} we know that the $\pi_V$ we choose will satisfy $\hat{f}(\pi_V, \beta^*) \geq .95$. Lastly, since $\hat{f}$ is pessimistic, it follows that $f(\pi_V, \beta^*) \geq f(\pi_V, \beta^*) \geq .95$ and so $\mcH(\pi_V, \beta^*)$ is padded on a $.95$ fraction of nodes by weight as desired.
                
                The deterministic polynomial runtime of our algorithm is immediate from the deterministic poly-time computability of $\hat{f}$ and the fact that $R$ is computable in deterministic poly-time.
        \end{proof}

        Using the above node-weighted derandomization lemma gives our deterministic copy tree embedding construction. In particular, we run the following multiplicative-weights-type algorithm with $\epsilon = .01$ and set the number of iterations as $\tau:=4 \ln n / \epsilon^2$. In the following we let $p_v^{(t)} : = w^{(t)}_v / \sum_v w_v^{(t)}$ be the proportional share of $v$'s weight in iteration $t$.
        
        \begin{enumerate}
                \item Uniformly set the initial weights: $w_v^{(1)}=1$ for all $v \in V$.
                \item For $t \in [\tau]$:
                \begin{enumerate}
                        \item Run the algorithm given in \Cref{lem:derandPadding} using distribution $p^{(t)}$ and let $\mcH_t$ be the resulting hierarchical decomposition.
                        \item \textbf{Set mistakes:} For each vertex $v$ which is $\Omega(\frac{1}{\log n})$-padded in $\mcH_t$ let $m_v^{(t)} = 1$. Let $m_v^{(t)} = 0$ for all other $v$.
                        \item \textbf{Update weights:} for all $v \in V$, let $w_v^{(t+1)} \gets \exp(-\epsilon m_v^{(t)}) \cdot w_v^{(t)}$.
                \end{enumerate}
                \item Return $(\mcH_t)_{t=1}^\tau$.
        \end{enumerate}
        
        We state a well-known fact regarding multiplicative weights in our notation. Readers familiar with multiplicative weights may recognize this as the fact that the expected performance of mutliplicative weights over logarithmically-many rounds is competitive with every expert.
        
        \begin{lemma}[\citet{arora2012multiplicative}]\label{lem:MWAvg}
                The above algorithm guarantees that for any $v \in V$ we have
                \begin{align*}
                \frac{1}{T} \sum_{t \leq \tau} p^{(t)} \cdot m^{(t)} \leq \epsilon + \frac{1}{T} \sum_{t \leq \tau} m_v^{(t)}
                \end{align*}
                where $p^{(t)} \cdot m^{(t)} := \sum_v p^{(t)}_v m_v^{(t)}$ is the usual inner product.
        \end{lemma}

        Using this fact we conclude that we are able to produce a good set of hierarchical decompositions.
        \begin{lemma}\label{lem:HDConstr}
            The above algorithm returns a collection of hierarchical decompositions $\{\mcH_t\}_{t=1}^\tau$ where $\tau = \Theta(\log n)$ and every vertex is $\Omega(1 / \log n)$-padded in at least $.9\tau$ of the decompositions.
        \end{lemma}
        \begin{proof}
                Since $\tau:=4 \ln n / \epsilon^2$ we know that $\tau = \Theta(\log n)$. 
                
                We need only argue, then, that each node is padded in at least a $.9$ fraction of the $\tau$ total $\mcH_t$. Let \[f_v := \frac{1}{\tau} \sum_{t \leq \tau} \mathbb{I}\left(\text{$v$ is  $\Omega\left(\frac{1}{\log n}\right)$-padded in $\mcH_t$}\right)\] be the fraction of the decompositions in which $v$ is padded. Consider a fixed node $v$. By \Cref{lem:MWAvg} we know that 
                \begin{align}\label{eq:mwguar}
                \frac{1}{\tau} \sum_{t \leq \tau} p^{(t)} \cdot m^{(t)} \leq \epsilon + \frac{1}{\tau} \sum_{t \leq \tau} m_v^{(t)}
                \end{align}
                
                By definition of $m_v^{(t)}$ we have that the right hand side of \Cref{eq:mwguar} is $\epsilon + f_v$. On the other hand, by how we set $m^{(t)}$, the left hand side of \Cref{eq:mwguar} is $\frac{1}{\tau}\sum_t\sum_v p_v^{(t)} \cdot \mathbb{I}(\text{$v$ is $\Omega(\frac{1}{\log n})$-padded in $\mcH$})$ which by \Cref{lem:derandPadding} is at least $.95$. Combining these facts we have $.95 \leq \epsilon + f_v$ and so by our choice of $\epsilon$ we know $.9 \leq f_v$ as desired.
        \end{proof}
    
        Combining \Cref{lem:HDConstr} with \Cref{lem:RTEFromHDs} gives \Cref{thm:repTreeConst}.

\subsection{Construction 2: Merging FRT Support}\label{sec:FRTSup}

In this section we observe that the support of the FRT distribution can be merged to produce copy tree embeddings with cost stretch $O(\log n)$ and copy number $O(n \log n)$. In particular, we rely on the known fact that one can make the size of the support of the FRT distribution $O(n \log n)$ and compute said support in deterministic poly-time, as summarized in the following theorem.

\begin{theorem}[\cite{charikar1998approximating,fakcharoenphol2004tight,konjevod2001approximating}]\label{thm:charikFRTSup}
    Given a weighted graph $G = (V,E, w)$ and root $r \in V$, there exists a distribution $\mcD$ being supported over $O(n \log n)$ well-separated weighted trees on $V$ rooted at $r$ where for any $u,v \in V$ we have $\E_{T \sim \mcD}[d_T(u,v)] \leq O(\log n \cdot d_G(u, v))$ and for every $T$ in the support of $\mcD$ we have $d_G(u,v) \leq d_T(u, v)$. Also, (the support and probabilities of) $\mcD$ can be computed in deterministic poly-time.
\end{theorem}

Merging the trees of this distribution and some simple probabilistic method arguments give a copy tree embedding with the desired properties.

\frtSupp*
\begin{proof}
    Let $T_1, \ldots, T_k$ with $k = O(n \log n)$ be the trees in the support of the distribution $\mcD$ as guaranteed by \Cref{thm:charikFRTSup}. Then, we let $T$ be the result of identifying each copy of $r$ as the same vertex in each $T_i$ (but not identifying copies of other vertices in $V$ as the same vertex); that is, $|V(T)| = k \cdot n - (k-1)$. $T$'s weight function is inherited from each $T_i$ in the natural way. Similarly, we let $\phi(v)$ be the set containing each copy of $v$ in each of the $T_i$. It is easy to verify that $\phi$ is indeed a copy mapping. Also, note that $\phi(v)$ is computable in deterministic poly-time, our copy number is $O(n \log n)$ by construction and that $T$ is well-separated since each $T_i$ is well-separated.
    
    We next specify $\pi_{G \to T}(F)$ for a fixed $F$. For tree $T_i$, let $T_i' \subseteq T_i$ be the subgraph of $T_i$ which contains the unique tree path between $u$ and $v$ iff $\{u, v\} \in F$. By \Cref{thm:charikFRTSup} we know that $\E_{T_i \sim D}[w_{T_i}(T_i')] \leq O(\log n \cdot w_G(H))$ and so there must be some $j$ such that $w_{T_j}(T_j') \leq O(\log n \cdot w_G(F))$. Thus, we let $\pi_{G \to T}(F) := T_j'$. We argue that $\pi_{G \to T}$ requires the stated connectivity properties. In particular, notice that by construction we have that if $u$ and $v$ are connected in $F$ then they will have some copy connected in $\pi_{G \to T}(F)$: if $u$ and $v$ are connected in $F$ by path $(v_1, v_2, \ldots)$ then the path in $T_j$ which connects the copy of $v_l$ and the copy of $v_{l+1}$ is contained in $\pi_{G \to T}(F)$ and the concatenation of these paths for all $l$ connects the copies of $u$ and $v$ contained in $T_j$. Moreover, notice that $\pi_{G \to T}(F)$ satisfies the required cost preservation properties since $w_T(\pi_{G \to T}(F)) = w_{T_j}(T_j') \leq O(\log n \cdot w_G(F))$ by construction. 
    
    Lastly, we specify $\pi_{T \to G}(F')$. We let $\pi_{T \to G}(F')$ be the graph induced by $\{P_{uv} : \{u',v' \} \in F' \}$ where $P_{uv}$ is an arbitrary shortest path in $G$ between $u$ and $v$ and $u'$ and $v'$ are copies of $u$ and $v$. We first verify the required connectivity preservation properties: if $u'$ and $v'$ are connected in $F'$ by path $(v_1', v_2' \ldots)$ then we know that $v_l$ and $v_{l+1}$ will be connected in $\pi_{T \to G}(F')$ for every $l$ by $P_{v_{l}v_{l+1}}$ where $v_i'$ is some copy of $v_i$. Thus, $u$ and $v$ will be connected in $\pi_{T \to G}(F')$. We next verify the required cost-preservation properties. By \Cref{thm:charikFRTSup} we have for every $i$ that $w_{T_i}(e') \geq w_G(P_{uv})$ for each $e' = \{u', v'\} \in T_i$. Thus, $w_T(F') = \sum_{e' \in F'} w_T(e') \geq \sum_{\{u',v'\} \in F'} w_G(P_{uv}) \geq w_G(\pi_{T \to G}(F'))$ where we have again used $u$ and $v$ to stand for the $\phi^{-1}(u')$ and $\phi^{-1}(v')$ respectively. Lastly, we note that $\pi_{T \to G}(F')$ is trivially computable in deterministic poly-time.
\end{proof}

\section{Deterministic Online Group Steiner Tree/Forest Reductions}\label{sec:detOGST}

In this section we prove that the guarantees of our copy tree embeddings are sufficient to generalize any deterministic algorithm for online group Steiner tree on trees to general graphs, thereby reducing an open question posed by \citet{alon2006general} to its tree case. We show that a similar result holds for the online group Steiner forest problem which generalizes online group Steiner tree.

In general, mapping an instance of a problem $P$ onto an equivalent instance $I'$ on the copy tree embedding often results that $I'$ is not an instance of the same problem $P$. However, group Steiner tree (resp., forest) problems have the notable property that mapping them onto a copy tree embedding simply results in another instance of the group Steiner tree (resp., forest) problem, this time on a tree. This property, albeit somewhat hidden in the proof, is the main reason why copy tree embeddings are well suited for these two problems.

Because past work on group Steiner and group Steiner forest have stated runtimes and approximation guarantees as functions of the maximum group size and number of groups rather than just $n$---see e.g.\ \cite{garg2000polylogarithmic,barta2020online}---we will give our results in the same generality with respect to these parameters.

\subsection{Deterministic Online Group Steiner Tree}\label{sec:det-online-group-steiner-tree}

We begin with our results for online group Steiner tree.

\textbf{Offline Group Steiner Tree:} In the group Steiner Tree problem we are given a weighted graph $G = (V, E, w)$ as well as pairwise disjoint groups $g_1, g_2, \ldots, g_k \subseteq V$ and root $r \in V$. We let $N:= \max_i |g_i|$ be the maximum group size. Our goal is to find a (connected) tree $T$ rooted at $r$ which is a subgraph of $G$ and satisfies $T \cap g_i \neq \emptyset$ for every $i$. We wish to minimize our cost, $w(T) := \sum_{e \in E(T)} w(e)$.\footnote{The assumption that the tree is rooted in group Steiner tree is without loss of generality as we may always brute-force search over a root. Similarly, the assumption that all groups are pairwise disjoint is without loss of generality since if $v$ is in groups $\{g_1, g_2, \ldots \}$ then we can remove $v$ from all groups and add vertices $v_1, v_2, \ldots$ to $G$ which are connected only to $v$ so that $v_i \in g_i$ and $w((v, v_i)) = 0$ for all $i$.}

\textbf{Online Group Steiner Tree:} Online group Steiner tree is the same as offline group Steiner tree but where our solution need not be a tree and groups are revealed in time steps $t = 1, 2, \ldots$. That is, in time step $t$ an adversary reveals a new group $g_t$ and the algorithm must maintain a solution $T_t$ where: (1) $T_{t-1} \subseteq T_{t}$; (2) $T_t$ is feasible for the group Steiner tree problem on groups $g_1, \ldots g_t$ and; (3) $T_t$ is competitive with the optimal offline solution for this problem where the competitive ratio of our algorithm is $\max_t w(T_t)/ \OPT_t$ where $\OPT_t$ is the cost of the optimal offline group Steiner tree solution on the first $t$ groups. Here, we will let $k$ be the number of possible groups revealed by the adversary.

\begin{theorem}
    If there exists:
    \begin{enumerate}\label{thm:detGST}
        \item A poly-time deterministic algorithm to compute an efficient, well-separated $\alpha$-approximate copy tree embedding with copy number $\chi$ and;
        \item A poly-time $f(n, N, k)$-competitive deterministic algorithm for online group Steiner tree on well-separated trees
    \end{enumerate}
    then there exists an $(\alpha \cdot f(\chi n, \chi N, k))$-competitive deterministic algorithm for group Steiner tree (on general graphs).
\end{theorem}
\begin{proof}
    We will use our copy tree embedding to produce a single tree on which we must solve deterministic online group Steiner tree.
    
    In particular, consider an instance of online group Steiner tree on weighted graph $G = (V, E, w)$ with root $r$. Then, we first compute a copy tree embedding $(T, \phi, \pi_{G \to T}, \pi_{T \to G})$ deterministically with respect to $G$ and $r$ as we assumed is possible by assumption. Next, given an instance $I_t$ of group Steiner tree on $G$ with groups $g_1, \ldots g_t$, we let $I_t'$ be the instance of group Steiner tree on $T$ with groups $\phi(g_1), \ldots \phi(g_t)$ and root $r' := \phi(r)$ where we have used the notation $\phi(g_i) := \bigcup_{v \in g_i} \phi(g_i)$. Then, if the adversary has required that we solve instance $I_t$ in time step $t$, then we require that our deterministic algorithm for online group Steiner tree on trees solves $I_t'$ in time step $t$ and we let $H_t'$ be the solution returned by our algorithm for $I_t'$. Lastly, we return as our solution for $I_t$ in time step $t$ the set $H_t := \pi_{T \to G}(H_t')$.
    
    Let us verify that the resulting algorithm is indeed feasible and of the appropriate cost. 
    
    First, we have that $H_t \subseteq H_{t+1}$ for every $t$ since $H_t' \subseteq H_{t+1}'$ because our algorithm for trees returns a feasible solution for its online problem and $\pi_{T \to G}$ is monotone by definition of a copy tree embedding. Moreover, we claim that $H_t$ connects at least one vertex from each $g_i$ to $r$ for $i \leq t$ and every $t$. To see this, notice that $H_t'$ connects at least one vertex from $\phi(g_t)$ to $r' = \phi(r)$ in $t$ since it is a feasible solution for $I_t'$ and so at least one copy of a vertex in $g_t$; by the connectivity preservation properties of a copy tree it follows that at least one vertex from $g_t$ is connected to $r$. Thus, our solution is indeed feasible in each time step.
    
    Next, we verify the cost of our solution. Let $\OPT_t'$ be the cost of the optimal solution to $I_t'$ and let $n'$ and $N'$ be the number of vertices and maximum size of a group in $I_t'$ for any $t$. By our assumption on the cost of the algorithm we run on $T$ and since $n' \leq \chi n$ and $N' \leq \chi N $ by definition of copy number, we know that 
    \begin{align*}
    w_T(H_t') \leq {\OPT}_t' \cdot f(n', N', k) = {\OPT}_t' \cdot f(\chi n, \chi N, k).
    \end{align*}

    Next, let $H^*_t$ be the optimal solution to $I_t$. We claim that $\pi_{G \to T}(H^*_t)$ is feasible for $I_t'$. This follows because $H^*_t$ connects a vertex from $g_1, \ldots, g_t$ to $r$ and so by the connectivity preservation property of copy tree embeddings we know that some vertex from each of $\phi(g_1), \ldots, \phi(g_t)$ is connected to $r' = \phi(r)$. Applying this feasibility of $\pi_{G \to T}(H_t^*)$ and the cost preservation property of our copy tree embedding, it follows that $\OPT_t' \leq w_T(\pi_{G \to T}(H_t^*)) \leq \alpha \cdot w_G(H_t^*) = \alpha \cdot \OPT_t$.
    
    Similarly, we know by the cost preservation property of our copy tree embedding that $w_G(\pi_{T \to G}(H_t')) \leq w_T(H_t')$. Combining these observations we have
    \begin{align*}
        w_G(\pi_{T \to G}(H_t')) \leq w_T(H_t') \leq {\OPT}_t' \cdot f(\chi n, \chi N, k) \leq {\OPT}_t \cdot \alpha \cdot f(\chi n, \chi N, k),
    \end{align*}
    thereby showing that our solution is within the required cost bound.
\end{proof}
Plugging in our first construction (\Cref{thm:repTreeConst}) or our second construction (\Cref{thm:frtSupp}) of a copy tree embedding immediately gives the follow corollary.

\begin{corollary}\label{cor:GST}
    If there is an $f(n, N, k)$-competitive deterministic algorithm for online group Steiner tree on well-separated trees then there are $O(\log n \cdot f(O(n^2 \log n), O(n N), k ))$ and $O(\log ^ 2 n \cdot f(O(n \log n), O(N \log n), k))$-competitive deterministic algorithms for online group Steiner tree (on general graphs).
\end{corollary}

\subsection{Deterministic Online Group Steiner Forest}\label{sec:det-online-group-steiner-forest}

In this section we show a black-box reduction from the poly-log-approximate online deterministic group Steiner forest in a general graph $G$ to poly-log-approximate online deterministic group Steiner forest when the underlying graph is a tree. A formal definition of the problem follows.

\textbf{Offline Group Steiner Forest:} In the group Steiner forest problem we are given a weighted graph $G = (V, E, w)$ as well as pairs of subsets of nodes $(A_1, B_1), (A_2, B_2), \ldots, (A_k, B_k)$ where $A_i, B_i \subseteq V$. Our goal is to find a forest $F$ which is a subgraph of $G$ and in which for each $i$ there is an $a_i \in A_i$ and $b_i \in B_i$ such that $a_i$ and $b_i$ are connected in $F$. We wish to minimize our cost, $w(F) := \sum_{e \in E(F)} w(e)$. We let $N := \max_i \max(|A_i|, |B_i|)$ be the maximum subset size.

\textbf{Online Group Steiner Forest:} Online group Steiner forest is the same as group Steiner forest but each pair $(A_t, B_t)$ is revealed at time step $t = 1,2, \ldots$ by an adversary and in each time step $t$ we must maintain a forest $F_t$ which is feasible for pairs $(A_1, B_1), \ldots (A_t, B_t)$ so that $F_{t-1} \subseteq F_t$. The competitive ratio of an online algorithm with solution $\{F_t\}_t$ is $\max_t w(F_t) / \OPT_t$ where $\OPT_t$ is the optimal offline solution for the group Steiner forest problem we must solve in time step $t$. For the online problem let $k$ be the number of possible pairs revealed by the adversary.

Note that the group Steiner forest directly generalizes group Steiner tree since a tree instance on a weighted graph $G$ with root $r \in V(G)$ can be reduced to an equivalent forest instance on the same graph $G$ by mapping each group $g$ to the pair $(\{r\}, g)$. This reductions is valid in both the offline and online setting (also in the later defined, demand-robust, setting).

We now show that a deterministic algorithm for online group Steiner forest on trees gives a deterministic algorithm for online group Steiner forest on general graphs up to small losses. These results and the corresponding proofs will be quite similar to those of the previous section so we defer a full proof to the appendix.

\begin{restatable}{theorem}{onGSF}\label{thm:onGSF}
    If there exists:
    \begin{enumerate}
        \item A poly-time deterministic algorithm to compute an efficient, well-separated $\alpha$-approximate copy tree embedding with copy number $\chi$ and;
        \item A poly-time $f(n, N, k)$-competitive deterministic algorithm for online group Steiner forest on well-separated trees
    \end{enumerate}
    then there exists an $(\alpha \cdot f(\chi n, \chi N, k))$-competitive deterministic algorithm for group Steiner forest (on general graphs).
\end{restatable}
\begin{proof}[Proof Sketch]
    The properties of a copy tree embedding show that an instance of group Steiner forest on a tree exactly map to an instance of group Steiner forest on our copy tree. In particular, if we must connect $(A_i, B_i)$ in the general graph then we can just connect $(\bigcup_{v \in A_i} \phi(v), \bigcup_{v \in B_i}\phi(v))$ on our copy tree and map back the solution with $\pi_{T \to G}$. The full proof is available in \Cref{sec:defProof}.
\end{proof}

Plugging in our first construction (\Cref{thm:repTreeConst}) or our second construction (\Cref{thm:frtSupp}) of a copy tree embedding immediately gives the follow corollary.

\begin{corollary}\label{cor:GSF}
    If there is an $f(n, N, k)$-competitive deterministic algorithm for online group Steiner forest on well-separated trees then there are $O(\log n \cdot f(O(n^2 \log n), O(n N), k ))$ and $O(\log ^ 2 n \cdot f(O(n \log n), O(N \log n), k ))$-competitive deterministic algorithms for online group Steiner forest (on general graphs).
\end{corollary}

Lastly, we note that \Cref{thm:GSTAndFor} follows immediately from \Cref{cor:GST} and \Cref{cor:GSF}.

\section{Online Partial Group Steiner Tree}\label{sec:onPGST}
In this section we give a deterministic bicriteria algorithm for the online partial group Steiner tree problem which is the same as online group Steiner tree but where we must connect at least $\frac{1}{2}$ of all vertices from each group to the root. The algorithm is bicriteria in the sense that it relaxes both the $1/2$-connectivity guarantee and the cost.


As mentioned in the introduction, this problem generalizes group Steiner tree. In particular, we can reduce an instance of group Steiner tree on weighted graph $G = (V, E, w)$ with groups $\{g_i\}_i$ and root $r$ to an instance of partial group Steiner tree as follows. For each group $g_i$ we add $|g_i|-1$ new vertices with an edge of cost $0$ attached to $r$. Our partial group Steiner tree problem will be on the resulting graph with root $r$ and groups $\{g_i'\}_i$ where $g_i'$ consists of $g_i$ along with its corresponding $|g_i|-1$ dummy nodes. Any partial group Steiner tree solution on the resulting graph will connect at least one vertex from each $g_i$ to $r$. Conversely, by connecting all of the dummy nodes we added to our graph to $r$ by their cost $0$ edges, it is easy to see that a solution for group Steiner tree on the input graph exactly corresponds to a solution for our partial group Steiner tree instance.\footnote{As a minor techincal caveat: we have assumed that edge weights are at least $1$ throughout this paper; it is easy to see that by scaling weights up by a polynomial factor and then using weight $1$ edges instead of weight $0$ edges this reduction still works.}

Moreover, it is also easy to see that any deterministic bicriteria algorithm for online partial group Steiner tree also gives a poly-log-competitive deterministic (unicriteria) algorithm for online (non-group) Steiner tree. In particular, given an instance of Steiner tree on weighted graph $G = (V, E, w)$ with root $r$ where we must connect terminals $A \subseteq V$ to $r$, it suffices to solve the partial group Steiner tree problem where each vertex in $A$ is in a singleton group with any constant bicriteria relaxation. This is because connecting any $c > 0$ fraction of each group to $r$ will connect at least one vertex to $r$ by the integrality of the number of connected vertices. Thus, our result generalizes the fact that deterministic poly-log approximations are known for online (non-group) Steiner tree \cite{imase1991dynamic}. However, we do note that our (deterministic) poly-log-approximate bicriteria online partial group Steiner tree algorithm does not imply there is a (deterministic) poly-log-approximate online (non-partial) group Steiner tree algorithm (due to the nature of the bicriteria guarantee).

Mapping the online partial group Steiner tree problem on a copy tree embedding yields a problem that is slightly different than the original one (unlike, e.g., group Steiner tree). Our result will, therefore, be for a problem which generalizes partial group Steiner tree: we give a deterministic $\tilde{O}(\max_i \frac{|g_i|}{f_i \cdot \epsilon})$ bicriteria approximation for what we call the $f$-partial group Steiner tree problem which requires connecting at least $f_i$ vertices from group $g_i$ to the root; our bicriteria algorithm will connect at least $f_i \cdot (1 - \epsilon)$ vertices from each group for any specified input $\epsilon > 0$. It will be convenient for us to consider this problem as opposed to partial group Steiner tree since group Steiner tree is just $f$-partial group Steiner tree with $f_i = 1$ for all $i$. Thus, as an immediate corollary of our algorithm we will be able to give a deterministic algorithm for online group Steiner tree with a competitive ratio that is linear in the maximum group size.

\textbf{Offline $f$-Partial Group Steiner:} In the $f$-partial group Steiner Tree problem we are given a weighted graph $G = (V, E, w)$ as well as pairwise disjoint groups $g_1, g_2, \ldots, g_k \subseteq V$, desired connected vertices $1 \leq f_i \leq |g_i|$ for each group $g_i$ and root $r \in V$. 
Our goal is to find a tree $T$ rooted at $r$ which is a subgraph of $G$ and satisfies $|T \cap g_i| \geq f_i$ for every $i$. We wish to minimize our cost, $w(T) := \sum_{e \in E(T)} w(e)$.\footnote{As with group Steiner tree the assumption that the tree is rooted and that the groups are pairwise disjoint is without loss of generality.}

\textbf{Online $f$-Partial Group Steiner:} Online $f$-partial group Steiner tree is the same as offline partial group Steiner tree but where our solution need not be a tree and groups are revealed in time steps $t = 1, 2, \ldots$. That is, in time step $t$ an adversary reveals a new group $g_t$ and the algorithm must maintain a solution $T_t$ where: (1) $T_{t-1} \subseteq T_{t}$; (2) $T_t$ is feasible for the (offline) $f$-partial group Steiner tree problem on groups $g_1, \ldots g_t$ and; (3) $T_t$ is cost-competitive with the optimal offline solution for this problem where the cost-competitive ratio of our algorithm is $\max_t w(T_t)/ \OPT_t$ where $\OPT_t$ is the cost of the optimal offline $f$-partial group Steiner tree solution on the first $t$ groups. We will give a bicriteria approximation for online $f$-partial group Steiner tree; thus we say that an online solution is $\rho$-connection-competitive if for each $t$ we have $|T_t \cap g_i| \geq (f_i \cdot \rho)$ for every $i \leq t$.  

We note that the partial group Steiner tree problem as mentioned above is simply the special case of $f$-partial group Steiner tree but where $f_i = \frac{g_i}{2}$ for every $i$.

\subsection{Online $f$-Partial Group Steiner Tree on a Tree}
We begin by giving a bicriteria deterministic online algorithm for $f$-partial group Steiner tree on trees based on a ``water-filling'' approach. Informally, in iteration $t$ each unconnected vertex in each group will grow the solution towards the root at an equal rate until at least $f_i \cdot (1 - \epsilon)$ vertices in $g_t$ are connected to $r$.

\subsubsection{Problem}
More formally we will solve a problem which is a slight generalization of $f$-partial group Steiner tree on trees. We solve this problem on a tree rather than just $f$-partial group Steiner tree on a tree because, unlike group Steiner tree, the ``groupified'' version of $f$-partial group Steiner tree is not necessarily another instance of $f$-partial group Steiner tree. Roughly, instead of groups we now have groups of groups, hence we call this problem $2$-level $f$-partial group Steiner tree.

\textbf{Offline $2$-Level $f$-Partial Group Steiner Tree}: In $2$-level $f$-Partial Group Steiner tree we are given a weighted graph $G = (V, E, w)$, root $r \in V$ and groups of groups $\mcG_1, \ldots \mcG_k$ where $\mcG_i$ consists of groups $\{g_1^{(i)}, \ldots g_{k_i}^{(i)}\}$ where each $g_j^{(i)} \subseteq V$. We are also given connectivity requirements $f_1, \ldots, f_k$. Our goal is to compute a minimum-weight tree $T$ containing $r$ where for each $i \leq k$ we have $|\{g_j^{(i)} : g_j^{(i)} \cap T \neq \emptyset\}| \geq f_i$. We let $n_i := |\{v : \exists j \st v \in g_j^{(i)}\}|$. Notice that $f$-partial group Steiner tree is just 2-level $f$-partial group Steiner tree where each $g_i^{(j)}$ is a singleton set.

\textbf{Online $2$-Level $f$-Partial Group Steiner Tree}: Online $2$-level $f$-Partial Group Steiner tree is the same as the offline problem but where $\mcG_t$ is revealed in time step $t$ by an adversary. In particular, for each time step $t$ we must maintain a solution $T_t$ where: (1) $T_{t-1} \subseteq T_{t}$ for all $t$; (2) $T_t$ is feasible for the (offline) 2-level $f$-partial group Steiner tree problem on $\mcG_1, \ldots, \mcG_t$ with connectivity requirements $f_1, \ldots, f_t$ and; (3) $T_t$ is cost-competitive with the optimal offline solution for this problem where the cost-competitive ratio of our algorithm is $\max_t w(T_t)/ \OPT_t$ where $\OPT_t$ is the cost of the optimal offline 2-level $f$-partial group Steiner tree solution on the first $t$ groups of groups.

We will give a bicriteria approximation for online 2-level $f$-partial group Steiner tree on trees; thus we say that an online solution is $\rho$-connection-competitive if for each $t$ we have $|\{g_j^{(i)} : g_j^{(i)} \cap T \neq \emptyset\}| \geq \rho \cdot f_i$ for every $i \leq t$.  

\subsubsection{Algorithm}

We now formally describe our algorithm for 2-level $f$-partial group Steiner tree on weighted tree $T = (V, E, w)$ given an $\epsilon > 0$. We will maintain a fractional variable $0 \leq x_e \leq w_e$ for each edge indicating the extent to which we buy $e$ where our $x_e$s will be monotonically increasing as our algorithm runs. Say that an edge $e$ is saturated if $x_e = w_e$.

Let us describe how we update our solution in the $t$th time step. Let $T_t$ be the connected component of all saturated edges containing $r$. Then, we repeat the following until $|\{g_j^{(t)} : g_j^{(t)} \cap T_t \neq \emptyset\}| \geq f_t \cdot (1-\epsilon)$. Let $\mcG_t' := \{g_j^{(t)} \in \mcG_t: g_j^{(t)} \cap T_t = \emptyset\}$ be all groups in $\mcG_t$ not yet connected and let $g_t' := \bigcup_{S \in \mcG_t'}S$ be all vertices in a group which have not yet been connected to $r$. We say that $e$ is on the frontier for $v \in g_t'$ if it is the first edge on the path from $v$ to $r$ which is not saturated. Similarly, let $r_e$ be the number of vertices in $g_t'$ for which $e$ is on the frontier for $v$. Then, for each edge $e$ we increase $x_e$ by $r_e \cdot \delta$ where $\delta = \min_e (w_e-x_e)/r_e$. Our solution in the $t$th time step is $T_t$ once $|\{g_j^{(t)} : g_j^{(t)} \cap T_t \neq \emptyset\}| \geq (1-\epsilon) \cdot f_t$.

We illustrate one iteration of this algorithm in \Cref{fig:waterFill}

    \begin{figure}
    \centering
    \begin{subfigure}[b]{0.24\textwidth}
        \centering
        \includegraphics[width=\textwidth,trim=120mm 100mm 120mm 10mm, clip]{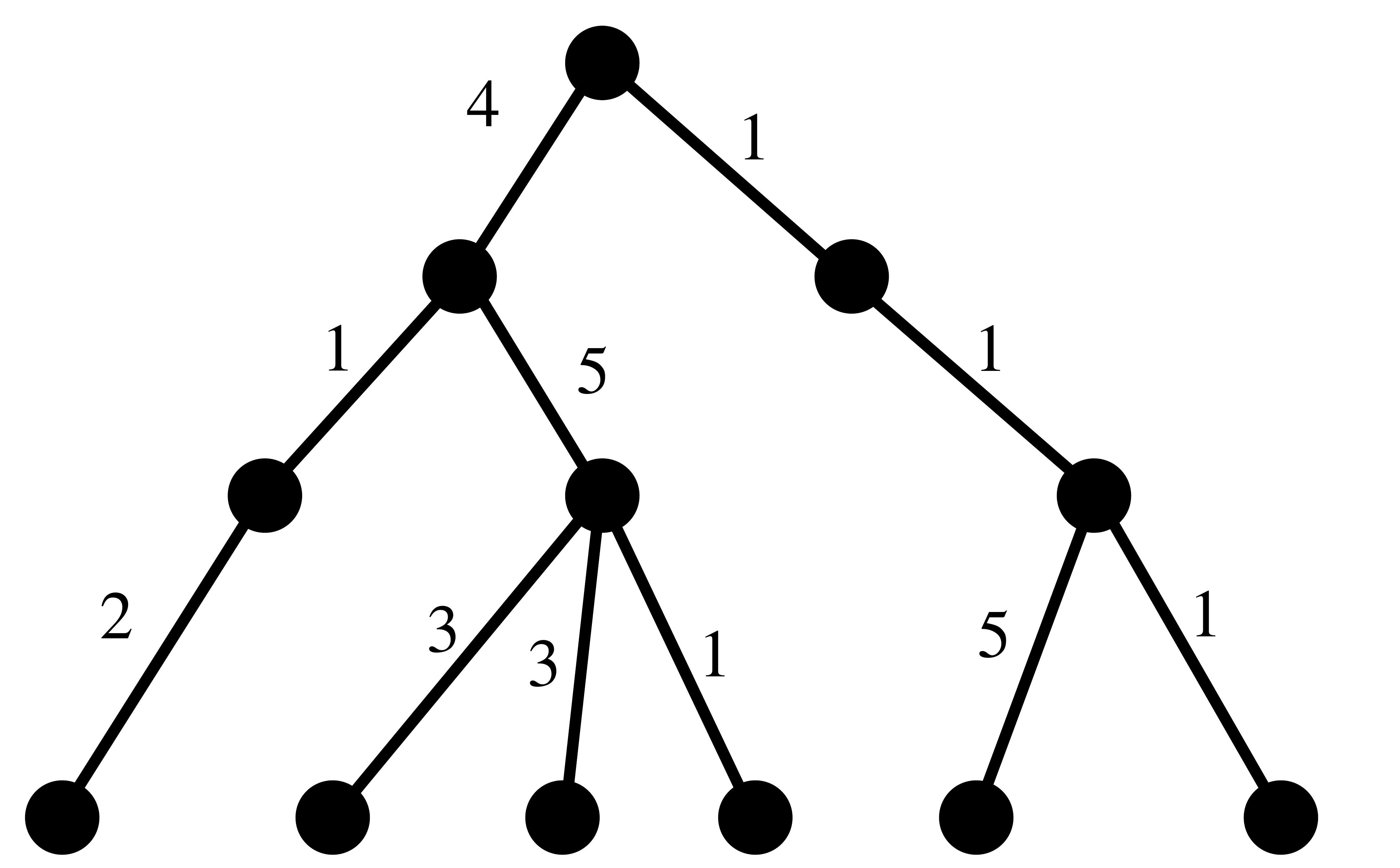}
        \caption{Graph $T$.}
    \end{subfigure}
    \hfill
    \begin{subfigure}[b]{0.24\textwidth}
        \centering
        \includegraphics[width=\textwidth,trim=120mm 100mm 120mm 10mm, clip]{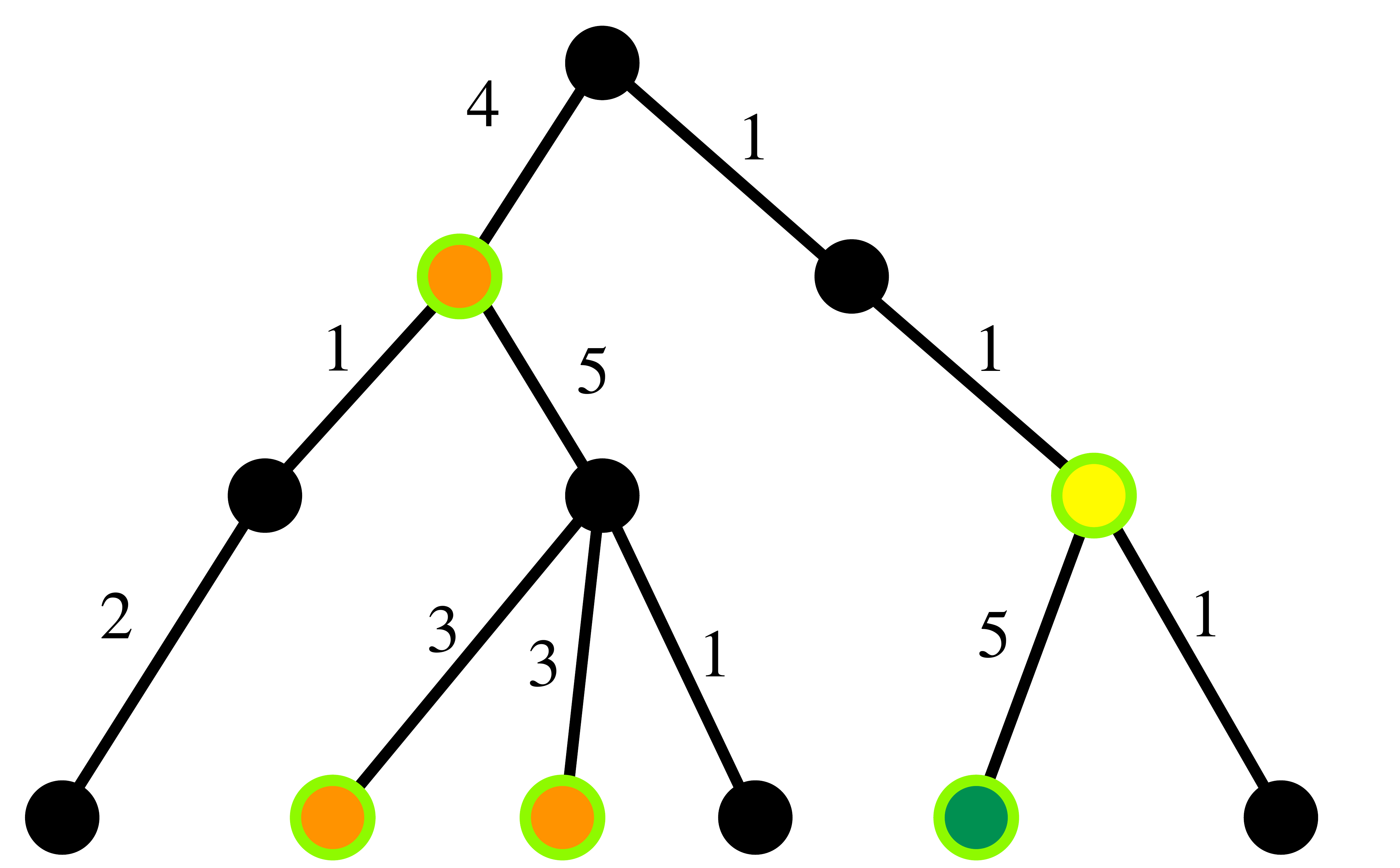}
        \caption{$\mcG_1$ arrives.}
    \end{subfigure}
    \hfill
    \begin{subfigure}[b]{0.24\textwidth}
        \centering
        \includegraphics[width=\textwidth,trim=120mm 100mm 120mm 10mm, clip]{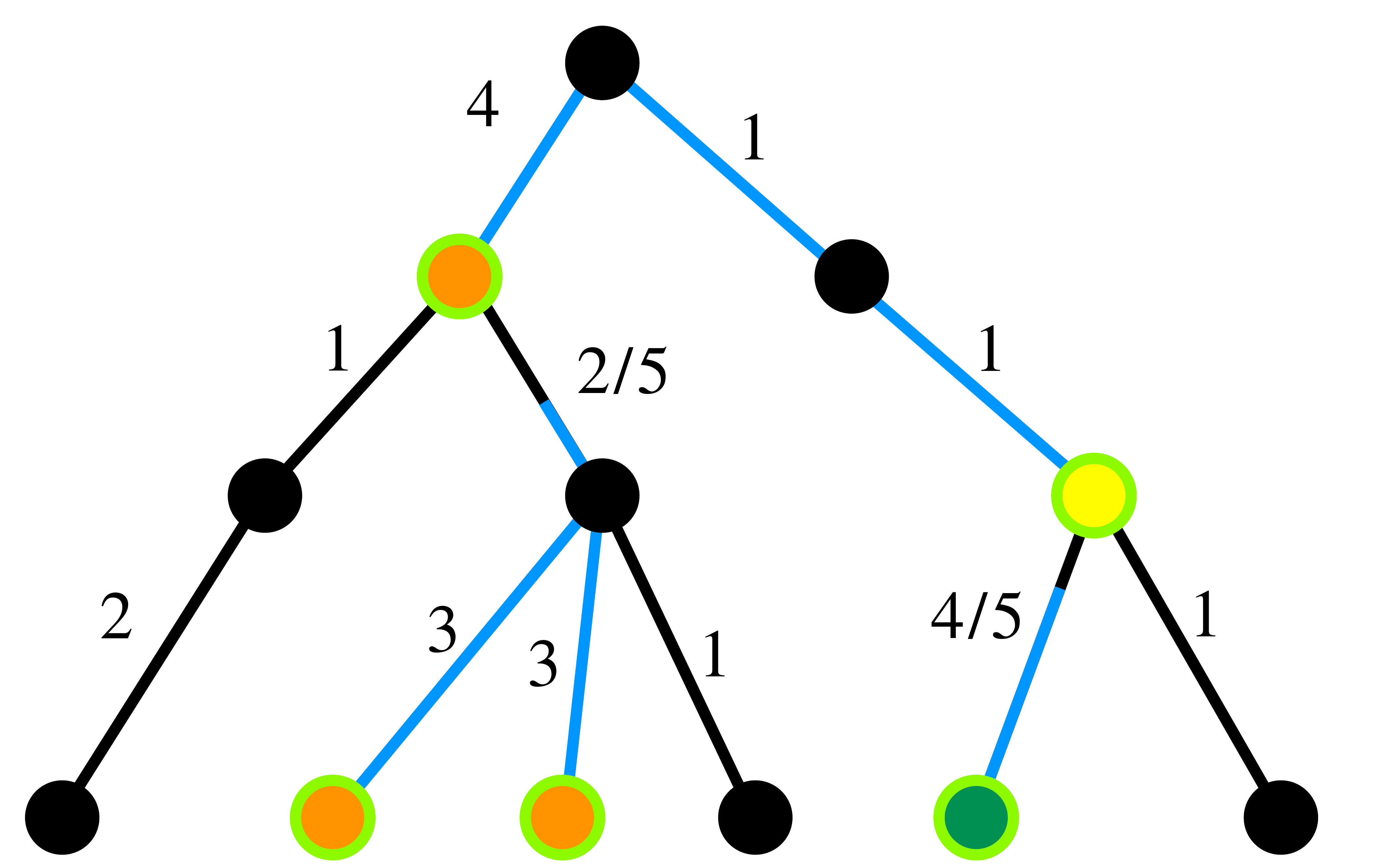}
        \caption{``Fill water.''}
    \end{subfigure}
    \hfill
    \begin{subfigure}[b]{0.24\textwidth}
        \centering
        \includegraphics[width=\textwidth,trim=120mm 100mm 120mm 10mm, clip]{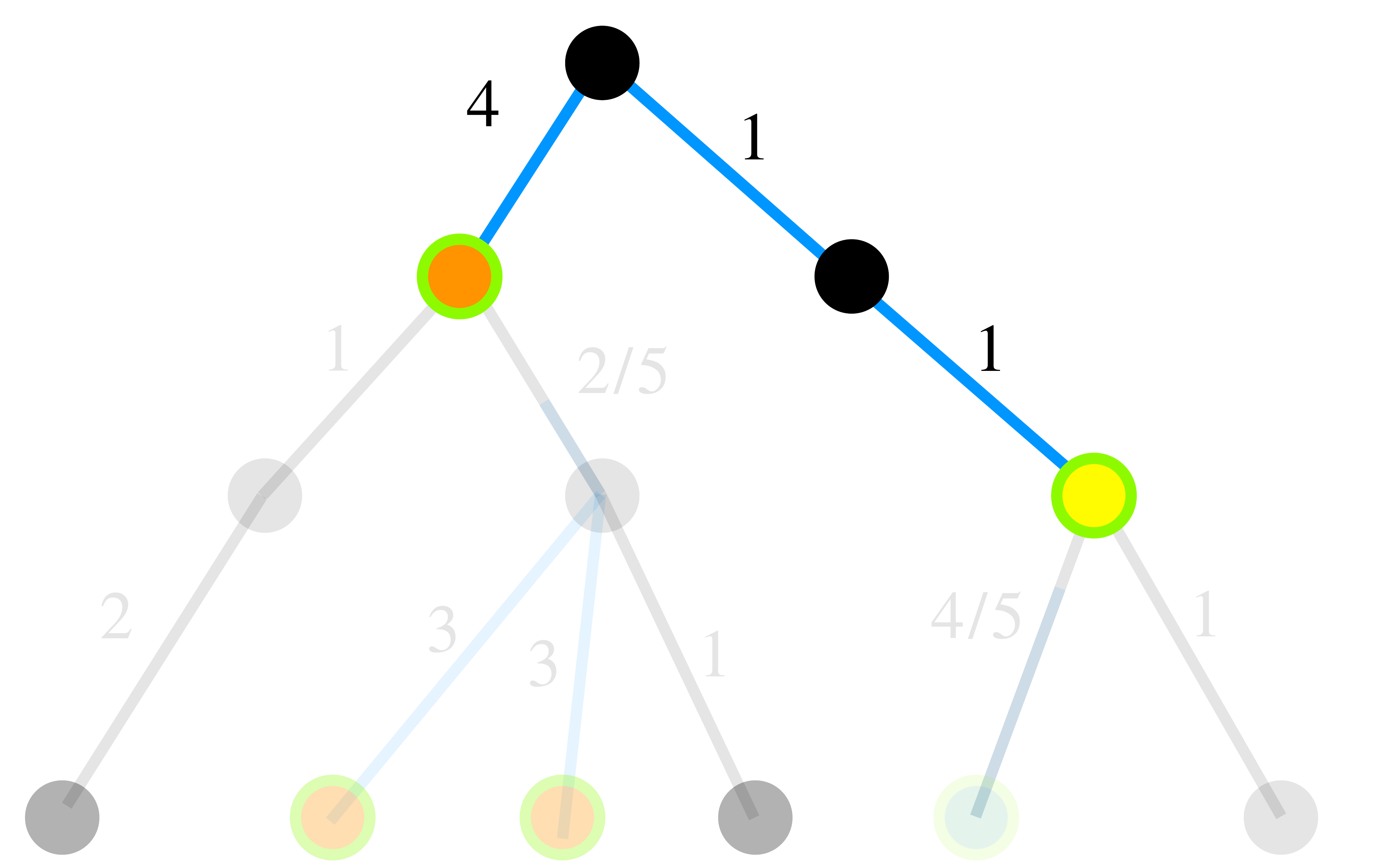}
        \caption{Choose solution.}
    \end{subfigure}
    \hfill
    \caption{Solution our algorithm gives after one group of groups, $\mcG_1$, is revealed where $f_1 = 2$. Nodes in groups in $\mcG_1$ outlined in green and nodes colored according to the group of $\mcG_1$ which contains them. Saturated edges given in blue and edges with $0 < x_e < w_e$ annoted with ``$x_e/w_e$''. All other edges labeled by $w_e$.}\label{fig:waterFill}
\end{figure}

\subsubsection{Analysis}

We proceed to analyze the above algorithm and give its properties.
\begin{theorem}\label{thm:partialGSTOnTrees}
    There is a deterministic poly-time algorithm for online $2$-level $f$-partial group Steiner tree on trees which is $\frac{1}{\epsilon} \cdot (\max_i \frac{n_i}{f_i})$-cost-competitive and $(1-\epsilon)$-connection-competitive.
\end{theorem}
\begin{proof}
    We begin by verifying that our algorithm returns a monotonically increasing and $(1-\epsilon)$-connection-competitive solution. First, notice that our solution is monotonically increasing since our $x_e$s are monotonically increasing and our solution only includes saturated edges. To see that our solution is $(1-\epsilon)$-connection-competitive notice that at least one new edge becomes saturated from each update to the $x_e$s (namely $\argmin_e (w_e-x_e)/r_e$) and since if all edges are saturated then $T_t = T$ which clearly satisfies $|\{g_j^{(t)} : g_j^{(t)} \cap T_t \neq \emptyset\}| \geq (1-\epsilon) \cdot f_t$, this process will eventually halt with a $(1-\epsilon)$-connection-competitive solution in the $t$th iteration. For the same reason our algorithm is deterministic poly-time.
    
    It remains to argue that our solution is $\frac{1}{\epsilon} \cdot (\max_i \frac{n_i}{f_i})$-cost-competitive. We will argue that we can uniquely charge each unit of increase of our $x_e$s to an appropriate cost portion of the optimal solution. Fix an iteration $t$. Next, let $\delta^{(i,j)}$ for $i \leq t$ be the value of $\delta$ in the $i$th iteration the $j$th time we increase the value of our $x_e$s. Similarly, let $\delta_x^{(i,j)}$ be the increase in $\sum_e x_e$ when we do so and let $\delta_y^{(i,j)}$ be the increase in $\sum_{e \in T_t^*}x_e$ where $T_t^*$ is the optimal offline solution to the $2$-level $f$-partial group Steiner problem we must solve in the $t$th iteration. Lastly, let $y := \sum_{i \leq t} \sum_j \delta_y^{(i,j)}$ be the value of $\sum_{e \in T_t^*}x_e$ at the end of the $t$th iteration; clearly we have $y \leq \OPT_t$. We claim that it suffices to show that for each $i \leq t$ and each $j$ that $\delta_x^{(i,j)} \leq \frac{1}{\epsilon} \delta_y^{(i,j)} \frac{n_i}{f_i}$ since it would follow that at the end of iteration $t$ we have that
    \begin{align*}
        w(T_t) \leq \sum_e x_e = \sum_{i \leq t} \sum_j \delta_x^{(i,j)} \leq \frac{1}{\epsilon}  \sum_{i \leq t} \sum_j \frac{n_i}{f_i} \delta_y^{(i,j)} \leq \frac{1}{\epsilon} \left(\max_i \frac{n_i}{f_i} \right) y \leq \frac{1}{\epsilon} \left(\max_i \frac{n_i}{f_i}\right) {\OPT}_t.
    \end{align*}
    
    We proceed to show that $\delta_x^{(i,j)} \leq \frac{1}{\epsilon} \delta_y^{(i,j)} \frac{n_i}{f_i}$ for each $i \leq t$ and $j$. We fix an $i$ and $j$ and for cleanliness of notation we will drop the dependence on $i$ and $j$ in our $\delta$s henceforth.
    
    First, notice that we have that
    \begin{align}\label{eq:xBound}
        \delta_x \leq n_i \cdot \delta
    \end{align}
    since each vertex $v \in g_i$ is uniquely responsible for up to a $\delta$ increase on $x_e$ where $e$ is the edge on $v$'s frontier.
    
    On the other hand, notice that if a group in $\mcG_i$ is connected to $r$ by $T_t^*$ but is not yet connected by $T_i$ then such a group uniquely contributes at least $\delta$ to $\delta_y$. Since $T_t^*$ connects at least $f_i$ groups in $\mcG_i$ to $r$ but at the moment of our increase $T_i$ connects at most $(1-\epsilon) \cdot f_i$, there are at least $\epsilon \cdot f_i$ such groups in $\mcG_i$ which are connected to $r$ by $T_t^*$ but not by $T_i$. Thus, we have that
    \begin{align}\label{eq:yBound}
        \delta_y \geq \epsilon \cdot f_i \cdot \delta
    \end{align} 
    Combining Equations \ref{eq:xBound} and \ref{eq:yBound} shows $\delta_x \leq \frac{1}{\epsilon} \delta_y \frac{n_i}{f_i}$ as required.
\end{proof}

\subsection{Online $f$-Partial Group Steiner Tree on General Graphs}

Next, we apply our first construction to give an algorithm for $f$-partial group Steiner tree on general graphs. Crucially, the following result relies on a single copy tree embedding with poly-logarithmic copy number, making our second construction unsuitable for this problem.

\begin{theorem}\label{thm:fPart}
    There is a deterministic poly-time algorithm for online $f$-partial group Steiner tree (on general graphs) which is $O(\frac{\log ^ 3 n}{\epsilon} \cdot \max_i \frac{|g_i|}{f_i})$-cost-competitive and $(1-\epsilon)$-connection-competitive.
\end{theorem}
\begin{proof}
    We will use our copy tree embedding to produce a single tree on which we must deterministically solve online $2$-level partial group Steiner tree. We will then apply the algorithm from \Cref{thm:partialGSTOnTrees} to solve online $2$-level partial group Steiner tree on this tree.
    
    More formally, consider an instance of online partial group Steiner tree on weighted graph $G = (V, E, w)$ with root $r$. Then, we first compute a copy tree embedding $(T, \phi, \pi_{G \to T}, \pi_{T \to G})$ deterministically with respect to $G$ and $r$ as in \Cref{thm:repTreeConst} with cost approximation $O(\log ^ 2 n)$ and copy number $O(\log n)$. Next, given our instance $I_t$ of partial group Steiner tree on $G$ with groups $g_1, \ldots g_t$ and connection requirements $f_1, \ldots, f_t$ we let $I_t'$ be the instance of $2$-level partial group Steiner tree on $T$ with groups of groups $\mcG_1, \ldots \mcG_t$ where $\mcG_i = \{\phi(v) : v \in g_i\}$, connection requirements $f_1, \ldots, f_t$ and root $\phi(r)$. Then if the adversary has required that we solve instance $I_t$ in time step $t$, then we require that the algorithm in \Cref{thm:partialGSTOnTrees} solves $I_t'$ in time step $t$ and we let $H_t'$ be the solution returned by our algorithm for $I_t'$. Lastly, we return as our solution for $I_t$ in time step $t$ the set $H_t := \pi_{T \to G}(H_t')$.
    
    Let us verify that the resulting algorithm is indeed feasible (i.e.\ monotone and $\frac{1}{2}$-connection-competitive) and of the appropriate cost. 
    
    First, we have that $H_t \subseteq H_{t+1}$ for every $t$ since $H_t' \subseteq H_{t+1}'$ because our algorithm for trees returns a feasible solution for its online problem and $\pi_{T \to G}$ is monotone by definition of a copy tree embedding. Moreover, we claim that $H_t$ connects at least $(1-\epsilon)\cdot f_i$ vertices from $g_i$ to $r$ for $i \leq t$ and every $t$. To see this, notice that there at least $(1-\epsilon)\cdot f_i$ groups from $\mcG_i$ containing a vertex connected to $r$ by $H_t'$. Since each such group consists of the copies of a distinct vertex, by the connectivity preservation properties of a copy tree it follows that $H_t$ connects at least $(1-\epsilon)\cdot f_i$ vertices from $g_i$ to $r$.
    
    Next, we verify the cost of our solution. Let $\OPT_t'$ be the cost of the optimal solution to $I_t'$. Notice that since our copy number is $O(\log n)$, it follows that $n_i \leq O(\log n \cdot |g_i|)$. Thus, by the guarantees of \Cref{thm:partialGSTOnTrees} we have
    \begin{align}\label{eq:optPBnd}
         w_T(H_t') \leq \frac{1}{\epsilon} \cdot \left(\max_i \frac{n_i}{f_i} \right) {\OPT}_t' \leq O\left(\frac{\log n}{\epsilon}\right) \cdot \left(\max_i \frac{|g_i|}{f_i} \right) {\OPT}_t'.
    \end{align}

    Next, we bound $\OPT_t'$. Let $H^*_t$ be the optimal solution to $I_t$. We claim that $\pi_{G \to T}(H^*_t)$ is feasible for $I_t'$. This follows because $H^*_t$ connects at least $f_i$ vertices from $g_i$ to $r$ for $i \leq t$ and so by the connectivity preservation property of copy tree embeddings we know that there are at least $f_i$ groups in $\mcG_i$ with a vertex connected to $r$ by $\pi_{G \to T}(H_t^*)$. Thus, combining this with the $O(\log^ 2 n)$ cost preservation of our copy tree embedding we have
    \begin{align}\label{eq:optPSBound}
        {\OPT}_t' \leq w_T(\pi_{G \to T}(H^*_t)) \leq O(\log ^ 2 n) \cdot w_G(H^*_t).
    \end{align}

    Lastly, by the cost preservation property of our copy tree embedding we have that $w_G(H_t) \leq w_T(H_t')$ which when combined with Equations \ref{eq:optPBnd} and \ref{eq:optPSBound} gives
    \begin{align*}
        w_G(H_t) \leq O\left(\frac{\log ^ 3 n}{\epsilon} \cdot \max_i \frac{|g_i|}{f_i} \right) \cdot w_G(H_t^*).
     \end{align*}
    thereby showing that our solution is within the required cost bound.
\end{proof}

As a consequence of the above result we have a poly-log bicriteria deterministic approximation algorithm for online partial group Steiner tree; we restate the relevant theorem below.

\partGST*

Since group Steiner tree is exactly $f$-partial group Steiner tree where $f_i = 1$ in which case $\max_i \frac{|g_i|}{f_i} \leq N$ where again $N$ is the maximum size of a group. Moreover, since any solution can only connect an integral number of vertices from each group, it follows that a $\frac{1}{2}$-connection-competitive solution for partial group Steiner tree where $f_i = 1$ (i.e.\ for group Steiner tree) connects at least one vertex from each group. Thus, as a corollary of the above result we have the following deterministic algorithm for online group Steiner tree.\footnote{We note that one can use an aforementioned property of our first construction---that if $u$ is connected to $r$ by $F \subseteq E$ then every vertex in $\phi(u)$ is connected to $\phi(r)$ in $\pi_{G \to T}(F)$---to reduce the $O(\log ^ 3 n)$s in this section to $O(\log ^ 2 n)$s. In particular, if one were to use this property then when we map the solution to our $f$-partial group Steiner tree problem on $G$ to our copy tree embedding, the resulting solution will connect at least $f_i$ groups in $\mcG_i$ at least $\Theta (\log n)$ times. It follows that when we run our water filling algorithm each time it increases $\sum_e x_e$ by $1$ we know that it cover at least $\Omega(\log n)$ units of the optimal solution by weight rather than $1$ unit of the optimal solution as in the current analysis.} 

\begin{corollary}
    There is an $O(N \log ^ 3 n)$-competitive deterministic algorithm for online group Steiner tree where $N := \max_i |g_i|$ is the maximum group size.
\end{corollary}

\section{Demand-Robust Group Steiner Tree/Forest}\label{sec:DRGSTF}

In this section, we give a poly-log-approximate algorithm for the demand-robust versions of the group Steiner tree and group Steiner forest problems. The high-level approach will be to find a black-box reduction from the problem on a general graph to a problem on a tree, and then to solve the tree problem. However, the properties that the copy tree embedding need to ensure in this setting are slightly  different, hence we will define and introduce a new, demand-robust copy tree embedding, in \Cref{def:demand-robust-copy-tree}.

On a general note, the demand-robust setting provides a robust counterpart to classic optimization problems like (group) Steiner tree, minimum cut, shortest path, etc. In this setting, instead of a single input, one is given a set of scenarios $\calS = \{ S_1, \ldots, S_m \}$, where each scenario $S_i$ corresponds to a classic input to the problem. The goal is to ``prepare'' for the worst-case scenario in $\calS$ by buying a ``first-stage solution'' $X_0$ at a discount before one knows which scenario is realized. After committing to $X_0$, the realized scenario $S_i$ is revealed and one needs to extend $X_0$ with a ``second-stage solution'' $X_i$ (where the cost of $X_i$ is inflated by a factor $\sigma_i \ge 1$) such that $X_0 \cup X_i$ satisfies scenario $S_i$. We want to minimize the total cost (of both the first-stage and the second-stage solution) in case of a realization of the worst-case scenario.

We first give formal descriptions of the demand-robust group Steiner tree and group Steiner forest problems. Note that the formal descriptions of the offline versions were given in \Cref{sec:det-online-group-steiner-tree} and \Cref{sec:det-online-group-steiner-forest}, respectively.

\textbf{Demand-robust versions of the group Steiner tree/forest problem:} Let $G = (V, E, w)$ be a weighted graph with a distinguished node $r \in V$ called the root where the weight $w(e)$ is the ``first-stage cost'' of an edge $e$. We are given a set of scenarios $\calS = \{ S_1, \ldots, S_m \}$ with $m := |\calS| \le \poly(n)$ where:
\begin{enumerate}
\item In the group Steiner tree problem, a scenario $S_i$ consists of a set of groups $g_{i,1}, g_{i,2}, \ldots, g_{i, k(i)}$, with $g_{i, j} \subseteq V$, and an inflation factor $\sigma_i \ge 1$. We assume $k(i) \le \poly(n)$.
\item In the group Steiner forest problem, a scenario $S_i$ consists of a set of pairs $(A_{i,1}, B_{i, 1}), (A_{i,1}, B_{i, 1}), \ldots, \allowbreak (A_{i, k(i)}, B_{i, k(i)})$, with $A_{i, j}, B_{i, j} \subseteq V$, and an inflation factor $\sigma_i \ge 1$. We assume $k(i) \le \poly(n)$.
\end{enumerate}
We wish to buy the (optimal) set of first-stage edges $X_0 \subseteq E$ in order to minimize the cost of the worst-case scenario being realized. The cost of scenario $S_i$ being realized is the smallest value $w(X_0) + \sigma_i \cdot w(X_i)$ over all set of edges $X_i \subseteq E$ such that $X_0 \cup X_i$ is a valid solution to the offline version of the problem for scenario $i$ (e.g., in the group Steiner tree problem, $X_0 \cup X_i$ connects at least one node $v \in g_{i, j}$ to the root $r$ for each group $g_{i, j}$in scenario $i$):

An alternative way to define the demand-robust version of the above problems is to say that we want to find subsets $X_0, X_1, \ldots, X_m$ which minimize $\max_{i=1}^m w(X_0) + \sigma_i \cdot w(X_i)$ such that $\forall 1 \le i \le m, X_0 \cup X_i$ satisfies scenario $S_i$ for the offline version. Let $\OPT := \max_{i=1}^m w(X_0) + \sigma_i \cdot w(X_i)$ be the cost of the optimal solution.

\subsection{Demand-Robust Copy Tree Embeddings}

We now introduce the demand-robust copy tree embedding and prove its existence. One notable difference between this embedding (which is appropriate for the demand-robust setting) and the copy tree embedding of \Cref{dfn:repTree} is that the forward- and backward-mapping function map tuples of subgraphs to tuples of subgraphs (of equal length). This is because the first- and second-stage solutions must be mapped in a coordinated fashion, a requirement that was not necessary in the previous settings.
\begin{definition}\label{def:demand-robust-copy-tree}  
  Let $G = (V, E, w)$ be a weighted graph with some distinguished root $r \in V$. An $\alpha$-approximate demand-robust copy tree embedding $\calC = (T, \phi, \pi_{G \to T}, \pi_{T \to G})$ consists of a weighted rooted tree $T = (V', E', w')$ with root $r'$, a copy mapping $\phi : V \to 2^{V'}$ with $\phi(r) = \{r'\}$, and edge mapping functions $\pi_{G \to G}$ and $\pi_{T \to G}$ that maps tuples of subgraphs (of any length $m$) to equal-length tuples of subgraphs.
  
  The ``forward-mapping function'' $\pi_{G \to T}$ maps at most $m \le \poly(n)$ subgraphs (more precisely, subsets of $E$), namely $X_0, X_1, \ldots, X_m$, to subsets of $E'$, namely $X'_0, X'_1, \ldots, X'_m$ such that the following always holds:
  \begin{enumerate}
  \item \textbf{Demand-robust Connectivity Preservation}: For all $1 \le i \le m$, and all $u, v \in V$ that are connected via $X_0 \cup X_i$, we have that $\phi(u)$ and $\phi(v)$ are connected via $X'_0 \cup X'_i$.
    
  \item \textbf{Cost Preservation}: For every $1 \le i \le m$ we have that $w'(X'_i) \le \alpha \cdot w(X_i)$.
  \end{enumerate}

  The ``backward-mapping function'' $\pi_{T \to G}$ maps $m \le \poly(n)$ subsets of $E'$, namely $X'_0, X'_1, \ldots, X'_m$, to subsets of $E$, namely $X_0, X_1, \ldots, X_m$ such that the following always holds:
  \begin{enumerate}
  \item \textbf{Demand-Robust Connectivity Preservation}: For all $1 \le i \le m$, and all $u, v \in V'$ that are connected via $X'_0 \cup X'_i$, we have that $\phi^{-1}(u)$ and $\phi^{-1}(v)$ are connected via $X_0 \cup X_i$.

  \item \textbf{Cost Preservation}: For every $1 \le i \le m$ we have that $w(X_i) \le w'(X'_i)$.
  \end{enumerate}
  
  A copy tree embedding is efficient if $T$, $\phi$, and $\pi_{T \to G}$ are all poly-time computable, and well-separated if $T$ is well-separated.
\end{definition}

Comparing the above with \Cref{dfn:repTree}, we note that an $\alpha$-approximate demand-robust copy tree embedding is also an $\alpha$-approximate copy tree embedding. However, the converse might not hold---for example, the ``merging FRT support construction'' as defined in \Cref{sec:FRTSup} (in particular, where the mapping function $\pi_{G \to T}$ simply embeds a subgraph into the cheapest tree) is not a $\log^{O(1)} n$-approximate demand-robust copy tree embedding. However, changing the forward mapping function of the FRT support construction, we are able to obtain the following guarantees.

\begin{theorem}\label{thm:demand-robust-copy-tree}
      There is a poly-time deterministic algorithm which given any weighted graph $G = (V, E, w)$ and root $r \in V$ computes an efficient and well-separated $O(\log^2n)$-approximate demand-robust copy tree embedding.
\end{theorem}
\begin{proof}
We show that the ``merging FRT support construction'' (same as \Cref{sec:FRTSup}, which we reintroduce here for convenience) also suffices for the demand-robust setting. We let $T_1, T_2, \ldots, T_{q}$ be the trees in the support of the FRT distribution guaranteed by \Cref{thm:charikFRTSup}. Then, we let $T$ be the result of identifying each copy of $r$ as the same vertex in each $T_i$ (but not identifying copies of other vertices in $V$ as the same vertex). $T$'s weight function $w'_T$ is inherited from each $T_i$ in the natural way. Similarly, we let $\phi(v)$ be the set containing each copy of $v$ in each of the $T_i$. It is easy to verify that $\phi$ is indeed a copy mapping. Also, note that $\phi(v)$ is computable in deterministic poly-time.
    
We now describe $\pi_{G \to T}$. Let $X_0, X_1, \ldots, X_m \subseteq E$ be a tuple of subgraphs of $E$. We use the probabilistic method to show there exists a tuple of subsets $X'_0, X'_1, \ldots, X'_m \subseteq E' := E(T)$ which satisfy the above properties. Note that the overall construction will still be deterministic as we only need to show the existence of $\pi_{G \to T}$ (e.g., we are not to be able to efficiently compute $\pi_{G\to T}$).

Independently sample $k := O(\log m) = O(\log n)$ random FRT trees, namely, $T'_1, \ldots, T'_{k}$ and let $w_{T_i}$ be their corresponding weights. In each $T'_i$ let $T'_i(X_0)$ be the unique forest (subgraph of $T'_i$) which has the same connected components as $X_0$. Finally, we set $X'_0 := \bigcup_{i=1}^k T'_i(X_0)$. Due to the properties of FRT, we have that $\E[w_{T'_i}( X_0 )] \le O(\log n) \cdot w_G(X_0)$, hence $w'_T(X'_0) \le k \cdot O(\log n) \cdot w_G(X_0) = O(\log^2 n) \cdot w_G(X_0)$ with at least constant probability.

We now fix a subset $X_i$. For each $j \in [k]$ we have that $w_{T'_i}( X_i ) \le O(\log n) \cdot w_G(X_i)$ with at least constant probability, hence with probability at least $1 - \exp(-O(k)) = 1 - n^{-O(1)}$ there exists some $j(i) \in [k]$ where the property holds. Assuming this is the case, we set $X'_i := T'_{j(i)}(X_i)$. Applying a union bound over all subgraphs $X'_i$ for $i \in \{0, 1, \ldots, m\}$, we conclude all of the above properties are satisfied with at least constant probability, hence via the probabilistic method at least one such mapping exists. By construction, the forward mapping satisfies the cost preservation properties with $\alpha = O(\log^2 n)$. Furthermore, if two nodes $u, v \in V(G)$ are connected in $X_0 \cup X_i$, then they are connected in $T'_{j(i)}(X_0) \cup T'_{j(i)}(X_i) \supseteq X'_0 \cup X'_i$---consider an edge either in $e \in X_0$ or in $e \in X_i$, in the former case the endpoints of the edge are connected in $T'_{j(i)}(X_0)$ and in the latter they are connected in $T'_{j(i)}(X_i)$.

Lastly, we specify $\pi_{T \to G}$. While the original definition acts on a tuple $( X'_i )_{i=1}^m$ of subsets of $E'$, we specify its action on a single subset $\pi_{T \to G}(F')$ and then this function to all elements of the tuple, i.e., $X_i := \pi_{T \to G}(X'_i)$ for all $i$. We let $\pi_{T \to G}(F')$ be $\bigcup_{(u', v') \in F'} P_{uv}$ where $P_{uv}$ is an arbitrary shortest path in $G$ between $u$ and $v$ and $u'$ and $v'$ are copies of $u$ and $v$. We first verify the cost preservation: for every $F' \subseteq E(G)$ we have $w_G(\pi_{T \to G}(F')) \le \sum_{(u', v') \in F'} w_G(P_{uv}) \le \sum_{(u', v') \in F'} w'_T(u', v') = w'(F')$, where the last inequality holds because distances in FRT trees dominate distances in $G$. This proves the cost preservation.

Next, we verify the demand-robust connectivity preservation: for each edge $X'_0 \cup X'_i$, its endpoints are connected either via $X_0 = \pi_{T\to G}(X'_0)$ (if $e \in X'_0$), or via $X_i = \pi_{T\to G}(X'_i)$ (if $e \in X'_i$), hence if two nodes are connected via $X'_0 \cup X'_i$, then they are connected via $X_0 \cup X_i$. It is easy to check that $T, \phi$, and $\pi_{G \to T}$ can all be constructed in deterministic poly-time.
\end{proof}

We also remark that the construction of merging partial tree embeddings can also be made into a demand-robust embedding of a smaller size. However, this approach seems more complicated and yields the same cost approximation, hence we do not present it here.


\subsection{Reducing from General Graphs to Trees}

In this section we show how to map the demand-robust group Steiner tree and forest problems on a general graph to an equivalent problem on a demand-robust copy tree embedding with a poly-log loss in the approximation factor. We formally describe the mapping and then proceed to prove its properties.

\textbf{Mapping to a copy tree embedding}. We describe how to map an instance $I = (G, r, \calS)$ of the demand-robust group Steiner tree/forest to a copy tree embedding $\calC = (T, \phi, \pi_{G \to T}, \pi_{T \to G})$. We define an instance $I' = (G', r', \calS')$ where $G' := T$ with $r'$ being the root of $T$. We set $\calS' \gets \calS$ with the following changed applied:
\begin{enumerate}
\item In the group Steiner tree problem, each group $g \in S_i \in \calS$ is changed to $g' := \bigcup_{v \in g} \phi(v)$. In other words, each node $v$ in a group is replaced by all of its copies $\phi(v)$ in the copy tree embedding.
\item In the group Steiner forest problem, each pair $(A, B) \in S_i \in \calS$ is changed to $(\bigcup_{a \in A} \phi(a), \bigcup_{b \in B} \phi(b))$.
\end{enumerate}

Note that the demand-robust group Steiner tree/forest instance maps to another instance of the same problem (e.g., a group Steiner tree problem maps to a group Steiner tree problem).

We remind the reader that the group Steiner forest problem directly generalizes the group Steiner tree problem---given a group Steiner tree problem on $g$ with groups $(g_i)$ we can reduce it to an equivalent group Steiner forest problem on the same graph $G$ and root $r$, where each group $g$ is mapped to the pair $(\{r\}, g)$.

Comparing the mapping to the copy-tree-embedding with the above reduction, a natural question arises whether one should apply the reduction before or after applying the mapping to the copy tree embedding. However, one can easily check that there is no difference---these two transformations ``commute''.

The following lemma illustrates why such a mapping definition is appropriate and it shows the utility of \Cref{def:demand-robust-copy-tree}.
\begin{lemma}\label{lemma:map-group-steiner-to-trees}
  Suppose that an instance $I$ of the demand-robust group Steiner tree (resp., forest) problem maps to a demand-robust group Steiner tree (resp., forest) instance $I'$ via a $\alpha$-approximate demand-robust copy tree embedding $\calC$. Then:
  \begin{enumerate}
  \item If $X_0, X_1, \ldots, X_{m}$ ($X_i \subseteq E(G)$) is a feasible solution for $I$ of cost $\OPT$, then $( X'_i )_{i=0}^{m} := \pi_{G \to T}\left( (X_i)_{i=0}^{m} \right)$ is a feasible solution to $I'$ with cost at most $\alpha \cdot \OPT$. \label{subclaim:graph-to-tree}
  \item If $X'_0, X'_1, \ldots, X'_{m}$ ($X'_i \subseteq E(T)$) is a feasible solution for $I'$ of cost $\ALG$, then $( X_i )_{i=0}^{m} := \pi_{T \to G}\left( (X_i')_{i=0}^{m} \right)$ is a feasible solution to $I$ with cost at most $\ALG$. \label{subclaim:tree-to-graph}
  \end{enumerate}
\end{lemma}
\begin{proof}
  We first prove (\ref{subclaim:graph-to-tree}). It is sufficient to prove the result for the forest problem---take the tree instance on $G$ with a feasible solution $X$ of cost $\OPT$, reduce it to an equivalent forest instance, map it to $\calC$ and, applying the forest claim, conclude there is a feasible solution $X'$ of value at most $\alpha \cdot \OPT$. By commutativity, $X'$ is also a feasible solution for the reduction of the original tree instance to the mapping to $\calC$, hence is a feasible solution (of cost at most $\alpha \cdot \OPT$) for the mapping of the original problem to $\calC$, proving the claim.

  We now prove (\ref{subclaim:graph-to-tree}) for the forest problem. Fix a scenario $S_i \in \calS$. By feasibility, for each pair $(A, B) \in S_i$ in the original instance, there exists $a \in A$ and $b \in B$ which are connected via $X_0 \cup X_i$. Therefore, by the demand-robust connectivity preservation, there exits $a' \in \phi(a)$ and $b' \in \phi(b)$ that are connected via $X'_0 \cup X'_i$. In other words, the set of vertices $\bigcup_{a \in A} \phi(a)$ is connected to the set of vertices $\bigcup_{b \in B} \phi(v)$ via $X'_0 \cup X'_i$, hence the solution is feasible for $I'$.

  Finally, we analyze the cost. By the cost preservation property, we have that $w'(X'_i) \le \alpha \cdot w(X_i)$, hence the cost is:
  \begin{align*}
    w'_G(X'_0) + \max_{1 \le i \le m} \sigma_i \cdot w'_G(X'_i) \le \alpha \cdot \left( w'(X'_0) + \max_{1 \le i \le m} \sigma_i \cdot w'(X'_i) \right) \le \alpha \cdot \OPT .
  \end{align*}
  Next, we prove (\ref{subclaim:tree-to-graph}). It is sufficient to prove the result for the forest problem---take the tree problem on $G$, map it to $\calC$, then reduce to a forest problem and obtain a feasible solution $X'$ of cost $\ALG$. By commutativity and assuming the claim for the forest problem, $X$ is a feasible solution to the reduction of the original tree instance to a forest instance. Hence, $X$ is a feasible solution (of cost at most $\ALG$) to the original tree instance.

  We now prove (\ref{subclaim:tree-to-graph}) for the forest problem. Fix a scenario $S_i \in \calS$. By feasibility, for each pair $(A, B) \in S_i$ in the original instance, the set of vertices $\bigcup_{a \in A} \phi(a)$ is connected to the set of vertices $\bigcup_{b \in B} \phi(v))$. Therefore, there exits $a' \in \phi(a), a \in A$ and $b' \in \phi(b), b \in B$ such that $a', b'$ are connected via $X'_0 \cup X'_i$. By the demand-robust connectivity preservation, we have that $a = \phi^{-1}(a')$ and $b = \phi^{-1}(b')$ are connected via $X_0 \cup X_i$, hence the solution is feasible for $I$.

  Finally, we analyze the cost. By the cost preservation property, we have that $w(X_i) \le w'(X'_i)$, hence the cost is:
  \begin{align*}
    w_G(X_0) + \max_{1 \le i \le m} \sigma_i \cdot w_G(X_i) \le w'(X'_0) + \max_{1 \le i \le m} \sigma_i \cdot w'(X'_i) \le \ALG.
  \end{align*}
\end{proof}


\subsection{Demand-Robust Group Steiner Tree When $G$ is a Tree}
\label{sec:demand-robust-group-steiner-tree-on-a-tree}

In this section we give a poly-log-approximation algorithm for the demand-robust group Steiner tree problem when the underlying graph $G$ is a weighted and rooted tree. The main result of the section follows.  

\DRSTT*


We note that combining \Cref{thm:demand-robust-steiner-tree-algo} with the mapping of \Cref{lemma:map-group-steiner-to-trees} and the demand-robust copy tree embedding construction \Cref{thm:demand-robust-copy-tree} immediately yields a randomized $O(\log^4)$-competitive poly-time algorithm for the group Steiner tree on general graphs, namely \Cref{thm:demand-robust-steiner-tree-algo-on-general-graph}.

The rest of this section is dedicated to proving \Cref{thm:demand-robust-steiner-tree-algo}. The general outline of our proofs is as follows.
\begin{enumerate}
\item We prove an important structural property on the first-stage solution that allows us to conclude that the there exists a first-stage solution that is a rooted subtree of $G$ (i.e., it is connected and contains the root of $G$).
\item We write the linear program that fractionally relaxes the demand-robust group Steiner tree problem.
\item We show how to utilize the randomized rounding for the online group Steiner tree problem of \cite{alon2006general} to construct a demand-robust solution. We remark that a more naive attempt at utilizing the randomized rounding techniques on a general graph (i.e., without transfering the problem to a demand-robust copy tree embedding) would not yield a poly-logarithmic approximation ratio---we crucially use the fact that $G$ is a tree to make the randomized rounding work.
\end{enumerate}

First, we prove an important structural property on the first-stage solution, first proved in \cite{dhamdhere2005pay}: there exists a 2-approximate first-stage solution that is a union of minimal feasible solutions for a subset of scenarios. For the demand-robust group Steiner tree problem, we say that $M_i \subseteq E$ is a minimal feasible solution to the scenario $S_i$ if no proper subset $M'_i \supsetneq M_i$ is feasible for the scenario (i.e., there exists at least one group in $S_i$ that is not connected to the root via $M'_i$). 

\begin{lemma}[Adapted from \cite{dhamdhere2005pay}]\label{lemma:first-stage-structure}
    In the demand-robust group Steiner tree problem on the graph $G = (V, E)$, there exists a first-stage solution $X_0 \subseteq E$ which can be extended to a solution of (worst-case realization) cost $2 \cdot \OPT$ which has the following structure. There exists a subset $I \subseteq \{1,2,\ldots,m\}$ and a set $\{ M_i \}_{i \in I}$, where $M_i$ is some minimal feasible solution (i.e., no proper subset is feasible) to the scenario $S_i$, such that $X_0 = \bigcup_{i \in I} M_i$.
\end{lemma}
  The proof of this result is directly argued via the proof of Lemma 4.1 in Section 4.1 of \cite{dhamdhere2005pay}. However, our claim requires slightly weaker structural properties compared to \cite{dhamdhere2005pay}---it stipulates that the first-stage solution $X_0$ is a union of minimal feasible solutions instead of being the minimal solution for a particular instance. The proof remains unchanged: every time when the \textsc{if} condition in (2b) is true (as given in  \cite{dhamdhere2005pay}), we add $I \gets I \cup \{i\}$ and observe that $M_i := X^*_{0i} \cup X_i^*$ is a minimal feasible solution for scenario $i$. By construction, $X_0 = \bigcup_{i \in I} M_i$ and, as argued in the proof, the cost of $X_0$ is at most $2 \cdot OPT$.

\paragraph{Relaxation $\LP_{GST}$.} We now give the linear program for a tree $G = (V, E, w)$ with a root $r \in V$ that relaxes the original problem. We say that a vector $x \in \mathbb{R}^E$ is decreasing on root-leaf path if for every $e \in E$ not incident to the root $r$ and its parent edge $\mathrm{parent}(e)$ we have $x_{\mathrm{parent}(e)} \ge x_e$---this condition is required by the randomized online rounding technique and can be argued to be a valid constraint due to \Cref{lemma:first-stage-structure}. The LP jointly optimizes over the first-stage solution $\{ x_{0, e} \}_e$ and second-stage parts of the solution $\{ x_{i, e} \}_{i \in [m], e \in E}$ while ensuring (1) the first-stage solution is decreasing on root-leaf paths, and (2) that the maximum flow between the root and each group $g_{i, j}$ (in scenario $i$) is at least $1$ when using $x_0 + x_i$ as edge capacities. We formally write out the linear program $\LP_{GST}$.

\begin{figure}[H]
  \centering
  \begin{align*}
    \min & \quad z \\
    \text{such that} \\
    \forall i \in [m] & \quad \sum_{e \in E} w(e) \left[ x_{0, e} + \sigma_i \cdot x_{i, e} \right ] \le z \\
    \forall i \in [m], \forall j \in [k(i)] & \quad \maxflow(x_0 + x_i, \{r\}, g_{i, j}) \ge 1 \\
    \forall e \in E & \quad \text{if $e$ is not incident to $r$, then } x_{0, \mathrm{parent}(e)} \ge x_{0, e}  \\
    \forall i \in \{0\} \cup [m], \forall e \in E & \quad x_{i, e} \ge 0
  \end{align*}
  \caption{Linear program $\LP_{GST}$}
\end{figure}

In the linear program we introduced the notation $\maxflow(x, A, B)$ where $x \in \mathbb{R}_{\ge 0}^E$, $A \subseteq V, B \subseteq V$ which corresponds to the maximum flow between the set $A$ and set $B$ when the capacity of an edge $e$ are set to $x_e$. The maximum flow between two sets $A$, $B$ is defined as the flow between the super-source $a$ and super-sink $b$ when a new virtual node $a$ is connected to all nodes in $A$ with infinite capacity and analogously for $b$.
The condition that this maximum flow using capacities $x_0 + x_i$ is at least $1$ can be expressed as a linear program with a polynomial number of variables and constraints, hence $\LP_{GST}$ can be solved in poly-time. 

Let $z^*$ be the optimal cost of the linear program. We argue that the LP is a relaxation of the original problem (with a factor-$2$ loss), i.e., $z^* \le 2 \OPT$. Let $X_0^*$ be the first-stage solution that satisfies the stipulations of \Cref{lemma:first-stage-structure}, hence $w(X_0^*) \le 2 \OPT$. The solution $X_0^*$ is decreasing on root-leaf paths since each minimal feasible solution is decreasing on root-leaf paths, hence we can deduce the same about their union. The flow and positivity properties are trivially satisfied by any feasible integral solution. Therefore, $z^* \le w(X_0^*) \le 2 \OPT$.

\paragraph{Rounding the LP.} We use the online algorithm for the group Steiner tree problem on trees from \citet{alon2006general}. Intuitively, given a sequence of fractional solutions $y_1, y_2, \ldots$, where each $y_i \in [0,1]^E$ represents the extent to which the edges in $E$ are bough and satisfy some simple monotonicity properties, the algorithm maintains a sequence of non-decreasing integral solutions $F_1, F_2, \ldots$ where $F_i \subseteq E$ such that (1) the cost of the integral solution is competitive with the cost of the fractional solution, and (2) the integral solution satisfies the same set of constraints as the fractional solution. The result is formalized as follows.

\begin{lemma}[\cite{alon2006general}]
    \label{thm:online-group-steiner-on-trees}
    Let $G = (V, E, w)$ be a weighted tree with a distinguished root $r \in V$. There exists a polynomial-time randomized algorithm which accepts a sequence of vectors $y_0, y_1, \ldots, y_T \in [0, 1]^E$ where each $y_i$ is decreasing on root-leaf paths for $i \in \{0, \ldots, T\}$ and $y_i(e) \le y_{i+1}(e)$ for all $i \in \{0, \ldots, T-1\}, e \in E$. For each $i \in \{0, \ldots, T\}$, upon receiving the vector $y_i$, the algorithm outputs a set $F_i \subseteq E$ which includes the previous output (i.e., $F_{i-1} \subseteq F_i$ if $i > 1$) and (1) $\Pr[e \in F_i] = y_i$ for each $e \in E$, and (2) for each $i$ and every set $g \subseteq V$ if $\maxflow(y_i, \{r\}, g) \ge 1$, then $F_i$ connects some node of $g$ to the root with probability at least $\Omega(1 / \log n)$.
\end{lemma}
This algorithm is explicitly explained in Section 4.2 of \cite{alon2006general}. Property (1) is argued via Lemma 10 and Property (2) matches Lemma 12.
%

Using the online rounding scheme of \Cref{thm:online-group-steiner-on-trees}, we show how to round $\LP_{GST}$ to obtain an (integral) demand-robust solution.

\begin{lemma}\label{lemma:rounding-lp-sol}
  Consider a demand-robust group Steiner tree problem on a weighted rooted tree $G = (V, E, w)$. Given a feasible solution $x$ to $\LP_{GST}$ with objective value $z$, there exists a polynomial-time randomized algorithm that outputs $X_0 \subseteq E, \ldots, X_m \subseteq E$ such that $w(X_0) + \sigma_i \cdot w(X_i) \le O(\log^2 n) \cdot z$ for all $i \in [m]$, and each group $g_{i, j}$ is connected to the root via $X_0 \cup X_i$ with probability at least $1 - n^{- O(1) }$ (both $O$-constants can be jointly increased).
\end{lemma}
\begin{proof}
  We run $C \cdot \log^2 n$ ($C > 0$ is a sufficiently large constant) independent copies of the algorithm described in \Cref{thm:online-group-steiner-on-trees} and continuously output the union of the copies' output. We set $y_0 := x_0$ and note that $x_0$ is valid, since it is decreasing on root-leaf paths due to the constraint in $\LP_{GST}$. We output (the union of all the copies) as the first stage solution $X_0$. We remember the state of the algorithm copies and perform the following for each scenario $i \in [m]$ (reverting the state upon completion).

  Suppose now that some scenario $S_i \in \calS$ is realized. We set $y_1 := x^*_0 + x^*_i$, hence clearly $y_0 \le y_1$. Furthermore, we can assume without loss of generality that $y_1$ is decreasing on root-leaf paths since otherwise we can lower the value of any violating edge value $(y_1)_e$ without decreasing the maximum flow to any group $g \subseteq V$; clearly, the value will not fall below $(y_0)_e$. Therefore, we can feed $y_1$ to all the algorithms and recover (the union of multiple copies of the their output) $X_i$, which will be our second-stage solution.

  We argue that this solution $X_0, X_1, \ldots, X_m$ is feasible. We remark here that $X_0$ only depends on $y_0$, and $X_i \supseteq X_0$. Furthermore, the probability that a single copy does not satisfy a group is $1 - 1 / O(\log n) \le \exp(- 1 / O(\log n) )$. Therefore, we can conclude via the independence of our algorithm copies' randomness and a union bound that every group is satisfied with at least one copy of the algorithm with probability at least $1 - \poly(n) \cdot \exp(- 1 / O(\log n) \cdot C \log^2 n ) \ge 1 - n^{-C' }$ (where $C' = O(1)$ can be made arbitrary by increasing $C = O(1)$).

  Finally, we argue our cost bound. Let $z$ be the objective value of $x$ and let $(F_0, F_1, \ldots, F_m)$ be the output of a fixed copy of the algorithm. For each $i \in [m]$ we have:
  \begin{align*}
    \E[ w(F_0) + \sigma_i\cdot w(F_1) ] & = \sum_{e \in E} w(e)\left( \Pr[e \in F_0] + \sigma_i \cdot \Pr[e \in F_1] \right) \\
                                            & \le \sum_{e \in E} w(e)\left( x_{0, e} + \sigma_i \cdot x_{i, e} \right ) \le z .
  \end{align*}
  Therefore, we have $\E[ w(X_0) + \sigma_i\cdot w(X_i) ] \le C \cdot \log^2 n \cdot z = O(\log^2 n) \cdot z$, bounding the cost.
\end{proof}

We conclude with our proof of \Cref{thm:demand-robust-steiner-tree-algo}.
\begin{proofof}{\Cref{thm:demand-robust-steiner-tree-algo}}
  Let $x^*$ represent the optimal solution to $\LP_{GST}$. We apply \Cref{lemma:rounding-lp-sol} on $x^*$ (with $f_{i, j} := 1$ for all $i,j$), the described poly-time algorithm outputs a feasible (integral) solution $X_0, X_1, \ldots, X_m$ such that $X_0 \cup X_i$ connects each group $g_{i, j}$ to the root with probability at least $1 - n^{-O(1)}$. Since there are at most $m \le \poly(n)$ scenarios, and each scenario has at most $\poly(n)$ groups, we can conclude via a union bound that the solution is feasible with probability at least $1 - \poly(n) \cdot n^{-O(1)} \ge 1 - n^{-100}$.%
\end{proofof}

\subsection{Demand-Robust Group Steiner Forest When $G$ is a Tree}

In this section we give a poly-log-approximation algorithm for the demand-robust group Steiner forest problem when the underlying graph $G$ is a weighted and rooted tree. The main result of the section follows.

\DRSFT*

We note that combining \Cref{thm:demand-robust-steiner-forest-algo} with the mapping of \Cref{lemma:map-group-steiner-to-trees} and the demand-robust copy tree embedding construction \Cref{thm:demand-robust-copy-tree} immediately yields a randomized $O(\log^6)$-competitive poly-time algorithm for the group Steiner forest on general graphs when the aspect ratio is polynomial, namely \Cref{thm:demand-robust-steiner-forest-algo-on-general-graph}. Note that here we used the fact that for graphs with polynomial aspect ratio the depth of the FRT trees can be assumed to be $D = O(\log n)$. The rest of this section is dedicated to proving \Cref{thm:demand-robust-steiner-forest-algo}.

We proceed in a similar way to the demand-robust group Steiner tree on a tree: first write a linear programming relaxation and then utilize the online rounding scheme for the group Steiner forest problem (presented in \cite{naor2011online}) to obtain a demand-robust solution. Again, we remark that using the randomized rounding scheme in a more naive way (without going through the demand-robust copy tree embedding) does not immediately yield poly-logarithmic approximation ratios.

\paragraph{Relaxation $\LP_{GSF}$.} We write a somewhat more complicated linear programming relaxation than we did in the demand-robust group Steiner tree case. Remember that $G$ is a rooted tree. We make $D+1$ copies, $G_0, G_1, \ldots, G_D$ of the tree $G$. Next, the $\ell^{th}$ copy $G_\ell$ deletes all nodes whose depth is less than $\ell$ (e.g., for $\ell = 0$ we copy $G$ and for $\ell = D$ the graph is a set of isolated nodes). Note that $G_\ell$ is a forest; let $\calT_\ell$ the set of (maximal) trees in $G_\ell$. For each edge $e \in E(G_\ell)$ in a copy $G_{\ell}$ we introduce first-stage and second-stage variables $x_{\ell, i, e}$ for $\ell \in \{0, 1, \ldots, D\}$ and $i \in \{0, 1, \ldots, m\}$. Similarly as in the group Steiner tree case, we require that the first-stage solution is root-leaf decreasing in order for the online rounding scheme to work. Lastly, over $\ell, i, j$ (same range as before) and for $T \in \calT_\ell$ we introduce a ``flow variable'' $f_{\ell, T, i, j}$ which corresponds to the amount of flow that can be routed via $x_{\ell, 0} + x_{\ell, i}$ between the root of $T$ and the nodes in $A_{i, j}$ and $B_{i, j}$ (we want the same amount of flow to be routable to both of them). The linear program requires that the total amount of flow $f$ across all the trees in $\bigcup_{\ell=0}^D \calT_\ell$ is at least $1$.


\begin{figure}[H]
\centering
\begin{align*}
  \min & \quad z \\
  \text{such that} \\
  \forall i \in [m] & \quad \sum_{\ell=0}^D \sum_{e \in E(G_\ell)} w(e) \left[ x_{\ell, 0, e} + \sigma_i \cdot x_{\ell, i, e} \right ] \le z \\
  \forall \ell \in \{0, \ldots, D\}, i \in [m], \forall j \in [k(i)], \forall T \in \calT_\ell & \quad \maxflow(x_{\ell, 0} + x_{\ell, i}, \{\text{root of T}\}, A_{i, j}) \ge f_{\ell, T, i, j} \\
  \forall \ell \in \{0, \ldots, D\}, i \in [m], \forall j \in [k(i)], \forall T \in \calT_\ell & \quad \maxflow(x_{\ell, 0} + x_{\ell, i}, \{\text{root of T}\}, B_{i, j}) \ge f_{\ell, T, i, j} \\
  \forall i \in [m], \forall j \in [k(i)] & \quad \sum_{\ell=0}^D \sum_{T \in \calT_\ell} f_{\ell, T, i, j} \ge 1 \\
  \forall \ell \in \{0, \ldots, D\}, \forall e \in E(G_\ell) & \quad \text{if $e$ is not at the top of its tree in $G_\ell$, then } x_{\ell, 0, \mathrm{parent}(e)} \ge x_{\ell, 0, e}  \\
  \forall \ell \in \{0, \ldots, D\}, \forall i \in \{0\} \cup [m], \forall e \in E(G_\ell) & \quad x_{\ell, i, e} \ge 0 \\
  \forall \ell \in \{0, \ldots, D\}, \forall i \in [m], \forall e \in E(G_\ell) & \quad f_{\ell, i, e} \ge 0
\end{align*}
\caption{Linear program $\LP_{GSF}$}
\end{figure}

The condition that this maximum flow using capacities $x_{\ell, 0} + x_{\ell, i}$ is at least $f_{\ell, T, i, j}$ can be expressed as a linear program with a polynomial number of variables and constraints, hence $\LP_{GSF}$ can be solved in poly-time.

We now argue that $\LP_{GSF}$ relaxes the original problem (up to a factor of $O(D)$ loss). To this end we introduce some notation. Let $p$ be a simple path in $G$ and consider the highest (closest to the root) node $x \in V(p)$ it passes through. We say that $p$ \textbf{peaks at node $x$}. The high-level idea is that we can consider the optimal integral solution and, for each pair $(A_{i, j}, B_{i, j})$ observe the path that connects a node in $A_{i, j}$ with a node in $B_{i, j}$. If this path peaks at node $x$, we assign this pair to the tree in $\calT_{\mathrm{depth}(x)}$ whose root is exactly $x$. Then, by applying the structural \Cref{lemma:first-stage-structure} on each tree in $\bigcup_{\ell=0}^D \calT_\ell$, we can conclude that there is a root-leaf decreasing integral solution that solves the assigned pairs to the tree, hence the integral solution satisfies all the properties of $\LP_{GSF}$ and is therefore a relaxation.

\begin{lemma}\label{lemma:lp-gsf-relaxation}
  Let $z^*$ be the optimal objective value of $\LP_{GSF}$ with respect to some demand-robust group Steiner forest problem with optimal value $\OPT$ on an underlying tree with depth $D$. Then $z^* \le O(D) \cdot \OPT$.
\end{lemma}
\begin{proof}
  Let $X_0^*, X^*_1, \ldots, X^*_m$ be the optimal first-stage and second-stage solutions (as defined on $G$). We define $X^*_{\ell, T, i}$ for $\ell \in \{0, 1, \ldots, D\}, i \in \{0, 1, \ldots, m\}, T \in \calT_\ell$ as a natural extension of $X_i^*$ to $T$: if $e' \in E(T)$ is copied from $e \in E(G)$, then $e' \in X^*_{\ell, T, i} \iff e \in X^*_i$. Therefore, since each edge is copied $D+1$ times, for all $i \in [m]$ we have that $\sum_{\ell=0}^D \sum_{T \in \calT_\ell} w(X^*_{\ell, 0, T}) + \sigma_i \cdot w(X^*_{\ell, T, i}) \le (D+1) \cdot \OPT$.

  Let $p$ be the path connecting (some node in) $A_{i,j}$ to (some node in) $B_{i,j}$. Suppose that $p$ peaks at node $x$, let $\ell$ be the depth (in $G$) of $x$, and let $T$ be the maximal tree in $G_{\ell}$ whose root is at $x$. Since in the optimal solution both $A_{i,j}$ and $B_{i,j}$ are connected to the root, we \textbf{assign} the ``groups'' $A_{i,j}$ and $B_{i,j}$ to $T$ (both $A_{i,j}$ and $B_{i,j}$ are considered stand-alone groups, i.e., we forget that they were paired beforehand). Clearly, since the optimal solution is feasible, each (element of a) pair is assigned to exactly one tree.

  Fix a particular (maximal) tree $T$ in $\bigcup_{\ell=0}^D G_{\ell}$ and consider the set $\calP_{T}$ of groups \textbf{assigned} to $T$. Grouping by the groups their originating scenario, we can rewrite $\calP_{T}$ as $\calP'_{T} := ( \calP_{T, i} )_{i=1}^m$ where $\calP_{T, i}$ is the set of groups from $\calP_{T}$ that originated from scenario $i$. Finally, we note that $( X^*_{\ell, T, i} )_{i=0}^m$ is a feasible solution to the demand-robust group Steiner \textbf{tree} problem with scenarios $\calP'_{\ell}$.

  Applying \Cref{lemma:first-stage-structure} on each such tree $T$ , there exists a (first-stage and second-stage) solution $( X'_{\ell, T, i} )_{i=0}^m$ such that for all $\ell, T, i$, we have (i) $w(X'_{\ell, T, i}) \le 2 \cdot w(X^*_{\ell, T, i})$, (ii) the first-stage solution $X'_{\ell, T, 0}$ is a subtree of $T$ with coinciding roots, (iii) $( X'_{\ell, T, i} )_i$ is a feasible solution to $\calP'_{\ell}$ (i.e., for each pair $(A_{i,j}, B_{i,j})$ assigned to $T$, $X'_{\ell, T, 0} \cup X'_{\ell, T, i}$ connects $A_{i,j}$ to the root of $T$ as well as $B_{i,j}$).

  We now define $x_{\ell, i, e} := 1$ if $e \in X'_{\ell, T, i}$ for the unique tree $T \in \calT_\ell$ such that $e \in E(T)$, and $0$ otherwise. Furthermore, if groups $A_{i, j}$ and $B_{i, j}$ are assigned to a tree $T \in \calT_\ell$, we can set $f_{\ell, T, i, j} := 1$ and $f_{\ell, T, i, j} := 0$ otherwise. We argue that $(x, f)$ is a feasible solution to the linear program $\LP_{GSF}$.
  
  Property (ii) of $X'$ ensures that the $x_{\ell, i}$ is decreasing on all root-leaf paths of each tree in $G_{\ell}$. Finally, from property (iii) we conclude that the maximum flow property being at least $f_{\ell, T, i, j}$ is also satisfied, hence proving that $(x, f)$ is a feasible solution. Therefore, the objective follows from condition (i); for all $i \in [m]$ we have that:
  \begin{align*}
    z^* & \le \sum_{\ell=0}^D \sum_{T \in \calT_{\ell}} \sum_{e \in E(T)} w(e) \left[ x_{\ell, 0, e} + \sigma_i \cdot x_{\ell, i, e} \right ] \\
        & = \sum_{\ell=0}^D \sum_{T \in \calT_{\ell}} w(X'_{\ell, T, 0}) + \sigma_i \cdot w(X'_{\ell, T, i}) \\
        & \le 2 \sum_{\ell=0}^D \sum_{T \in \calT_{\ell}} w(X^*_{\ell, T, 0}) + \sigma_i \cdot w(X^*_{\ell, T, i}) \\
        & \le O(D) \cdot \OPT.\qedhere
  \end{align*}
\end{proof}

We now present the randomized online rounding scheme from \cite{naor2011online} which enables us to round $\LP_{GSF}$ into a demand-robust solution.

\begin{lemma}[\cite{naor2011online}]
  \label{thm:online-pair-group-forest}
  Let $G = (V, E, w)$ be a forest, namely a collection of (maximal) rooted trees $G_1, G_2, \ldots, G_m$ with roots $r_1, \ldots, r_m$. There exists a polynomial-time randomized algorithm which accepts a sequence of vectors $y_0, y_1, \ldots, y_T \in [0, 1]^E$ where each $y_i$ is decreasing on root-leaf paths for $i \in \{0, \ldots, T\}$ and $y_i(e) \le y_{i+1}(e)$ for all $i \in \{0, \ldots, T-1\}, e \in E$. For each $i \in \{0, \ldots, T\}$, upon receiving the vector $y_i$, the algorithm outputs a set $F_i \subseteq E$ which includes the previous output (i.e., $F_{i-1} \subseteq F_i$ if $i > 1$) and such that (1) $\Pr[e \in F_i] = y_i$ for each $e \in E$, and (2) for each $i$ and each pair $(A, B)$ where $A \subseteq V, B \subseteq V$, if $\sum_{j=1}^m \min(\maxflow(y_i, r_j, A), \maxflow(y_i, r_j, B)) \ge 1$, then with probability $\Omega(1/\log^2 n)$ there is a root $r_j$ connects to both a node in $A$ and a node in $B$ via $F_i$.
\end{lemma}
  The algorithm is implicitly explained in Section 3 of \cite{naor2011online}. Their description talks about an online rounding algorithm for the group Steiner forest problem on a tree $G$. The algorithm accepts an increasing sequence of vectors $y_0, \ldots, y_T \in [0,1]^{E(G)}$ and proceeds by splitting $G'$ into a forest $\bigcup_{\ell=0}^D \calT_\ell$ and providing the guarantees specified in this claim. The guarantees are proven in Lemma 6 of the paper.

Finally, we combine the relaxation with the LP rounding to prove the main result of this section.

\begin{proofof}{\Cref{thm:demand-robust-steiner-forest-algo}}
  Let $(x^*, f^*)$ be the optimal LP solution of the demand-robust Steiner forest problem with respect to scenarios $\calS$ and let $z^* \le O(D) \cdot \OPT$ be the objective value (\Cref{lemma:lp-gsf-relaxation}).

  \textbf{Splitting $G$ into a forest $G'$.} Given a tree $G$, we construct a forest $G'$ as composed of $\bigsqcup_{\ell=0}^D \calT_\ell$ (i.e., each tree in $\calT_\ell$ will be included as a component in $G'$). Note that for each $i \in \{0, \ldots, m\}$ the input $x_i$ can be naturally understood as a real vector indexed over the set $E(G')$.

  Furthermore, an edge in $e \in E(G)$ \textbf{corresponds} to possibly multiple (but at most $O(D)$) edges in $E(G')$, whereas an edge $e'\in E(G')$ corresponds to a unique edge $e \in E(G)$. Therefore, we define a projection $\pi_{G' \to G} : 2^{E(G')} \to 2^{E(G)}$ which maps an edge $e' \in E(G')$ to its corresponding edge $e = \pi_{G' \to G}(\{e'\})$, and we extend this to subgraphs $F' \subseteq E(G')$ via $\pi_{G' \to G}(F') = \bigcup_{e' \in E(F')} \pi_{G' \to G}(\{e'\})$.
  
  \textbf{Constructing the solution.} We set $y_0 := x_0$ and apply \Cref{thm:online-pair-group-forest} on $G'$ to obtain the integral first-stage $X_0$. Note that $x_0$ is decreasing on root-leaf paths due to a constraint in $\LP_{GST}$.

  The second-stage solutions $(X_1, \ldots, X_m)$ are obtained by saving the state of the algorithm and performing the following for each scenario $i \in [m]$ (reverting the state upon completion). In case some scenario $S_i \in \calS$ is realized, we set $y_1 := x^*_0 + x^*_i$, hence clearly $y_0 \le y_1$. First, we note that $x_0$ is decreasing on root-leaf paths due to the constraint in $\LP_{GSF}$. Furthermore, we can assume without loss of generality that $y_1$ is decreasing on root-leaf paths since otherwise we can lower the value of any violating edge value $(y_1)_e$ without decreasing the maximum flow to any subset of $V$; clearly, the value will not fall below $(y_0)_e$. Therefore, $y_1$ is valid and can be fed to all the algorithms, recovering $X_i \subseteq E(G')$.

  \textbf{Analysis.} The cost analysis is straightforward: $\E[ w(X_0) + \sigma_i w(X_i) ] \le z^* \le O(D) \cdot \OPT$.
  
  By construction of $\LP_{GSF}$, for each pair $(A_{i, j}, B_{i, j})$ the fractional solution $x_0 + x_i \in \mathbb{R}_{\ge 0}^{E(G')}$ yields a flow of at least $1$ across all $\bigcup_\ell T_\ell$, or equivalently, $G'$. Therefore, by \Cref{thm:online-group-steiner-on-trees}, with probability $\Omega(1/\log^2 n)$ some node in $A_{i, j}$ and some node in $B_{i,j}$ will be connected to the same root of a tree in $G'$ via $X_0 \cup X_i \subseteq E(G')$. Furthermore, by construction of $\pi_{G' \to G}$, this implies that (with the same probability) $\pi_{G' \to G}(X_0 \cup X_i)$ are connected in $G$. We can run $O(\log^3 n)$ independent copies to recover the result with high probability (at least $1 - 1/n^{100}$) and have the cumulative cost be $O(D \cdot \log^3 n) \cdot \OPT$.
\end{proofof}

\section{Conclusion and Future Work}
Online and dynamic algorithms built on probabilistic tree embeddings seem inherently randomized and necessarily not robust to adaptive adversaries. In this work we gave an alternative to probabilistic tree embeddings---the copy tree embedding---which is better suited to deterministic and adaptive-adversary-robust algorithms. We illustrated this by giving several new results in online and demand-robust algorithms, including a reduction of deterministic online group Steiner tree and group Steiner forest to their tree cases, a bicriteria deterministic algorithm for online partial group Steiner tree and new algorithms for demand-robust Steiner forest, group Steiner tree and group Steiner forest.

As a conceptual contribution we believe that copy tree embeddings will prove to be useful far beyond the selected algorithmic problems covered in this paper. We conclude by providing just some directions for such future works.

As mentioned earlier, \citet{bienkowski2020nearly} recently gave a deterministic algorithm for online non-metric facility location---which is equivalent to online group Steiner tree on trees of depth $2$---with a poly-log-competitive ratio and stated that they expect their techniques will extend to online group Steiner tree on trees. A very exciting direction for future work would thus be to extend these techniques to general depth trees which, when combined with our reduction to the tree case, would prove the existence of a deterministic poly-log-competitive algorithm for online group Steiner tree, settling the open question of \citet{alon2006general}.

While our focus has been on two specific constructions, it would be interesting to prove lower bounds on copy tree embedding parameters, such as, more rigorously characterizing the tradeoffs between the number of copies and the cost approximation factor. One should also consider the possibility of improved constructions. For example: Is it possible to get a logarithmic approximation with few copies, maybe even a constant number of copies? It is easy to see that with an exponential number of copies---one for each possible subgraph---a perfect cost approximation factor of one is possible. Can one show that a sub-logarithmic distortion is impossible with a polynomial number of copies?  We currently do not even have a proof that excludes a constant cost approximation factor with a constant copy number.

Furthermore, while this paper focused on online group Steiner problems, there are many other online and dynamic algorithms where copy tree embeddings might be able to give deterministic and adaptive-adversary-robust solutions for general graphs. Several such works are: \citet{englert2007reordering} and \citet{englert2017reordering} give an algorithm for the reordering buffer problem; \citet{guo2020facility} recently gave a dynamic algorithm for facility location; \citet{gupta2019permutation} gives an algorithm for fully dynamic metric matching. All these works feature a deterministic algorithm which works against adaptive adversaries in trees but then use FRT to obtain a randomized algorithm for general graph, which unsurprisingly only works against oblivious adversaries. The work on the reordering buffer problem seems especially promising since the algorithm for trees is quite similar in spirit to our water-filling algorithm for partial group Steiner tree. We believe that the natural generalization of this water-filling algorithm to copy tree embeddings should work and generalize the deterministic algorithm from trees to general graphs. While there has been follow-up work on this problem which does not use FRT for this problem \cite{kohler2017reordering} this would still improve the known bounds for this problem for some parameter settings.  

Lastly, a recent work of \citet{barta2020online} gave online embeddings for network design with logarithmic approximation guarantees in the number of terminals rather than $n$. It would be exciting to marry these ideas with the ones presented here to get the best of both worlds: a deterministic online copy tree embedding with distortion as a function of the number of terminals.

\bibliographystyle{plainnat}
\bibliography{abb,main}

\appendix



\section{Deferred Proofs}\label{sec:defProof}
\onGSF*
\begin{proof}
    We will use our copy tree embedding to produce a single tree on which we must solve deterministic online group Steiner forest.
    
    In particular, consider an instance of online group Steiner forest on weighted weighted $G = (V, E, w)$. Then, we first compute a copy tree embedding $(T, \phi, \pi_{G \to T}, \pi_{T \to G})$ deterministically with respect to $G$ and an arbitrary root $r \in V$ as we assumed is possible by assumption. Next, given an instance $I_t$ of group Steiner forest on $G$ with pairs $(S_1, T_1), \ldots (S_t, T_t)$, we let $I_t'$ be the instance of group Steiner forest on $T$ with pairs $(\phi(S_1), \phi(T_1)), \ldots (\phi(S_t), \phi(T_t))$ where we have used the notation $\phi(W) := \bigcup_{v \in W} \phi(v)$ for $W \subseteq V$. Then if the adversary has required that we solve instance $I_t$ in time step $t$, then we require that our deterministic algorithm for online group Steiner forest on trees solves $I_t'$ in time step and we let $H_t'$ be the solution returned by our algorithm for $I_t'$. Lastly, we return as our solution for $I_t$ in time step $t$ the set $H_t := \pi_{T \to G}(H_t')$.
    
    Let us verify that the resulting algorithm is indeed feasible and of the appropriate cost. 
    
    First, we have that $H_t \subseteq H_{t+1}$ for every $t$ since $H_t' \subseteq H_{t+1}'$ because our algorithm for trees returns a feasible solution for its online problem and $\pi_{T \to G}$ is monotone by definition of a copy tree embedding. Moreover, we claim that $H_t$ connects at least one vertex from $S_i$ to at least one vertex from $T_i$ for $i \leq t$ and every $t$. To see this, notice that $H_t'$ connects at least one vertex from $\phi(S_i)$ to some vertex in $\phi(T_i)$ since it is a feasible solution for $I_t'$ and so at least one copy of a vertex in $\phi(S_i)$ is connected to at least one copy of a vertex in $\phi(T_i)$; by the connectivity preservation properties of a copy tree it follows that at least one vertex from $S_i$ is connected to at least one vertex from $T_i$. Thus, our solution is indeed feasible in each time step.
    
    Next, we verify the cost of our solution. Let $\OPT_t'$ be the cost of the optimal solution to $I_t'$, let $n'$ be the number of vertices in $T$ and let $N'$ be the maximum size of a set in a pair in $I_t'$ for any $t$. By our assumption on the cost of the algorithm we run on $T$ and since $n' \leq \chi n$ and $N' \leq \chi N $ by definition of copy number, we know that 
    \begin{align*}
        w_T(H_t') \leq {\OPT}_t' \cdot f(n', N', k) = {\OPT}_t' \cdot f(\chi n, \chi N, k).
    \end{align*}
    
    Next, let $H^*_t$ be the optimal solution to $I_t$. We claim that $\pi_{G \to T}(H^*_t)$ is feasible for $I_t'$. This follows because $H^*_t$ connects a vertex from $S_i$ to $T_i$ for every $i \leq t$ and so by the connectivity preservation property of copy tree embeddings we know that some vertex from $\phi(S_i)$ is connected to some vertex of $\phi(T_i)$ for every $i \leq t$ in $\pi_{G \to T}(H^*_t)$. Applying this feasibility of $\pi_{G \to T}(H_t^*)$ and the cost preservation property of our copy tree embedding, it follows that $\OPT_t' \leq w_T(\pi_{G \to T}(H_t^*)) \leq \alpha \cdot w_G(H_t^*) = \alpha \cdot \OPT_t$.
    
    Similarly, we know by the cost preservation property of our copy tree embedding that $w_G(\pi_{T \to G}(H_t')) \leq w_T(H_t')$. Combining these observations we have
    \begin{align*}
        w_G(\pi_{T \to G}(H_t')) \leq w_T(H_t') \leq {\OPT}_t' \cdot f(\chi n, \chi N, k) \leq {\OPT}_t \cdot \alpha \cdot f(\chi n, \chi N, k),
    \end{align*}
    thereby showing that our solution is within the required cost bound.
\end{proof}

\end{document}